\documentclass[12pt,english]{article}

\usepackage[T1]{fontenc}
\usepackage[utf8]{inputenc}
\usepackage[dvipsnames]{xcolor}
\usepackage{babel}
\usepackage{amsthm}
\usepackage{amssymb}
\usepackage{mathtools,bbm,dsfont}
\usepackage{mathrsfs}
\usepackage{graphicx}
\usepackage{float}
\PassOptionsToPackage{normalem}{ulem}
\usepackage{etoolbox}
\usepackage{subcaption}
\usepackage{float}
\usepackage{soul}
\usepackage{cancel}
\usepackage{wasysym}
\usepackage{ulem}

\allowdisplaybreaks[1]

\usepackage{array,longtable}

\newtheorem{theorem}{Theorem}[section]

\newtheorem{fact}[theorem]{Fact}

\DeclareMathOperator{\Diff}{Diff}
\DeclareMathOperator{\Symp}{Symp}

\DeclareMathOperator{\sinc}{sinc}

\newcommand{\sgn}[1]{\text{sgn}(#1)} 
\newcommand{\tr}{\text{tr}} 
\newcommand{\M}{\mathcal{M}}

\newcommand{\R}{\mathbb{R}}
\newcommand{\Z}{\mathbb{Z}}
\newcommand{\N}{\mathbb{N}}

\makeatletter

\usepackage{jheppub}
\def\@fpheader{\relax}
\allowdisplaybreaks[1]

\makeatother

\begin{document}

\title{Spherically Symmetric Gravity on a Graph I:\\ Theoretical Foundations}


\author[a]{Klaus Liegener,}
\author[b,c,d]{Saeed Rastgoo,}
\author[b]{Jorden Roberts,}
\affiliation[a]{Walther-Mei\ss ner-Institut, Bayerische Akademie der Wissenschaften, 85748 Garching, Germany} 
\affiliation[b]{Department of Physics, University of Alberta, Edmonton, Alberta T6G 2G1, Canada}
\affiliation[c]{Department of Mathematical and Statistical Sciences, University of Alberta, Edmonton, Alberta T6G 2G1, Canada}
\affiliation[d]{Theoretical Physics Institute, University of Alberta, Edmonton, Alberta T6G 2G1, Canada}
\emailAdd{klaus.liegener@wmi.badw.de}
\emailAdd{srastgoo@ualberta.ca}
\emailAdd{jorden3@ualberta.ca}










\abstract{This manuscript is the first in a series of instalments that investigate spherically symmetric solutions within the effective dynamics program of Loop Quantum Gravity. The choice of lattice is adapted such that it remains invariant under a set of symmetry transformations maximally mapping spherical symmetry to the discrete setting. The conditions for symmetry restriction of the dynamics are investigated and a subspace is identified to make computations feasible. Afterwards symplectic structure and scalar constraint are explicitly computed on this subspace. This lays the groundwork to target several particular solutions, such $k=1$ cosmology and black holes, which will serve as the subjects of forthcoming follow-up papers.}







\maketitle

\section*{Notation \& Conventions}
\begin{longtable}{>{\raggedright\arraybackslash}p{0.25\textwidth} >{\raggedright\arraybackslash}p{0.65\textwidth}}
    \hline
    \textbf{Symbol} & \textbf{Meaning} \\
    \hline
    \(\mu,\nu,\ldots=0,1,2,3\) & Spacetime indices. \\
    \(a,b,\ldots=1,2,3\) & Spatial indices. \\
    \(I,J,\ldots=1,2,3\) & Internal (Lie algebra) indices. \\
    \(i,j,\ldots=\pm 1,\pm 2,\pm 3\) & Graph indices. \\ 
        
    \\
    \(G\) & Newton's gravitational constant. \\
    \(\kappa = 16\pi G\) & Gravitational coupling constant. \\
    \(\beta\in\R\backslash\{0\}\) & Immirzi parameter. \\

    \\
    \(S\) & Globally-hyperbolic spacetime manifold. \\
    \(\Sigma\) & Spatial manifold embedded into $S$. \\
    \(T\Sigma, T^*\Sigma\) & Tangent and cotangent bundles over $\Sigma$ \\
    \(T_x\Sigma, T_x^*\Sigma\) & Tangent and cotangent spaces of $\Sigma$ at $x$, i.e., fibers over $x\in\Sigma$ in $T\Sigma,T^*\Sigma$ \\
    \(g, g_{\mu\nu}\) & Spacetime metric and its components in local coordinates $(x^\mu)$. \\
    \(q, q_{ab}\) & Spatial metric on $\Sigma$ and its components in local coordinates $(x^a)$. \\

    \\
    \(A\cong B\) & $A$ and $B$ are isomorphic, i.e., there exists a bijection $f:A\to B$ such that $f,f^{-1}$ preserve the defining structure (e.g., group operations, brackets) \\
    \(\subset,\,<\) & Subset, subgroup. \\ 
    \(G_1\rtimes G_2\) & Semidirect product of groups $G_1,G_2$, with $G_2$ admitting a non-trivial action on $G_1$. \\
    \(\mathrm{Aut}(\cdot)\) & Group of automorphisms, i.e., isomorphisms from an algebraic structure to itself \\
    \(\Diff(\cdot)\) & Group of smooth diffeomorphisms. \\
    \(\Symp(\cdot)\) & Group of (Hamiltonian) symplectomorphisms. \\ 
    \(\mathrm{Ad}_U\) & Conjugation action with respect to a group element $U$. \\ 
    \([\cdot,\cdot]\) & Lie bracket/commutator. \\
    \(\mathcal{L}\) & Lie derivative. \\
    \(\iota_X\) & Contraction/interior product with respect to a vector field $X$.
        
    \\
    \(\M\) & An arbitrary even-dimensional differentiable manifold.\\
    \(\mathcal{A}^R, \mathcal{A}^L\) & Right and left action of a group $G$ on a manifold $\M$, i.e., smooth maps $\M\times G \to \M:(m,g)\mapsto \mathcal{A}^R_g(m)$ and $G\times\M\to\M:(g,m)\mapsto \mathcal{A}^L_g(m)$. \\
    \(\omega\) & Closed non-degenerate 2-form over a differentiable manifold $\M$, i.e., a symplectic form. \\
    \((\M,\omega)\) & General symplectic manifold with underlying phase space $\M$ and symplectic form $\omega$. \\
    \((\overline{\M},\overline{\omega})\) & Symplectic submanifold of $(\M,\omega)$ obtained via symmetry restriction with respect to a group $\Phi < \Symp(\M,\omega)$. \\
    \(\mathfrak{X}(\M)\) & Set of all vector fields over $\M$, i.e., sections of the tangent bundle $T\M$. \\ 
    \(X_f\) & Hamiltonian vector field associated with the phase space function $f\in C^\infty(\M)$. \\
    \(\{f,g\}\) & Poisson bracket between functions $f,g\in C^\infty(\M)$. \\
        
    \hline
    \caption{Definitions of various mathematical symbols and notational conventions employed throughout this work.}
    \label{tab:notation}
\end{longtable}

\section{Introduction}\label{s1:intro}
The search for a consistent theory of quantum gravity remains a central open problem in modern theoretical physics \cite{Unruh:76,Jafferis:2022,IceCube:22}. Among the various proposals, Loop Quantum Gravity (LQG) \cite{Thiemann_MCQGR,AL04,GP00,Rov04} has stood out as a promising candidate, though its dynamical sector continues to suffer from significant quantization ambiguities \cite{Perez_2006, PhysRevD.98.106026, PhysRevD.80.124030, LLT18, Singh_2013, Bojowald_2002}. Considerable effort has therefore focused on extracting concrete physical predictions to constrain these ambiguities. Remarkable progress has been achieved in the cosmological sector, particularly for flat models \cite{Bojowald_LQC, Ashtekar_2011, LQC_mathematical_structure, APS_06_2}, using the {\it effective dynamics program} developed in Loop Quantum Cosmology (LQC) \cite{Rovelli_Effective_LQC,Alesci_2019,Chiou_2007,Corichi_2012, Kelly_2020_1,Kelly_2020_2, LQC_k=1}. Although it remains mathematically open whether LQC uniquely captures the cosmological sector of full LQG \cite{Engle2007RelatingLQ,Engle2005QuantumFT,fleischhack2014,brunnemann2009}, the approach yields robust insights: for $k=0$ models, the big-bang singularity is replaced by a bounce whose features depend on the underlying quantization choices \cite{APS_06_1,APS_06_2,Dapor_2018_1,Dapor_2018_2,Ashtekar_Improved_Dynamics,Emergent_de_Sitter,Emergent_de_Sitter_detailed}. Crucially, physically viability requires the so-called \textit{improved dynamics} (or $\bar{\mu}$-) \textit{scheme}, in which the ad-hoc discretization (or ``polymerization'') parameter $\varepsilon$ becomes phase-space dependent \cite{Ashtekar_Improved_Dynamics}. Although this scheme has not yet been derived from the full theory, its physical importance is widely recognized.

Extending these successes beyond isotropic cosmology to other solutions such as black hole spacetimes has long been a major objective of the field. However, spherical symmetry has proven technically challenging, and several competing proposals currently exist (see, e.g., \cite{Kelly_2020_1, Gambini:2020nsf, Modesto:2008im, Ashtekar:2018cay, Bodendorfer:2019nvy, Alonso-Bardaji:2022ear, Bojowald:2024naz, ADL20, Gambini:2009vp, Corichi:2015vsa, Corichi:2016nkp, Hergott:2022hjm, Fragomeno:2024tlh, Gingrich:2024mgk, Hergott:2025elg}). A central obstacle in this regard is the absence of a unique $\bar{\mu}$-scheme for such spacetimes \cite{Ashtekar2018PRL,Bodendorfer2019_note,AshtekarOlmedo2020_properties}. Since the physical regulator was introduced heuristically in LQC, it is unclear how to generalize it to spherically symmetric contexts, and no clear relation between the polymerized cosmological and black-hole sectors is currently known.

A key motivation for this work is to establish such a relation between cosmological and black hole solutions, both of which arise as special cases of spherically symmetric gravity. Starting from the full theory, our aim is to derive an effective description of discretized spherically symmetric gravity capable of accommodating both sectors simultaneously. At present, LQG is formulated predominantly with a phase-space independent regulator $\varepsilon$, and many effective models have already arisen in this context \cite{ADL20, Klaus_k=1, DL20}. Recently, the symmetry restriction framework of \cite{Symmetry_Restriction} has provided a rigorous method for accessing symmetric subsectors of the full phase space, and we shall apply it explicitly to the case of spherical symmetry. This allows us to unify and compare earlier approaches employing phase-space independent regulators and exhibiting strong sensitivity to discretization choices.

In our setting, gravity on a graph is described by holonomies and gauge-covariant fluxes on each edge. Choosing a symmetry group $\Phi$ of symplectomorphisms, we show that the $\Phi$-invariant phase space is parameterized by only a few degrees of freedom that describe spherically symmetric configurations. Under appropriate assumptions, the generator of dynamics may also be consistently restricted to this subspace \cite{Symmetry_Restriction}. The contributions of this manuscript are threefold: (1) we demonstrate that Thiemann's construction of the scalar constraint on cubic-like lattices \cite{Thiemann_QSD,Thiemann_QSD_2,AQG_1} is invariant under arbitrary diffeomorphism-induced symplectomorphisms; (2) we provide a unified method for parameterizing the invariant phase space for holonomies and fluxes simultaneously; and (3) we present, to our knowledge, the first example of non-trivial corrections to the restricted symplectic structure arising at the kinematical level in the form of \textit{phase-space dependent} corrections. \\

\noindent The organization of this article is as follows: \\

\noindent In Sec.~\ref{s2:continuum}, we introduce notation and briefly review the Hamiltonian formulation of gravity \cite{Arnowitt_2008}, focusing on the \emph{connection formulation}, where general relativity (GR) is recast as a Yang-Mills theory over $SU(2)$ \cite{Ashtekar:1986yd,Barbero_1995,Thiemann_MCQGR}. After outlining the kinematical and dynamical content, we recall the formalism of \textit{symmetry restriction}, which reduces the full phase space of the theory to a subspace of symmetric solutions. Provided certain conditions are met, the dynamics remain confined to this subspace, ensuring that symmetric configurations remain symmetric throughout their evolution \cite{Symmetry_Restriction}. This framework is then illustrated explicitly for spherically symmetric GR in Sec.~\ref{s2.4:spherical_symmetry}.

Sec.~\ref{s3:discrete} introduces a discretized version of canonical GR inspired by lattice-based approaches to quantum gravity \cite{Symmetry_Restriction,Thiemann_MCQGR}. In particular, we present the \textit{gravity on a graph} framework, in which the continuum phase space is \textit{truncated} to a countable set of degrees of freedom defined over a discrete spatial structure. Sec.~\ref{s3.2:discretized_dynamics} discusses the corresponding dynamics, focusing on Thiemann's regularized scalar constraint in the case of coordinate-adapted graphs \cite{Thiemann_QSD,Thiemann_QSD_2,Klaus_Thesis}. Finally, in Sec.~\ref{s3.3:symmetry_restriction}, we extend the discussion of Sec.~\ref{s2.3:symmetry_restriction} to the discrete setting, resulting in a general prescription for the construction of symmetry groups in the context of gravity on a graph.

In Sec.~\ref{s4:sph_symm_GOG}, we specialize to spherically symmetric configurations in gravity on a graph. Sec.~\ref{s4.1:algorithm} outlines a general algorithm for computing effective dynamics via symmetry restriction, independent of any specific symmetry group and incorporating essential features of the \textit{improved dynamics scheme} (i.e., $\bar{\mu}$-like schemes) \cite{Ashtekar_Improved_Dynamics}. Sec.~\ref{s4.2:spherical_graphs} introduces a new class of spatial discretizations adapted to spherical symmetry, enabling the identification of a subgroup of the continuum symmetries which survives truncation. In Sec.~\ref{s4.3.1:invariant_submanifold}, we characterize the full subspace of configurations invariant under the action of this group, followed in Sec.~\ref{s4.3.2:physical_subspace} by the introduction of a more manageable, physically-relevant subspace thereof. Sec.~\ref{s4.4:symplectic_structure} analyzes the restricted symplectic structure, revealing finite-lattice corrections that are expected to encode features of the improved dynamics scheme. Finally, Sec.~\ref{s4.5:scalar_constraint} verifies the invariance of the regularized scalar constraint under the symmetry group, thereby validating the full effective dynamics procedure.

We conclude in Sec.~\ref{s5:discussion}, summarizing our findings and discussing applications to cosmology and black holes which will serve as the subject of forthcoming publications.

\section{Continuum Theory}\label{s2:continuum}
This section serves as a refresher on canonical general relativity, its interplay with symmetry restriction, and the implementation of this framework in the context of spherical symmetry.

In Sec.~\ref{s2.1:connection_formulation}, the kinematical framework of canonical GR is reformulated in terms of the Ashtekar-Barbero (AB) variables \cite{Ashtekar:1986yd,Barbero_1995}. The resulting phase space $\mathcal{M}_{\text{AB}}$ will then be equipped with a symplectic form $\omega_{\text{AB}}$ and subjected to three classes of constraints in Sec.~\ref{s2.2:connection_dynamics}. Famously, the Thiemann identities \cite{Thiemann_QSD,Thiemann_QSD_2} allow us to recast one of these --- namely, the scalar constraint ---
in a regular expression\footnote{In this context, a ``regular expression'' is one that is free of any factors involving reciprocal square roots.} involving Poisson brackets, which is beneficial for the purpose of quantization. Sec.~\ref{s2.3:symmetry_restriction} reviews the core concepts of \textit{symmetry restriction}, as proposed in detail in \cite{Symmetry_Restriction}. In particular, we will discuss the implementation of symmetries on $\M_{\text{AB}}$, allowing for the reduction of one's attention to a submanifold $\overline{\M}_{\text{AB}} \subset \M_{\text{AB}}$ consisting of only symmetric configurations. After having laid the necessary groundwork, we will illustrate this process in the specific case of spherical symmetry in Sec.~\ref{s2.4:spherical_symmetry}.

The experienced reader familiar with the standard terminology may skip to Sec.~\ref{s3:discrete}.

\subsection{Connection Formulation of Canonical GR}\label{s2.1:connection_formulation}
Hamiltonian GR describes gravity in terms of a globally-hyperbolic spacetime manifold $S\cong \R\times\Sigma$ equipped with a Lorentzian metric $g$. In local coordinates\footnote{Our index conventions are as follows: $\mu,\nu,\ldots = 0,1,2,3$; $a,b,\ldots = 1,2,3$; $I,J,\ldots = 1,2,3$.} $(x^\mu) = (t,x^a)$, the metric can be decomposed as
\begin{align} \label{g_ADM}
    g_{\mu\nu}dx^\mu dx^\nu = -N^2\,dt^2 + q_{ab}\left(dx^a + \mathcal{N}^a\,dt\right)\left(dx^b + \mathcal{N}^b\,dt\right),
\end{align}
where $q$ is the induced metric on the spatial hypersurfaces $\Sigma$, and $N\in C^\infty(\Sigma)$, $\mathcal{N} = \mathcal{N}^a\partial_a\in\mathfrak{X}(\Sigma)$ are the \textit{lapse function} and \textit{shift vector}, respectively.
Together with its conjugate momentum, the spatial metric $q$ gives rise to the \textit{ADM phase space} $\M_{\text{ADM}}$ \cite{Arnowitt_2008}. Significant progress toward canonical quantization, however, has been made through an extension of this formalism. In particular, $\M_{\text{ADM}}$ can be viewed as the symplectic reduction of a larger phase space, in which spatial geometry is described by \textit{connections} rather than metrics. This extended framework, known as the \textit{connection formulation}, recasts GR in the language of non-Abelian gauge theories and is based on a new set of canonical variables --- namely, the \textit{Ashtekar-Barbero} (AB) \textit{variables} \cite{Ashtekar:1986yd,Barbero_1995}.

The AB variables coordinatize the phase space $\M_{\text{AB}}$ of an $SU(2)$ Yang-Mills theory on $\Sigma$, described by a collection of 1-forms $A^I(x)\in T_x^*\Sigma$ and vector fields $E_I(x)\in T_x\Sigma$, labelled by internal indices $I = 1,2,3$ \cite{Symmetry_Restriction}. These fields are commonly referred to as the \textit{Ashtekar connection} and \textit{densitized triad}, respectively, with the latter serving as the conjugate momentum of the former. The relation between this phase space and $\M_{\text{ADM}}$ arises through the introduction of a \textit{co-triad}, $\{e^I\}$, satisfying the defining property $q_{ab} = \delta_{IJ}\,e^I_a\,e^J_b$. The AB variables are then defined component-wise as
\begin{align} \label{AE_def}
    A_a^I(x) = \Gamma_a^I(x) + \beta K^I_a(x), &&  E^a_I(x) = |\det e(x)|\,e^a_I(x),
\end{align}
where $\Gamma^I$ is the spin connection compatible with the co-triad, $K^I$ is the extrinsic curvature 1-form, $\beta\in\R\backslash\{0\}$ is the \textit{Barbero-Immirzi parameter}, and $\det e$ denotes the determinant of $e^I_a$ \cite{Barbero_1995,Ashtekar_2021}.

Because the relation $q = \delta_{IJ} e^I e^J$ is invariant under local $O(3)$-transformations of the co-triad, $e^I\mapsto \mathcal{O}^I_{\phantom{I}J} e^J$, one can regard $e^I$ as a component of an $\mathfrak{su}(2)$-valued 1-form\footnote{We remind the reader that even though $SU(2)$ constitutes a double-cover of $SO(3)$, the associated Lie algebras are equivalent: $\mathfrak{so}(3) \cong \mathfrak{su}(2)$ \cite{Arfken}.}. In this case, we write $e(x) = e^I_a(x)\tau_I dx^a$, where the basis elements $\tau_I = -(i/2)\,\sigma_I\in\mathfrak{su}(2)$ are determined by the Pauli matrices $\{\sigma_I\}$ \cite{Thiemann_MCQGR,Isham_1999}. Internal indices are raised and lowered using the Euclidean metric, so that $\tau_I = \delta_{IJ}\tau^J$. Accordingly, the AB variables may be expressed as $A(x) = A^I_a(x)\tau_I dx^a\in T_x^*\Sigma\otimes\mathfrak{su}(2)$ and $E(x) = E_I^a(x)\tau^I\partial_a\in T_x\Sigma\otimes \mathfrak{su}(2)$.

The AB phase space $\M_{\text{AB}}$ is equipped with the symplectic form
\begin{align}
    \label{omega_AB}
    \omega_{\text{AB}} = \frac{2}{\kappa\beta}\int_\Sigma \delta A^I_a(x)\wedge \delta E_I^a(x)\,d^3x,
\end{align}
where $\kappa = 16\pi G$, while $\delta$ denotes the functional exterior derivative appropriate for field-theoretic contexts.
This is a closed and non-degenerate 2-form, thereby endowing $\M_{\text{AB}}$ with the structure of a symplectic manifold. Given any function $f\in C^\infty(\M_{\text{AB}})$, we define the associated Hamiltonian vector field $X_f\in\mathfrak{X}(\M_{\text{AB}})$ via 
\begin{align}
    \iota_{X_f}\omega_{\text{AB}} = \delta f = \left(\delta_{A^I_a}f\right)\delta A^I_a + \left(\delta_{E_I^a}f\right)\delta E_I^a,
\end{align}
with $\delta_F = \delta/\delta F$ denoting functional differentiation. This in turn allows for the construction of Poisson brackets through \cite{Symplectic_Geometry}
\begin{align}
    \left\{f,g\right\} = \iota_{X_g}\iota_{X_f}\omega_{\text{AB}} = \mathcal{L}_{X_g}f.
\end{align}
Applying this to the AB variables, the corresponding Hamiltonian vector fields are
\begin{align}
    X_{A^I_a(x)} = -\frac{\kappa\beta}{2}\,\delta_{E_I^a(x)}, && X_{E_I^a(x)} = \frac{\kappa\beta}{2}\,\delta_{A^I_a(x)},
\end{align}
which leads to the canonical Poisson algebra on $\M_{\text{AB}}$,
\begin{align}
    \label{M_AB_alg}
    \left\{A^I_a(x), E_J^b(x')\right\} = \frac{\kappa\beta}{2}\,\delta^I_J\,\delta_a^b\,\delta^3(x-x').
\end{align}
This finalizes the kinematical description of the connection formulation of gravity.

\subsection{Connection Dynamics}\label{s2.2:connection_dynamics}
The dynamical content enters the framework in the form of a fully constrained Hamiltonian,
\begin{align} \label{AB_Hamiltonian}
    H = \int_\Sigma\left(N C + \mathcal{N}^a D_a + \Lambda^I G_I\right) d^3x \eqqcolon C[N] + D[\mathcal{N}] + G[\Lambda],
\end{align}
with Lagrange multipliers given by the lapse function $N$, shift vector $\mathcal{N}$, and an $\mathfrak{so}(3)$-valued scalar field $\Lambda:\Sigma\to \mathfrak{so}(3)$ \cite{Thiemann_MCQGR,Ashtekar_2004}. Consequently, $\M_{\text{AB}}$ is subjected to three classes of constraints, expressed locally as
\begin{align}
    &C = \frac{1}{\kappa}\left[F^J_{ab} - (1+\beta^2)\,\epsilon^J_{\phantom{J}KL} K^K_a K^L_b\right]\epsilon_J^{\phantom{J}MN}\frac{E_M^a E_N^b}{\sqrt{|\det E|}}, \label{scalar_constraint} \\
    &D_a = \frac{2}{\kappa\beta}\,F^I_{ab} E_I^b - G_I A^I_a, \label{diffeo_constraint} \\
    &G_I = \frac{2}{\kappa\beta}\left(\partial_a E_I^a + \epsilon_{IJ}^{\phantom{IJ}K} A^J_a E_K^a\right), \label{Gauss_constraint}
\end{align}
where $F = dA + A\wedge A$ is the $SU(2)$-curvature of the Ashtekar connection, and $\det E$ denotes the determinant of $E_I^a$ \cite{Symmetry_Restriction,Kobayashi_vol_I,Kobayashi_vol_II}.

The constraints $D_a$ and $C$ are direct analogs of the \textit{diffeomorphism} and \textit{scalar constraints} of the ADM formulation, respectively. In contrast, $G_I$ is a \textit{Gauss constraint} reflecting the additional gauge freedom arising from the choice of co-triad. In fact, one can show that $\M_{\text{AB}}$ reduces to $\M_{\text{ADM}}$ precisely on the constraint surface $G^{-1}(0)\coloneqq \bigcap_I G_I^{-1}(0)$. That is, $\M_{\text{ADM}}$ is the \textit{symplectic reduction} of $\M_{\text{AB}}$ with respect to the Gauss constraint \cite{Symmetry_Restriction,Thiemann_MCQGR}:
\begin{align}
    \M_{\text{ADM}} \cong \M_{\text{AB}} /\!/ G,
\end{align}
where $\M_{\text{AB}}/\!/G$ refers to the subspace $G^{-1}(0) \subset\M_{\text{AB}}$ modulo the gauge flow generated by $G[\Lambda]$ for arbitrary test functions $\Lambda$ \cite{Symplectic_Geometry}. Moreover, the diffeomorphism and Gauss constraints generate spatial diffeomorphisms 
and internal rotations, respectively, which are realized infinitesimally through Poisson brackets of the form $\{\cdot, D[\mathcal{N}]\}$ and $\{\cdot, G[\Lambda]\}$ \cite{Ashtekar_2004}. 
In our dynamical analysis, however, we focus exclusively on the scalar constraint $C[N]$. Since we work entirely at the classical level throughout this article, it is justified to impose the gauge-fixing condition $\mathcal{N} = 0$, thereby eliminating the contribution of the diffeomorphism constraint. In this gauge, time evolution on $\M_{\text{AB}}/\!/G$ is generated solely by $C[N]$, and thus corresponds to normal deformations of the spatial hypersurfaces \cite{Bojowald_2010}.

Now, it is well known that significant complications arise in the transition to a quantum theory if one attempts to employ the scalar constraint in the form \eqref{scalar_constraint} directly, particularly due to the explicit appearance of $|\det E|^{-1/2}$ \cite{Gambini_1998,Ashtekar_2004}. However, these difficulties can be circumvented by means of \textit{Thiemann's identities} \cite{Thiemann_QSD,Thiemann_QSD_2}, which relate certain singular expressions to Poisson brackets involving well-behaved functionals
\begin{equation} \label{Thiemann_identities}
    \begin{gathered}
        \left\{A^I_a,V[\Sigma]\right\} = s_e\frac{\kappa\beta}{8}\,\epsilon^{IJK}\epsilon_{abc}\frac{E^b_J\,E^c_K}{\sqrt{|\det E|}}, \\
        \left\{A^I_a, K[\Sigma]\right\} = s_e\frac{\kappa\beta^3}{2}\,K^I_a,
    \end{gathered}
\end{equation}
where $s_e = \sgn{\det e}$, and the quantities $V[\Sigma],\,K[\Sigma]$ are defined as 
\begin{align}
    \label{V_K_Sigma}
    V[\Sigma] \coloneqq \int_\Sigma\sqrt{|\det E|}\,d^3x, && K[\Sigma] \coloneqq s_e\,\beta^2\int_\Sigma K^I_a\,E_I^a\,d^3x.
\end{align}
The first of these, $V[\Sigma]$, represents the volume of the spatial hypersurface,
while $K[\Sigma]$ is a smeared analog of the extrinsic curvature of $\Sigma$. Using the Thiemann identities, the scalar constraint \eqref{scalar_constraint} can be rewritten as
\begin{align} \label{PB_scalar_constraint}
    C = \frac{4}{\kappa^2\beta}\left[F_{I ab} - \frac{4}{\kappa^2}\left(\frac{1+\beta^2}{\beta^6}\right)\epsilon_{IJK}\left\{A^J_a,K[\Sigma]\right\}\left\{A^K_b,K[\Sigma]\right\}\right]\epsilon^{abc}\left\{A^I_c,V[\Sigma]\right\},
\end{align}
whereby the problematic factor $|\det E|^{-1/2}$ is eliminated\footnote{To obtain this form, we have implicitly absorbed a factor of $s_e$ into the lapse function, i.e., $N\to s_e N$.}.

\subsection{Symmetry Restriction}\label{s2.3:symmetry_restriction}

In this subsection we review the procedure of {\it symmetry restriction}, which provides a rigorous framework for reducing a given theory to a particular sub-sector of interest. Kinematically, this procedure yields invariant subspaces in phase space that reproduce familiar solutions, such as isotropic cosmology or spherically symmetric black holes. The nontrivial aspect of symmetry restriction, however, concerns the generator of dynamics: \textit{a priori}, there is no guarantee that the Hamiltonian flow preserves a chosen symmetric subspace. Naively restricting both the symplectic structure and Hamiltonian may therefore lead to erroneous dynamical predictions. Following \cite{Symmetry_Restriction}, we shall recall how a symmetry group $\Phi$ determines an invariant subspace $\overline{\M}_{\text{AB}}\subset\M_{\text{AB}}$ such that the flow of any $\Phi$-invariant function $f\in C^\infty(\M_{\text{AB}})$ remains within $\overline{\M}_{\text{AB}}$ and coincides with the flow generated by its restriction $f|_{\overline{\M}_{\text{AB}}}$ using the restricted symplectic form $\omega_{\text{AB}}|_{\overline{\M}_{AB}}$.

Let $\Psi < \Diff(\Sigma)$ be any subgroup of spatial diffeomorphisms. Each element $\psi\in\Psi$ induces a \textit{symplectomorphism} on the ADM phase space $\M_{\text{ADM}}$ via a smooth group action $\mathcal{A}^R:\M_{\text{ADM}}\times\Psi\to\M_{\text{ADM}}:(m,\psi)\mapsto\mathcal{A}^R_\psi(m)$ satisfying $\mathcal{A}^R_{\psi\circ\psi'} = \mathcal{A}^R_{\psi'}\circ\mathcal{A}^R_\psi$\footnote{Here, we are referring specifically to a \textit{right} action of $\Psi$ on $\M_{\text{ADM}}$, but the discussion applies equally well to an analogous \textit{left} action $\mathcal{A}^L:\Psi\times\M_{\text{ADM}}\to\M_{\text{ADM}}:(\psi,m)\mapsto \mathcal{A}^L_{\psi}(m)$, which satisfies the definitive property $\mathcal{A}^L_{\psi\circ\psi'} = \mathcal{A}^L_\psi\circ\mathcal{A}^L_{\psi'}$, $\forall \psi,\psi'\in\Psi$.}. Restricting the phase space to $\Psi$-symmetric configurations amounts to identifying the \textit{invariant subspace},
\begin{align} \label{M_bar_ADM}
    \overline{\M}_{\text{ADM}} \coloneqq \left\{m\,\big|\,\mathcal{A}^R_\psi(m) = m,\,\forall\psi\in\Psi\right\} \subset \M_{\text{ADM}},
\end{align}
i.e., the set of \textit{fixed points} with respect to the given $\Psi$-action.

While the construction of $\overline{\M}_{\text{ADM}}$ is often straightforward, the AB phase space $\M_{\text{AB}}$ generally admits no fixed points under a $\Psi$-action alone due to the presence of internal gauge symmetries. To obtain a nontrivial invariant subspace in $\M_{\text{AB}}$, one must therefore construct a new symmetry group $\Phi_{\text{AB}}^\Psi < \Symp(\M_{\text{AB}},\omega_{\text{AB}})$ associated with $\Psi$. Given the underlying structure of $\M_{\text{AB}}$, this group naturally embeds into the semidirect product $O(3)\rtimes\Psi$, where $\Psi$ acts on spatial indices and $O(3)$ acts through internal gauge transformations.

In particular, we consider the right action of $\Psi$ on $\M_{\text{AB}}$ defined by
\begin{align} \label{Psi_action_AB}
    \mathcal{A}^R_\psi(A,E) \coloneqq \left(\psi^*A,\psi^*E\right),
\end{align}
where the components of an arbitrary tensor field $T$ of density weight $w$ transform under the pullback $\psi^*$ as
\begin{align} \label{pullback}
    \left(\psi^* T\right)^{a_1\cdots a_k}_{b_1\cdots b_\ell}(x) &= \big|\det D\psi(x)\big|^w\,\frac{\partial\left(\psi^{-1}\right)^{a_1}(y)}{\partial y^{a_1'}}\cdots \frac{\partial\left(\psi^{-1}\right)^{a_k}(y)}{\partial y^{a_k'}}\Bigg|_{y=\psi(x)} \nonumber \\
    &\qquad\qquad \times\frac{\partial\psi^{b_1'}(x)}{\partial x^{b_1}}\cdots \frac{\partial\psi^{b_\ell'}(x)}{\partial x^{b_\ell}}\, T^{a_1'\cdots a_k'}_{b_1'\cdots b_\ell'}\big(\psi(x)\big),
\end{align}
with $D\psi$ denoting the Jacobian of $\psi$. On the other hand, the isomorphism $O(3) \cong \Z_2\times SO(3)$ provides an association $\mathcal{O} \leftrightarrow (s_\mathcal{O},O)$ for each $\mathcal{O}\in O(3)$, where 
\begin{align}
    s_\mathcal{O} \coloneqq \sgn{\det\mathcal{O}}\in\{-1,1\}\cong\Z_2, && O \coloneqq s_\mathcal{O}\,\mathcal{O}\in SO(3).
\end{align}
The induced action $\mathcal{A}^R:\M_{\text{AB}}\times O(3) \to \M_{\text{AB}}$ is given by
\begin{align} \label{SO3_action}
    \mathcal{A}^R_{\mathcal{O}}(A,E) = \mathcal{A}^R_{(s_\mathcal{O},O)}(A,E) \coloneqq \mathcal{A}^R_{s_\mathcal{O}}\left(O^T A O + O^T dO, O^T E O\right),
\end{align}
where the explicit $\Z_2$-action reads (recalling $A\coloneqq \Gamma + \beta K$)
\begin{align}\label{Z2_action}
    \mathcal{A}^R_{s_\mathcal{O}}(A,E) \coloneqq \left(\Gamma + s_\mathcal{O}\beta K, s_\mathcal{O} E\right).
\end{align}
In other words, internal reflections effectively flip the sign of $\beta$ and the orientation of the co-triad. Finally, the total action of $O(3)\rtimes \Psi$ on $\M_{\text{AB}}$ is taken as $\mathcal{A}^R_{(\mathcal{O},\psi)}(A,E) \coloneqq \mathcal{A}^R_\mathcal{O}(\psi^*A,\psi^*E)$, which is indeed symplectic \cite{Symmetry_Restriction,Thiemann_MCQGR}.

The symmetry group $\Phi_{\text{AB}}^\Psi$ itself is obtained by associating to each $\psi\in\Psi$ a gauge transformation $\mathcal{O}_\psi\in O(3)$ whose action is equivalent to the \textit{pushforward} $\psi_* = (\psi^{-1})^*$. Following the approach of \cite{Symmetry_Restriction}, one begins by identifying a parameterization of the $\Psi$-invariant ADM phase space $\overline{\M}_{\text{ADM}}$. Any embedding $\overline{\M}_{\text{ADM}}\hookrightarrow \M_{\text{AB}}$ produces a subspace of $\M_{\text{AB}}$ which collapses to $\overline{\M}_{\text{ADM}}$ upon symplectic reduction by the Gauss constraint. If this subspace is characterized by AB variables $(\overline{A},\overline{E})$, one finds
\begin{align} \label{O_psi}
    \left(\mathcal{O}_\psi\right)^I_{\phantom{I}J} \coloneqq \frac{\left(\psi^*\overline{E}\right)^I_a\overline{E}_J^a}{\left|\det\overline{E}\right|} \implies \mathcal{A}^R_{\mathcal{O}_\psi}\left(\psi^*\overline{E}\right) = \overline{E},
\end{align}
so that $\mathcal{O}_\psi$ precisely compensates the pullback of $\overline{E}$ under $\psi$\footnote{As shown in \cite{Symmetry_Restriction}, it is sufficient to consider $\overline{E}$ alone in the construction of $\mathcal{O}_\psi$.}. The resulting symmetry group on $\M_{\text{AB}}$ is therefore
\begin{align}\label{Phi_AB}
    \Phi_{\text{AB}}^\Psi &= \big\{\left(\mathcal{O}_\psi,\psi\right)\,\big|\,\psi\in\Psi,\,\left(\mathcal{O}_\psi\right)^I_{\phantom{I}J}\left(\psi^*\overline{E}\right)_I = \overline{E}_J\big\} < O(3)\rtimes\Psi,
\end{align}
with each $\mathcal{O}_\psi$ given by \eqref{O_psi}.

Once the $\Phi_{\text{AB}}^\Psi$-invariant phase space $\overline{\M}_{\text{AB}}$ has been identified, we wish to understand the dynamics of the symmetric configurations it contains. Given a Hamiltonian $H$ on $\M_{\text{AB}}$ (e.g., \eqref{AB_Hamiltonian}), there is no guarantee that $\overline{\M}_{\text{AB}}$ is preserved under the associated flow. In other words, one cannot be certain \textit{a priori} that initially-symmetric configurations will remain symmetric throughout their evolution. However, for many symmetry groups of physical interest, this issue can be conveniently bypassed by means of the following result \cite{Symmetry_Restriction}:
\begin{theorem}[Symmetry Restriction of Dynamics] \label{thm}
    Let $(\M,\omega)$ be a symplectic manifold and $\Phi < \Symp(\M,\omega)$ a symmetry group. Suppose that $\overline{\M}\subset \M$ is the $\Phi$-invariant subspace and $\overline{\omega}\coloneqq \omega|_{\overline{\M}}$ is nondegenerate. If $H\in C^\infty(\M)$ is $\Phi$-invariant, then the Hamiltonian flows of $H$ and $\overline{H}\coloneqq H|_{\overline{\M}}$ coincide on $\overline{\M}$. In particular, $\overline{\M}$ is preserved under the evolution generated by $H$.
\end{theorem}

The significance of Theorem~\ref{thm} lies in the implication that for suitably behaved symmetry groups, the invariant subspace evolves independently of the full phase space. This means that the restricted phase space inherits a well-defined Hamiltonian dynamics, thus ensuring consistency between the original and reduced systems. In practical terms, this justifies symmetry restriction in Hamiltonian mechanics, allowing for the study of a lower-dimensional system that retains the essential physics, while potentially offering significant analytical simplifications.

In the following subsection, we showcase the application of symmetry restriction to the well-known case of spherically symmetric solutions in continuum GR. Doing so, we obtain the expected phase space parameterization, symplectic structure and restricted dynamics familiar from the literature \cite{Gambini_2020,Bojowald_Spherical_Symmetry,Kelly_2020_1}.

\subsection{Spherical Symmetry}\label{s2.4:spherical_symmetry}

Let $\Sigma$ be an orientable, spherically symmetric manifold. By assumption, there exists a subgroup $\Psi < \Diff(\Sigma)$ such that $\Psi \cong O(3)$ and whose orbits are 2-spheres \cite{Wald_GR}. For each $x\in\Sigma$, the pullback of the spatial metric to the orbit $\Psi(x) = \{\psi(x)\,|\,\psi\in\Psi\}$ must therefore be proportional to the standard $S^2$-metric, $d\Omega^2 = d\theta^2 + \sin^2(\theta)\,d\varphi^2$. A representative $\Psi$-invariant metric on $\Sigma$ may thus be written as
\begin{align} \label{q_bar_f_g}
    \overline{q}_{ab}\,dx^a dx^b = \frac{p_b^2}{p_a}\, dr^2 + p_a\,d\Omega^2,
\end{align}
where $p_a,p_b$ depend only on the radial coordinate $r$ and possibly time \cite{Kelly_2020_1}.

A co-triad compatible with the metric $\overline{q}$ in \eqref{q_bar_f_g} is given by
\begin{align}
    e^I_a\,\tau_I\,dx^a = \frac{p_b}{\sqrt{p_a}}\,\tau_1\,dr + \sqrt{p_a}\left(\tau_2\,d\theta + \sin(\theta)\,\tau_3\,d\varphi\right),
\end{align}
which satisfies $\overline{q} = \delta_{IJ}e^I e^J$. It follows from \eqref{AE_def} that the associated spherically symmetric AB variables $(\overline{A},\overline{E})$ read
\begin{equation} \label{AE_bar}
    \begin{gathered}
    \overline{A} = a\,\tau_1\,dr + \left(b\,\tau_2 + \Pi\,\tau_3\right)d\theta + \left[\cos(\theta)\,\tau_1 - \big(\Pi\,\tau_2 - b\,\tau_3\big)\sin(\theta)\right]d\varphi,\\
        \overline{E} = \sgn{p_b}\, p_a\sin(\theta)\,\tau^1\,\partial_r + |p_b|\left(\sin(\theta)\,\tau^2\,\partial_\theta + \tau^3\,\partial_\varphi\right),
    \end{gathered}
\end{equation}
where we have defined
\begin{align}\label{Pi_def}
    \Pi \coloneqq \frac{\partial_r p_a}{2\,p_b},
\end{align}
and the $\theta,\varphi$-independent functions $a,b$ characterize the extrinsic curvature \cite{Kelly_2020_1,Kelly_2020_2,Klaus_k=1}
\begin{align}
    \overline{K} = \frac{a}{\beta}\,\tau_1\,dr + \frac{b}{\beta}\left(\tau_2\,d\theta + \sin(\theta)\,\tau_3\,d\varphi\right).
\end{align}

To construct the phase space symmetry group $\Phi_{\text{AB}}^\Psi < \Symp(\M_{\text{AB}},\omega_{\text{AB}})$, we work within the decomposition $\Psi \cong SO(3)\times \Z_2$, where $\Z_2 \cong \{\mathrm{Id}_\Sigma,\rho_\theta\}$ consists of the identity map and reflection $\rho_\theta$ through the equatorial plane. The rotation subgroup is generated by the three Killing vector fields on $S^2$ \cite{Carroll_2019},
\begin{align}
    k_1 = \partial_\varphi, && k_2 = \cos(\varphi)\,\partial_\theta - \cot(\theta)\sin(\varphi)\,\partial_\varphi, && k_3 = -\sin(\varphi)\,\partial_\theta - \cot(\theta)\cos(\varphi)\,\partial_\varphi,
\end{align}
which satisfy $[k_J,k_K] = \epsilon^I_{\phantom{I}JK}k_I$ with respect to the Lie bracket. Acting infinitesimally, these fields generate transformations of $(\overline{A},\overline{E})$ described by the Lie derivatives,
\begin{equation} \label{AE_bar_inf}
    \begin{gathered}
        \mathcal{L}_{k_L}\overline{A}^I = d\Lambda_{k_L}^I + \epsilon^I_{\phantom{I}JK}\overline{A}^J\Lambda_{k_L}^K = -\big\{\overline{A}^I,G[\Lambda_{k_L}]\big\}, \\
        \mathcal{L}_{k_L}\overline{E}_I = -\epsilon_{IJ}^{\phantom{IJ}K}\Lambda_{k_L}^J\overline{E}_K = -\left\{\overline{E}_I,G[\Lambda_{k_L}]\right\},
    \end{gathered}
\end{equation}
where each $\Lambda_{k_L}\in\mathfrak{so}(3)$ determines an infinitesimal gauge transformation generated by the Gauss constraint. In terms of $SO(3)$-generators $(T_I)_{JK} = -\epsilon_{IJK}$, 
these are given explicitly by
\begin{align}\label{continuum_gauge_params}
    \Lambda_{k_1} = 0, && \Lambda_{k_2} = -\csc(\theta)\sin(\varphi) \,T_1, && \Lambda_{k_3} = -\csc(\theta)\cos(\varphi) \,T_1.
\end{align}
Thus, the action of the global diffeomorphisms $\psi_{k_L}(t) = \exp(t\,k_L)$ may be compensated by internal rotations $O_{k_L}(t) = \exp(t\,\Lambda_{k_L})$, and the orientation-preserving subgroup $\Phi_{\text{AB}}^{\Psi^+} < \Phi_{\text{AB}}^\Psi$ consists of pairs $(O_\psi,\psi)$ associated with these flows.

Turning our attention to the reflection subgroup $\Z_2 < \Psi$, a straightforward computation shows
\begin{align}
    &\left(\rho_\theta^*\overline{A}, \rho_\theta^*\overline{E}\right) = \mathcal{A}^R_{\mathcal{O}_\theta}\left(\overline{A},\overline{E}\right),
\end{align}
with $\mathcal{O}_\theta = \mathrm{diag}(1,-1,1)\in O(3)$. Since $\mathcal{O}_\theta^{-1} = \mathcal{O}_\theta$, we obtain the orientation-reversing subgroup $\Phi_{\text{AB}}^{\Psi^-} = \left\{\left(\mathrm{Id}_\Sigma,\mathrm{Id}_{O(3)}\right),\left(\rho_\theta,\mathcal{O}_\theta\right)\right\}< \Phi_{\text{AB}}^\Psi$. Combining the two components, we obtain the full symmetry group of interest, $\Phi_{\text{AB}}^\Psi \cong \Phi_{\text{AB}}^{\Psi^+}\times \Phi_{\text{AB}}^{\Psi^-}$.

Because $\Phi_{\text{AB}}^\Psi$ is a finite-dimensional compact Lie group, it automatically satisfies the conditions of Theorem~\ref{thm} \cite{Symmetry_Restriction,Brocker_Compact_Lie_Groups}. The invariant subspace $\overline{\M}_{\text{AB}}\subset\M_{\text{AB}}$ consists of triads and connections of the form \eqref{AE_bar}, and the induced symplectic form reads\footnote{Note that the AB variables \eqref{AE_bar} do not exhaust the space of $\Phi_{\text{AB}}^\Psi$-invariant connections and triads. Some references (e.g., \cite{Ashtekar_2005,Gambini_2020,Bojowald_Spherical_Symmetry}) include off-diagonal triad components, but such configurations violate the Gauss constraint and require additional gauge fixing. In the physical situations we consider, one can always choose $(\overline{A},\overline{E})$ as in \eqref{AE_bar}, for which $G_I\equiv 0$.}
\begin{align} \label{omega_bar_AB}
    \overline{\omega}_{\text{AB}} = \omega_{\text{AB}}\big|_{\overline{\M}_{\text{AB}}} = \frac{8\pi}{\kappa\beta}\int\sgn{p_b}\left(\delta  a\wedge\delta  p_a + 2\,\delta  b\wedge \delta  p_b\right) dr.
\end{align}
Accordingly, the restricted phase space is coordinatized by two canonical pairs, $(a,p_a)$ and $(b,p_b)$, satisfying the formal Poisson algebra
\begin{align} \label{M_bar_AB_alg}
    \left\{a(r),p_a(r')\right\} = \frac{\kappa\beta}{8\pi}\,\sgn{p_b}\,\delta(r-r') = 2\left\{b(r),p_b(r')\right\},
\end{align}
with all other brackets vanishing. Let us note in passing that \eqref{omega_bar_AB} has turned out to be independent of the quantity $\Pi$ defined in \eqref{Pi_def} --- a simplification that will \textit{not} persist in the discretized theory, resulting in non-trivial corrections to the symplectic structure in the presence of finite regulators.

To apply Theorem~\ref{thm}, it remains to verify the $\Phi_{\text{AB}}^\Psi$-invariance of the scalar constraint $C[N]$. Given any smooth function $f\in C^\infty(\M_{\text{AB}})$, the generators $(k_L,\Lambda_{k_L})\in\mathfrak{X}(\Sigma)\times\mathfrak{so}(3)$ act via \cite{Symplectic_Geometry}
\begin{align}
    \mathcal{A}^R_{(k_L,\Lambda_{k_L})}(f) = \left\{f, D[k_L] + G[\Lambda_{k_L}]\right\}.
\end{align}
Applying this to $C[N]$, we find
\begin{align}
    \mathcal{A}^R_{(k_L,\Lambda_{k_L})}\left(C[N]\right) = C\left[\mathcal{L}_{k_L}N\right],
\end{align}
which follows immediately from the relations $\{C[N],G[\Lambda]\} = 0$ and $\{C[N],D[\mathcal{N}]\} = C[\mathcal{L}_\mathcal{N}N]$ \cite{Ashtekar_2004,Bojowald_2010}. Thus, choosing a lapse function such that $\mathcal{L}_{k_L}N = 0$ for each rotational Killing vector yields a scalar constraint that is invariant under $\Phi_{\text{AB}}^{\Psi^+}$. On the other hand, a direct computation demonstrates that $C[N]$ is also invariant with respect to the reflection $(\rho_\theta,\mathcal{O}_\theta)\in\Phi_{\text{AB}}^{\Psi^-}$, thereby establishing complete $\Phi_{\text{AB}}^\Psi$-invariance of the scalar constraint.

Computing the restriction $\overline{C}[N] \coloneqq C[N]|_{\overline{\M}_{\text{AB}}}$ using either \eqref{scalar_constraint} or \eqref{PB_scalar_constraint}, we obtain
\begin{align}\label{C_Bar_AB}
    \overline{C}[N] = -\frac{8\pi}{\kappa}\int \sgn{p_b}\,N\sqrt{p_a}\,\Bigg(\left[1 + \left(\frac{b^2}{\beta^2} - \Pi^2\right)\right]\frac{p_b}{p_a} + 2\left[\frac{a b}{\beta^2} - \partial_r \Pi\right]\Bigg) dr,
\end{align}
with $N = N(r)$ ensuring invariance under the flow generated by the $S^2$ Killing vector fields. In the absence of any matter content, the equations of motion on $\overline{\M}_{\text{AB}}$ follow from Hamilton's equations, $\dot{f}(a,b,p_a,p_b) = \{f,\overline{C}[N]\} + \partial_t f$, and can be shown to reproduce the standard spherically symmetric dynamics \cite{Kelly_2020_1}.

\section{Discretized Theory}\label{s3:discrete}
We now turn our attention to a discretization of the continuum theory developed above. Throughout this section, the underlying lattice is kept arbitrary and will be specified explicitly only in Sec.~\ref{s4:sph_symm_GOG}. Beyond the ambiguities encountered upon quantization, discretization itself introduces additional \textit{discretization ambiguities}, most notably in the regularization of the continuum scalar constraint \eqref{PB_scalar_constraint}. Here we adopt a regularization closely related to that of \cite{AQG_1}, while incorporating \textit{gauge-covariant fluxes} due to their preferable transformation properties under the theory's $SU(2)$ gauge group \cite{QSD_VII}.

After establishing the kinematical and dynamical structure of discretized gravity, Sec.~\ref{s3.3:symmetry_restriction} extends the symmetry restriction framework of Sec.~\ref{s2.3:symmetry_restriction} to the discrete phase space, which we denote by $\M_\gamma$. In particular, we show how a continuum symmetry group $\Phi_{\text{AB}}^\Psi < \Symp(\M_{\text{AB}},\omega_{\text{AB}})$ can be reduced to a subgroup compatible with a given spatial discretization. Lifting the action of this subgroup from $\M_{\text{AB}}$ to $\M_\gamma$, one obtains a discretized symmetry group $\Phi_\gamma^\Psi < \Symp(\M_\gamma,\omega_\gamma)$ that is potentially suitable for the application of Theorem~\ref{thm}.

\subsection{Holonomy-Flux Representation: Gravity on a Graph}\label{s3.1:GOG}
To transition to a description of GR more suitable for quantization, our aim is to replace the continuous AB variables with appropriate discretized analogs. For that, we \textit{truncate} the full AB phase space, retaining only certain information about the connection and triad pertaining to an underlying discrete spatial structure. In particular, the Ashtekar connection is replaced by its $SU(2)$-valued \textit{holonomies} over the edges of a graph $\gamma$, while the densitized triad is represented by its $\mathfrak{su}(2)$-valued (gauge-covariant) \textit{fluxes} over the surfaces of an associated dual cell complex $\gamma^*$ \cite{AQG_1,AQG_2,Thiemann_MCQGR,QSD_VII}.

We consider a graph $\gamma$ to be a (countable) collection of semianalytic parameterized curves (i.e., \emph{edges}) $e: [0,1]\rightarrow\Sigma$ 
modulo orientation. Each such edge runs between a pair of adjacent \emph{vertices} (i.e., distinguished points in $\Sigma$), and we denote the collection of all vertices by $\gamma^v \subset \Sigma$. On the other hand, $\gamma^*$ is taken to be the faces of an oriented polyhedral decomposition of $\Sigma$ dual to $\gamma$ \cite{Thiemann_MCQGR}. That is, the dual cell complex can be viewed as a collection of parameterized 2-surfaces $\mathcal{S}:(-1/2,1/2)^2\rightarrow\Sigma$ such that for each $e\in\gamma$, there is exactly one corresponding surface $\mathcal{S}_e\in\gamma^*$ which is transversely intersected by $e$ at a single point
\cite{Symmetry_Restriction,QSD_VII,Ashtekar_2004}. Figure~\ref{fig:spherical_graph_cells} provides illustrative examples of these constructions in the spherically symmetric setting, which will be explored further in Sec.~\ref{s4:sph_symm_GOG}.

The pair $(\gamma,\gamma^*)$ serves as a discretization of the spatial manifold $\Sigma$. To define an associated phase space $\M_\gamma$, we assign to each edge $e\in\gamma$ a \textit{discretization map} $\mathfrak{D}_e:\M_{\text{AB}}\to SU(2)\times\mathfrak{su}(2)$ such that
\begin{align}
    \label{D_gamma_def}
    \mathfrak{D}_e(A,E) &= (h_e,P_e) \coloneqq \left(\mathcal{P}\exp\left(\int_e A\right), h_{1/2}(e)\left[\int_{\mathcal{S}_e}h(\rho_e)\,(\star E) \,h^\dagger(\rho_e)\right]h_{1/2}^\dagger(e)\right),
\end{align}
where $h_e = h_e[A]$ denotes the holonomy of $A$ along $e$, while $P_e = P_e[A,\star E]$ refers to the gauge-covariant flux\footnote{The use of holonomies in the definition of $P_e$ ensures that the flux transforms covariantly under the gauge group $SU(2)$ \cite{Thiemann_MCQGR,Symmetry_Restriction}.} of $E$ over the dual surface $\mathcal{S}_e$. In these expressions, $\mathcal{P}$ is the path-ordering symbol, while $\star E = (1/2)\,\epsilon_{abc}\,E^a\,dx^b\wedge dx^c$ is the Hodge dual of the densitized triad. Moreover, $h_{1/2}(e)$ refers to the holonomy along $e$ from its source vertex $v = e(0)\in\gamma^v$ to the point of intersection with the surface, $p_e \coloneqq e\cap\mathcal{S}_e$, which we identify with the midpoint $e(1/2)$. 
Finally, $h(\rho_e)$ denotes the holonomy along a semianalytic path $\rho_e \subset \mathcal{S}_e$ (not to be confused with any edge in the graph), starting from the reference point $p_e$ and ending at an arbitrary point of integration on the surface\footnote{Note that the construction depends on the explicit choice of paths $\rho_e$ (see \cite{LS19} for examples).}.

The discretization map \eqref{D_gamma_def} provides the local structure of a cotangent bundle over each edge in the graph, with $(h_e,P_e)\in SU(2)\times\mathfrak{su}(2) \cong T^*SU(2)$ \cite{Twisted_geometries,Ashtekar_2004}. The full phase space of discretized gravity on a fixed graph can therefore be expressed as a direct product over the edges:
\begin{align}
    \label{M_Gamma}
    \M_\gamma = \left\{\mathfrak{D}_e(A,E)\in SU(2)\times\mathfrak{su}(2)\,\big|\,(A,E)\in\M_{\text{AB}},\,e\in\gamma\right\} \cong \prod_{e\in\gamma} T^* SU(2),
\end{align}
or, more concisely, $\M_\gamma = \mathfrak{D}_\gamma(\M_\text{AB})$, where $\mathfrak{D}_\gamma \coloneqq \left\{\mathfrak{D}_e\right\}_{e\in\gamma}$. 
Accordingly, $\M_\gamma$ is equipped with the  
symplectic form \cite{QSD_VII,Symmetry_Restriction}
\begin{align} \label{omega_gamma}
    \omega_\gamma = \frac{2}{\kappa\beta}\sum_{e\in\gamma}\left(\delta_{IJ}\,\Theta_e^I\wedge\delta P_e^J + \frac{1}{2}\,\epsilon_{IJK} P_e^I\,\Theta_e^J\wedge\Theta_e^K\right),
\end{align}
where $\Theta_e = h_e^\dagger\,\delta h_e:\mathfrak{X}(SU(2))\to\mathfrak{su}(2)$ is the left-invariant Maurer-Cartan form on $SU(2)$ associated with the edge $e$ \cite{Cartan_1904}. Using this expression, the Hamiltonian vector fields associated with the holonomy and flux variables are found to be
\begin{align} \label{VF_M_gamma}
    X_{h_e} = -\frac{\kappa\beta}{2}\,h_e\,\tau^I\,\delta _{P_e^I}, && X_{P_e^I} = \frac{\kappa\beta}{2}\left(L_e^I - \epsilon^I_{\phantom{I}JK}\,P_e^J\,\delta _{P_e^K}\right),
\end{align}
where $L_e^I$ are the left-invariant vector fields on $SU(2)$ associated with the basis $\{\tau_I\}\subset \mathfrak{su}(2)$, and which satisfy $\Theta_e^I(L_e^J) = \delta^{IJ}$. In other words, the set $\{L_e^I\}$ forms a basis for derivatives on $SU(2)$ \cite{Symmetry_Restriction}. Consequently, the fundamental Poisson algebra on $\M_\gamma$ --- i.e., the \textit{holonomy-flux algebra} --- reads
\begin{align}
    \label{hf_alg}
    \left\{h_e,h_{e'}\right\} = 0, && \left\{P_e^I,P_{e'}^J\right\} = -\frac{\kappa\beta}{2}\,\delta_{e,e'}\,\epsilon^{IJ}_{\phantom{IJ}K}P_e^K, && \left\{h_e,P_{e'}^I\right\} = \frac{\kappa\beta}{2}\,\delta_{e,e'}\, h_e \tau^I.
\end{align}

\subsection{Discretized Dynamics}\label{s3.2:discretized_dynamics}
For the purposes of this article, we shall restrict our attention to compact subsets of $\Sigma$ and finite graphs constructed in the following manner: Let $n\in\N$ be a \textit{discretization parameter}, which determines the number of vertices in the graph --- in particular, $|\gamma^v|\sim n^3$. The pair $(\gamma,\gamma^*)$ is taken to be a \textit{cubulation} of $\Sigma$ in the sense that almost every $v\in\gamma^v$ will be 6-valent, with three of the distinct valent edges determining the ``positive'' coordinate-directions on the graph\footnote{We assume that vertices with different valency appear due to some boundary conditions with a reduced scaling $n^2$ at most. Hence, for sufficiently fine graphs, the contributions of those boundary effects can be neglected.}. Thus, for each $v\in\gamma^v$ and $i\in\{1,2,3\}$, there are two associated edges and surfaces, $e^v_{\pm i}\in\gamma$ and $\mathcal{S}^v_{\pm i}\in\gamma^*$, parameterized by $s\in[0,1]$ and $t,u\in(-1/2,1/2)$, respectively, and defined as\footnote{Note that the edge $e^v_{-i}$ is defined in terms of the \textit{inverse} of the positively-oriented edge beginning at the \textit{previous} vertex: $e^v_{-i}(s) \coloneqq \left(e^{v-\mu_i}_i\right)^{-1}(s)$. This definition extends naturally to the surface $\mathcal{S}^v_{-i}$ as well, with only slight modifications due to the different parameterization domains --- see \eqref{surface_param}.}
\begin{align}
    \label{edge_param}
    e^v_i(s) = v + \mu_i(s), && e^v_{-i}(s) = e^{v-\mu_i}_i(1-s),
\end{align}
and
\begin{align}
    \label{surface_param}
    \mathcal{S}^v_i(t,u) = (v+\mu_i(1/2)) + \mu_j(t) + \mu_k(u), && \mathcal{S}^v_{-i}(t,u) = \mathcal{S}^{v-\mu_i}_i(u,t),
\end{align}
where $\mu_i(s)$ denotes the parameterized coordinate-displacement along the edge (note that we typically write $\mu_i$ instead of $\mu_i(1)$), while in the second line it is assumed that $\epsilon_{ijk} = +1$. Notice that when the vertices in $\gamma^v$ are evenly distributed over the three coordinate directions, each axis contains vertices separated by a coordinate distance $\varepsilon_i \sim 1/n$.

Denoting the holonomies and fluxes associated with the coordinate-adapted parameterizations \eqref{edge_param}--\eqref{surface_param} by $h_i(v) \coloneqq h_{e^v_i}$ and $P_i(v) \coloneqq P_{e^v_i}$, the discretized scalar constraint utilizing the regularization of Thiemann \cite{Thiemann_QSD,Thiemann_QSD_2,AQG_1} can be written as follows:
\begin{equation} 
    C_{\gamma}[N] = \sum_{v\in\gamma^v}N(v)\left[C_{\gamma}^E(v) + C_{\gamma}^L(v)\right],\label{C_Thiemann}
\end{equation}
with the \emph{Euclidean} and \emph{Lorentzian} parts of the constraint density given by \cite{Klaus_Thesis,Symmetry_Restriction,Thiemann_MCQGR}
\begin{equation}
    \label{CE_Gamma}
    C_{\gamma}^E(v) = -\frac{1}{2\kappa^2\beta}\sum_{i,j,k\in\mathcal{I}}\epsilon(i,j,k)\,\tr\left(\left[h\left(\square^v_{ij}\right) - h^\dagger\left(\square^v_{ij}\right)\right]\big\{h_k(v),V[\gamma]\big\} h_k^\dagger(v)\right),
\end{equation}
and
\begin{align}
    \label{CL_Gamma}
    C_{\gamma}^L(v) &= \frac{2^3(1 + \beta^2)}{\kappa^4\beta^7}\sum_{i,j,k\in\mathcal{I}}\epsilon(i,j,k)\, \tr\Big[\big\{h_i(v),K[\gamma]\big\} h_i^\dagger(v)\big\{h_j(v),K[\gamma]\big\} \nonumber \\ 
    &\qquad\qquad\qquad \times h_j^\dagger(v) \big\{h_k(v),V[\gamma]\big\} h_k^\dagger(v)\Big],
\end{align}
respectively, where $\mathcal{I} \coloneqq \{\pm 1,\pm 2, \pm 3\}$. In these expressions, $N:\gamma^v\rightarrow\mathbb{R}$ is some compactly-supported lapse function on the graph and $\epsilon(i,j,k) = \sgn{ijk}\,\epsilon_{|i||j||k|}$ is a generalized 
Levi-Civita symbol\footnote{From now on, we adopt the index convention $i,j,\ldots\in\mathcal{I}$, allowing us to account for both positive and negative edge orientations. In case we need to refer only to graph indices taking positive values, we will write $|i|,|j|,\ldots\in\mathcal{I}^+$.}. Moreover, we note that
\begin{align} \label{V_Gamma}
    V[\gamma] &= \sum_{v\in\gamma^v}V_\gamma(v), & V_\gamma(v) &\coloneqq \sqrt{\frac{1}{2^3\cdot 3!}\bigg|\sum_{i,j,k}\epsilon(i,j,k)\epsilon_{IJK}P^I_i(v)P^J_j(v)P^K_k(v)\bigg|} 
\end{align}
is the volume of our discretized spatial manifold, and
\begin{align} \label{K_Gamma}
    K[\gamma] &= \left\{C_{\gamma}^E[1],V[\gamma]\right\}, & C_{\gamma}^E[1] \coloneqq s_e\sum_{v\in\gamma^v}C_{\gamma}^E(v)
\end{align}
refers to a regularized analog of the continuum quantity $K[\Sigma]$ defined in (\ref{V_K_Sigma}), and 
\begin{align}
    \label{h_box}
    h\left(\square^v_{ij}\right) \coloneqq h_i(v) h_j(v+\mu_i) h_i^\dagger(v+\mu_j) h_j^\dagger(v)
\end{align}
is ``plaquette holonomy'' --- that is, the aggregate holonomy around a closed loop in the $ij$-plane of the graph.

The above equations (\ref{C_Thiemann})--(\ref{h_box}) form the building block on which quantizations in the LQG and Algebraic Quantum Gravity (AQG) frameworks rests. However, in this article, we do not perform any actual quantization, but rather adopt the {\it effective dynamics} perspective --- namely, that semi-classical effects can already be captured by studying the classical discrete phase-space in the presence of a finite regulator (either fixed, i.e. $\varepsilon=\mu_0$, or dynamical, i.e. $\varepsilon\mapsto \bar\mu$). As the investigations in this paper remain on a classical phase space and target rather the isolation of particular solution, such as spherical symmetry, it is actually possible to evaluate the Poisson brackets appearing in \eqref{CE_Gamma} and \eqref{CL_Gamma} already at this stage. Given the symplectic structure on the graph phase space (\ref{hf_alg}), this yields well-defined expressions as polynomial fractions of $h_e$ and $P_e$. Those expressions are explicitly derived in \cite{ADL20,DL20} as well as in Appendix \ref{app:poly_scalar_constr} and we continue to use those for the investigations in Sec.~\ref{s4:sph_symm_GOG}\footnote{In fact, the expressions (\ref{C_Thiemann})--(\ref{h_box}) have obtained using conventional fluxes $E(\mathcal{S})=\int_{\mathcal{S}} \star E$, which are commuting with respect to each other in the Poisson brackets. Hence, using gauge-covariant fluxes to compute the dynamics, makes only sense after the Poisson brackets have been resolved. In simple test cases it has been found that this has no impact on the resulting dynamics \cite{Symmetry_Restriction}.}.

\subsection{Truncated Symmetry Restriction}\label{s3.3:symmetry_restriction}
In order to implement an analog of the symmetry restriction framework outlined in Sec.~\ref{s2.3:symmetry_restriction}, we aim to translate the continuum symmetry group $\Phi_{\text{AB}}^\Psi < \Symp(\M_{\text{AB}},\omega_{\text{AB}})$ into a corresponding group $\Phi_\gamma^\Psi < \Symp(\M_\gamma,\omega_\gamma)$ acting on the discretized phase space. This requires identifying a suitable lift of the combined diffeomorphism and gauge group action on $\M_{\text{AB}}$ to the level of holonomies and fluxes, subject to the structural constraints imposed by truncation to a finite graph. The following construction defines this lift explicitly and identifies the discrete symmetry group with respect to which the restricted dynamics will be formulated.

The right action \eqref{Psi_action_AB} of $\Psi < \Diff(\Sigma)$ on $\M_{\text{AB}}$ admits a natural lift to the discrete phase space $\M_\gamma$, defined by
\begin{align}
    \label{Psi_action_gamma_AE}
    \mathcal{A}^R_\psi\left(h_e[A],P_e[A,\star E]\right) \coloneqq \left(h_e[\psi^*A],P_e[\psi^*A,\psi^*(\star E)]\right) = \left(h_{\psi(e)}[A],P_{\psi(e)}[A,\star E]\right),
\end{align}
in which $\psi(e)\coloneqq \psi\circ e$, and $P_{\psi(e)}$ is to be understood as the flux over $\psi(\mathcal{S}_e)\coloneqq \psi\circ\mathcal{S}_e$. The final equality here allows the transformation to be viewed in terms of an equivalent \textit{left} action of $\Psi$ on the graph and dual cell complex, namely,
\begin{align}
    \label{Psi_left_action}
    \mathcal{A}^L:\Psi\times(\gamma,\gamma^*)\to(\gamma,\gamma^*):\big(\psi,(e,\mathcal{S}_e)\big)\mapsto \mathcal{A}^L_\psi(e,\mathcal{S}_e)\coloneqq \big(\psi(e),\psi(\mathcal{S}_e)\big).
\end{align}
From this perspective, the right action on the phase space variables simplifies to a relabelling of the edges and surfaces:
\begin{align} 
    \label{Psi_action_gamma}
    \mathcal{A}^R_\psi\left(h_e,P_e\right) = \left(h_{\psi(e)},P_{\psi(e)}\right),
\end{align}
where the dependence on the underlying fields $(A,E)$ via \eqref{Psi_action_gamma_AE} is implicit.

The action of the internal gauge group on $\M_\gamma$, on the other hand, is implemented through $SU(2)$-valued fields $U:\Sigma\to SU(2)$. Explicitly, we define
\begin{align}
    \label{SU2_action_gamma}
    \mathcal{A}^R_U\left(h_e,P_e\right) \coloneqq \left(U^\dagger(b_e)\,h_e\,U(f_e), U^\dagger(b_e)\,P_e\,U(b_e)\right),
\end{align}
with $b_e\coloneqq e(0)$ and $ f_e\coloneqq e(1)$. The full discretized symmetry group $\Phi_\gamma^\Psi$ thus arises as a subgroup of $SU(2)\rtimes\Psi$, whose action can be deduced by applying \eqref{SU2_action_gamma} to the result of \eqref{Psi_action_gamma}:
\begin{align}
    \label{full_action_gamma}
    \mathcal{A}^R_{(U,\psi)}\left(h_e,P_e\right) = \left(U^\dagger\left(b_{\psi(e)}\right)\,h_{\psi(e)}\,U\left(f_{\psi(e)}\right), U^\dagger\left(b_{\psi(e)}\right)\,P_{\psi(e)}\,U\left(b_{\psi(e)}\right)\right).
\end{align}

While the action \eqref{full_action_gamma} reflects the structure of the continuum symmetry group $\Phi_{\text{AB}}^\Psi$, the lift from $\M_{\text{AB}}$ to $\M_\gamma$ is not automatic: only certain elements of $\Phi_{\text{AB}}^\Psi$ induce well-defined transformations on the discretized phase space \cite{Symmetry_Restriction}. Because the truncation of $\M_{\text{AB}}$ to $\M_\gamma$ deprives us of access to $\Sigma$ in its entirety, the group $\Psi < \Diff(\Sigma)$ must be reduced to only those elements that preserve the discrete spatial structure imposed by $(\gamma,\gamma^*)$. More precisely, our interest lies in the group of \textit{graph-preserving diffeomorphisms}, defined as the stabilizer of $(\gamma,\gamma^*)$ in $\Psi$\footnote{Notice that in the parameterization \eqref{edge_param}--\eqref{surface_param}, every graph-preserving diffeomorphism can be viewed as a permutation of the vertices and edge-orientation, i.e., $\mathrm{Stab}_\Psi(\gamma,\gamma^*) < \mathrm{Aut}(\gamma^v\times\mathcal{I})$.}:
\begin{align}
    \label{Stab_Psi_gamma}
    \mathrm{Stab}_\Psi(\gamma,\gamma^*) \coloneqq \left\{\psi\in\Psi\,\big|\,\mathcal{A}^L_\psi\left(e,\mathcal{S}_e\right)\in (\gamma,\gamma^*),\,\forall e\in\gamma\right\}.
\end{align}
Additionally, we shall neglect internal reflections and only consider the subgroup of rotations generated by the Gauss constraint \eqref{Gauss_constraint}, which can always be mapped to corresponding gauge transformations on $\M_\gamma$ \cite{Symmetry_Restriction}. Due to the construction of $\Phi_{\text{AB}}^\Psi$, this latter condition is achieved by a further reduction of $\mathrm{Stab}_\Psi(\gamma,\gamma^*)$ to its orientation-preserving subgroup,
\begin{align}
    \label{Psi_gamma}
    \Psi_\gamma \coloneqq \left\{\psi\in\mathrm{Stab}_\Psi(\gamma,\gamma^*)\,\big|\,\det(D\psi(x)) > 0\right\},
\end{align}
i.e., the group of \textit{proper} graph-preserving diffeomorphisms (recall that $D\psi$ denotes the Jacobian of the transformation induced by $\psi$). The subgroup of $\Phi_{\text{AB}}^\Psi$ whose lift we are concerned with is then given by
\begin{align}
    \label{Phi_AB_gamma}
    \Phi_{\text{AB}}^{\Psi_\gamma} \coloneqq \left\{\left(O_\psi,\psi\right)\in\Phi_{\text{AB}}^\Psi\,\big|\,\psi\in\Psi_\gamma\right\} < SO(3)\rtimes \Psi_\gamma.
\end{align}

For each $\psi\in\Psi_\gamma$, the associated rotation $O_\psi\in SO(3)$ is given by \eqref{O_psi} and may be lifted to an element $U_\psi\in SU(2)$ by inversion of the \textit{double-cover homomorphism} (or \textit{spin-1 representation} of $SU(2)$), $D^{(1)}:SU(2)\to SO(3)$. In terms of the $\mathrm{su}(2)$-basis $\{\tau_I\}$, this map takes $\pm U\in SU(2)$ to the orthogonal matrix $D^{(1)}(U)\in SO(3)$ having components \cite{Isham_1999}
\begin{align}
    \label{pi_U}
    D^{(1)}_{IJ}(U) = -2\,\tr\big[\mathrm{Ad}_U(\tau_I)\,\tau_J\big],
\end{align}
where $\mathrm{Ad}_U(\cdot) \coloneqq U (\cdot) U^\dagger$ denotes the adjoint action of $SU(2)$ on its Lie algebra. Pragmatically speaking, the identification of the pre-image $(D^{(1)})^{-1}(O)\subset SU(2)$ for a given $O\in SO(3)$ can often be conveniently streamlined through the following result:
\begin{fact}\label{fact:double_cover}
    Let $U = \exp(u)\in SU(2)$ for some $u = u^I\tau_I\in\mathfrak{su}(2)$. Then the components \eqref{pi_U} of the image $D^{(1)}(U)\in SO(3)$ are given explicitly by
    \begin{align}\label{pi_U_components}
        D^{(1)}_{IJ}(U) = \cos(\lambda)\,\delta_{IJ} - \frac{1}{\lambda}\sin(\lambda)\,\epsilon_{IJK} u^K + \frac{1}{\lambda^2}\big[1-\cos(\lambda)\big] u_I u_J,
    \end{align}
    with $\lambda\coloneqq \Vert u \Vert = \sqrt{\delta_{IJ}u^I u^J}$.
\end{fact}
\begin{proof}
    We simply derive the expression \eqref{pi_U_components} by direct computation. To do so, we first recall that $U = \exp(u)$ can be expressed in the Euler form,
    \begin{align}
        U = \mathrm{Id}_{SU(2)}\cos(\lambda/2) + \frac{2}{\lambda}\sin(\lambda/2)\, u.
    \end{align}
    Combining this with the standard identities
    \begin{align}
        [\tau_J,\tau_K] = \epsilon^I_{\phantom{I}JK}\tau_I, && \tau_I\tau_J\tau_K = -\frac{1}{8}\left[\epsilon_{IJK}\mathrm{Id}_{SU(2)} + 2\left(\delta_{IJ}\tau_K - \delta_{KI}\tau_J + \delta_{JK}\tau_I\right)\right],
    \end{align}
    it is a straightforward task to show
    \begin{align}\label{Ad_U_tau}
        \mathrm{Ad}_U(\tau_K) = \cos(\lambda)\,\tau_K - \frac{1}{\lambda}\sin(\lambda)\,\epsilon^J_{\phantom{J}KL}\tau_J\, u^L + \frac{1}{\lambda^2}\big[1 - \cos(\lambda)\big] u_K u.
    \end{align}
    Then one obtains \eqref{pi_U_components} almost immediately using the fact that $-2\,\tr(\tau_I\tau_J) = \delta_{IJ}$.
\end{proof}
\noindent The preceding result provides us with the necessary means to define the discretized symmetry group of interest,
\begin{align}
    \label{Phi_gamma}
    \Phi_\gamma^\Psi \coloneqq \left\{\left(U_\psi,\psi\right)\,\big|\,\psi\in\Psi_\gamma,\,U_\psi\in\left(D^{(1)}\right)^{-1}(O_\psi)\right\} < SU(2)\rtimes\Psi_\gamma,
\end{align}
which admits a symplectic action on $\M_\gamma$ via \eqref{full_action_gamma}. This group defines the symmetry structure with respect to which we apply Theorem~\ref{thm} in the case where the dynamics are generated by the Thiemann constraint \eqref{C_Thiemann}.

\section{Spherically Symmetric Gravity on a Graph}\label{s4:sph_symm_GOG}
This section presents an in-depth development of the spherically symmetric sector of gravity on a graph. While the underlying discretization and constraint regularization employed are designed to be compatible with loop quantization techniques \cite{Thiemann_MCQGR}, the discussion remains purely classical. Our motivation stems from the framework of \textit{effective dynamics}, which has recently gained firmer theoretical support through the results of \cite{Symmetry_Restriction}.

In Sec.~\ref{s4.1:algorithm}, we outline the current status of the \textit{effective dynamics program} in LQG. This program provides a formal procedure for deriving practical models in line with those of Loop Quantum Cosmology (LQC), where significant success has already been achieved \cite{Bojowald_LQC,Bojowald_2010,Ashtekar_2011}. A key distinction in earlier work was whether to adopt a ``restrict-then-regularize'' or ``regularize-then-restrict'' approach, with both yielding effective models corresponding to complicated solutions in GR\footnote{Famously, both directions were found to agree in the special case of isotropic  $k=0$ cosmology \cite{Ashtekar_Improved_Dynamics,Dapor_2018_1} as well as Bianchi I \cite{AWE_09_Bianchi_I,GQMM20_DL_BI}. However, this is not the case anymore for general solutions such as curved cosmology (i.e., $k\neq 0$), gravitational waves, black holes etc.}. In this work, we adopt the latter approach. Moreover, ensuring physical viability requires that one incorporates a suitable improved dynamics ($\bar{\mu}$-like) scheme, which has become standard in LQC \cite{Ashtekar_Improved_Dynamics}. The outcome of this procedure is a symmetry-adapted phase space along with a consistent Hamiltonian governing the dynamics. This section is intended to remain fairly general and consolidate recent insights from across the field.

Sec.~\ref{s4.2:spherical_graphs} initiates the implementation of this program by introducing a notion of spherical symmetry for discretized gravity. We define a family of \textit{spherical graphs} $\gamma$, parameterized by a discretization parameter $n\in\N$, along with their associated dual cell complexes $\gamma^*$. Employing the framework of Sec.~\ref{s3.3:symmetry_restriction}, we then construct a discrete spherical symmetry group $\Phi_\gamma^\Psi < \Symp(\M_\gamma,\omega_\gamma)$. 

Sec.~\ref{s4.3:phase_space} identifies the set of fixed points under the action of the group $\Phi_\gamma^\Psi$, as well as a physically-relevant subspace thereof. This latter phase space offers better analytical tractability and links naturally to the continuum phase space $\M_{\text{AB}}$. 

In Sec.~\ref{s4.4:symplectic_structure}, we compute the symplectic structure on the physically-relevant phase space, deriving a Poisson algebra that reduces to its continuum counterpart only in the limit $n\to\infty$. The corrections arising at finite lattice spacing are of particular interest, as the introduction of an improved dynamics scheme is expected to leave an imprint on those.

Finally, Sec.~\ref{s4.5:scalar_constraint} demonstrates that the regularized scalar constraint \eqref{C_Thiemann} is $\Phi_\gamma^\Psi$-invariant, thereby justifying the use of Theorem~\ref{thm}. Afterwards, the constraint can be restricted to the physically-relevant phase space, yielding an effective Hamiltonian for spherically symmetric gravity on a graph. Although the resulting expression --- computed using the form of $C_\gamma[N]$ presented in Appendix~\ref{app:poly_scalar_constr} --- is lengthy, it can be obtained explicitly and provides a concrete foundation for the physical applications to be explored in subsequent work.

\subsection{Effective Dynamics Program: An Algorithmic Approach \label{s4.1:algorithm}}
We now present a step-by-step guide for the extraction of effective dynamics in loop-inspired models of quantum general relativity. The steps outlined below are understood to be necessary for constructing a viable effective theory that (a) maintains a clear relation to a full quantization of gravity, (b) faithfully captures the dynamics of particular solutions while preserving their symmetries, (c) remains computationally tractable, and (d) yields physically meaningful predictions. While the individual elements of this framework are by no means novel, our goal here is to distill a modern, actionable summary of procedures that have appeared in various forms throughout the literature \cite{Symmetry_Restriction,AWE_09_Bianchi_I,Vandersloot_07,Klaus_k=1,Dapor_2018_1,Dapor_2018_2,DL20,ADL20,Emergent_de_Sitter_detailed,Kelly_2020_1,Kelly_2020_2,Gambini_1998,Gambini_2020,BMM19c,BMM21,GQMM20_DL_BI}, and to consolidate them into a single, unified presentation (see Fig.~\ref{fig:algorithm}):
\begin{enumerate}
    \item {\bf Identification of the symmetry group:} Given a particular set of solutions of interest, choose a symmetry group $\Psi < \Diff(\Sigma)$ that characterizes the restricted physical system. Lift this symmetry group to a group of symplectomorphisms $\Phi^\Psi$ acting on the continuum phase space.
    
    \item {\bf Discretization:} Choose a family of graphs or lattices $\gamma$ parameterized by a regulator $\varepsilon = 1/n\ll 1$ and adapted to the symmetry group $\Phi^\Psi$, such that there exists a maximal subgroup $\Phi_\gamma^\Psi < \Phi^\Psi$ whose associated diffeomorphisms $\Psi_\gamma < \Psi$ leave $\gamma$ invariant. Truncate the continuum phase space $\M$ to the phase space $\M_\gamma$ defined over $\gamma$ (in terms of holonomies and gauge-covariant fluxes in the case of GR). 
    
    \item {\bf Invariant subspace:} Within the truncated phase space, identify the subspace of {\it all} points invariant under the action of $\Phi_{\gamma}^\Psi$, i.e. $\overline{\M}_\gamma\subset \M_{\gamma}$. The flow of suitable phase space functions can be shown to remain in this subspace, so it is advantageous to choose a convenient parameterization of $\overline{\M}_\gamma$ (typically with finitely many degrees of freedom) so as to ease the computations which follow.

    \item {\bf Restriction of symplectic structure:} The symplectic structure $\omega_{\gamma}$ with which one has endowed $\M_{\gamma}$ needs to be restricted to the $\Phi_\gamma^\Psi$-invariant subspace $\overline{\M}_\gamma$. Check that with this restricted symplectic structure, $\overline{\omega}_\gamma\coloneqq\omega_\gamma|_{\overline{\M}_{\gamma}}$, one indeed has a symplectic manifold $(\overline{\M}_\gamma, \overline{\omega}_\gamma)$. If this is the case, one is justified in using $\overline{\omega}_\gamma$ to evaluate Poisson brackets with respect to the chosen parameterization of $\overline{\M}_\gamma$ \cite{Symmetry_Restriction}.
    
    \item {\bf Truncation of the dynamics:} Truncate the phase space function of interest $H$ (e.g. a Hamiltonian or constraint) to the graph $\gamma$ by means of a preferred regularization, yielding a new function $H_\gamma$. Then, check that $H_\gamma$ is $\Phi_\gamma^\Psi$-invariant. In the affirmative case, one can obtain an effective expression by restriction to $\overline{\M}_\gamma$, i.e. $\overline{H}_\gamma\coloneqq H_\gamma|_{\overline{\M}_\gamma}$. This yields the regularized function in terms of the chosen parameterization preferred for explicit computations.
    
    \item {\bf Effective approximations:} Optionally, one might be interested in simplifying the analytic expression of both $\overline{\omega}_\gamma$ and $\overline{H}_\gamma$. Due to their combinatorial form, discrete structures such as $\gamma$ lead (in)famously to involved expression when integrals are replaced by sums and integrals over the various vertices and edges of the graph. Hence, the resulting expression might make precise manipulations cumbersome. Strategies that can be employed to provide \textit{effective approximations} $\overline{H}_{\text{eff}}\approx \overline{H}_\gamma$, $\overline{\omega}_{\text{eff}}\approx \overline{\omega}_\gamma$ include: expanding $\overline{H}_\gamma$ and $\overline{\omega}_\gamma$ in terms of a power series over $\varepsilon$ \cite{ADL20}, or sampling $\overline{H}_\gamma$ and $\overline{\omega}_\gamma$ over $\overline{\M}_\gamma$ to obtain simpler fitted functions \cite{Klaus_k=1}.
    
    \item {\bf Physical regularization $\varepsilon\mapsto\bar\mu$:} The resulting expression $\overline{H}_{\text{eff}}$ famously depends on the denseness of the graph through $\varepsilon$. For finite regulators, evolving a system into classical singularities is bound to bring discretization artifacts to the forefront and affect physical predictions. While eventually, the framework of the renormalization group \cite{Saeed_RG_1,Saeed_RG_2,LLT18,Bahr_2022} is envisioned to settle this caveat, independent investigations in LQC have found a strategy to avoid nonphysical results due to fixed graphs by virtue of improved dynamics schemes \cite{APS_06_1,APS_06_2,Ashtekar_Improved_Dynamics}. Hence, one needs to identify a suitable \textit{physical regulator} $\bar\mu=\bar\mu(m)$, with $m$ being a parameterization of the Dirac observables on $\overline{\M}_\gamma$. Finally, one obtains the {\it LQC-effective-dynamics} by performing the substitution $\varepsilon\to \bar{\mu}$ in both $\overline{H}_{\text{eff}}$ and $\overline{\omega}_{\text{eff}}$.
\end{enumerate}

After executing the whole program, one is ensured to describe a physically meaningful, consistent, regularized version of gravity envisioned to capture the main effects of loop-inspired quantization. It is important to note that this program has evolved over several decades in order to attain the full form presented here. It is therefore reasonable to expect that deviations from the exact prescription at different points in the preceding algorithm will result in corrections with varying degrees of relevance. For this reason, one is justified in maintaining hope that legitimate dynamical predictions can be extracted from this framework, even if executed only partially.

\begin{figure}
    \centering
    \includegraphics[width=\linewidth]{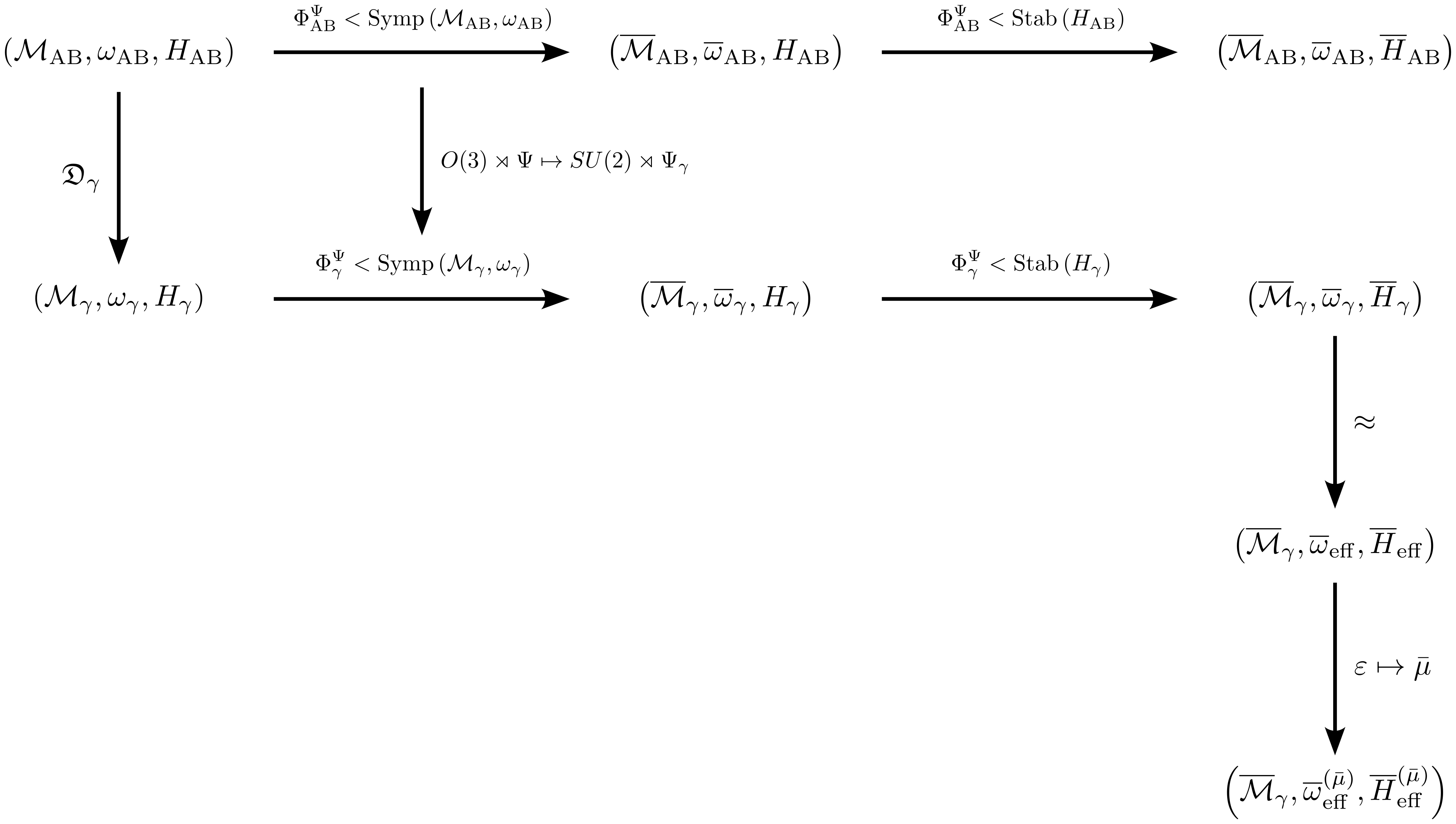}
    \caption{
        Diagrammatic representation of the effective dynamics algorithm outlined in Sec.~\ref{s4.1:algorithm}. The top row depicts the symmetry restriction procedure in the continuum theory (see Sec.~\ref{s2.3:symmetry_restriction}), while the second row illustrates the analogous procedure in the discretized theory (see Sec.~\ref{s3.3:symmetry_restriction}). In transitioning from the continuum to the discretized theory, one must choose a discretization map $\mathfrak{D}_\gamma$, a regularized Hamiltonian $H_\gamma$, and translate the continuum symmetry group $\Phi^\Psi_\text{AB} < O(3)\rtimes\Diff(\Sigma)$ to a discrete analog $\Phi^\Psi_\gamma < SU(2) \rtimes \mathrm{Stab}(\gamma,\gamma^*)$, where $\mathrm{Stab}(\gamma,\gamma^*) < \Diff(\Sigma)$ is the set of all graph-preserving diffeomorphisms for the chosen spatial discretization. Once the truncated system is restricted to $(\overline{\M}_\gamma,\overline{\omega}_\gamma,\overline{H}_\gamma)$, one may choose to impose controlled approximations and construct an effective system $(\overline{\M}_\gamma, \overline{\omega}_\text{eff},\overline{H}_\text{eff})$. Lastly, one replaces the unphysical lattice spacing with a physical regulator, $\varepsilon \mapsto \bar{\mu}$, in order to obtain the final effective theory.
    }
    \label{fig:algorithm}
\end{figure}

\subsection{Spherical Graphs and Their Symmetries}\label{s4.2:spherical_graphs}

Let us suppose that the spatial manifold $\Sigma$ is spherically symmetric, and consider a spherical volume contained within some finite fiducial radius $r_0 > 0$. We discretize the standard spherical coordinate axes by segmenting them into pieces of length
\begin{align}\label{edge_length}
    \varepsilon_r = r_0/n, && \varepsilon_\theta = \pi/n, && \varepsilon_\varphi = 2\pi/n,
\end{align}
respectively, for some discretization parameter $n\in\N$. A corresponding set of vertices is then given by\footnote{Note that we require the $\theta$-coordinate to run from $0$ to $\pi$ so that $\gamma^v$ contains both the north ($\jmath_\theta = 1$) and south ($\jmath_\theta = n+1$) poles.}
\begin{align}
    \gamma^v = \left\{\left(\jmath_r\,\varepsilon_r,(\jmath_\theta-1)\,\varepsilon_\theta,\jmath_\varphi\,\varepsilon_\varphi\right)\in\Sigma\,\big|\,(\jmath_r,\jmath_\theta,\jmath_\varphi)\in\mathcal{J}_n\times\mathcal{J}_{n+1}\times\mathcal{J}_n\right\},
\end{align}
in which $\mathcal{J}_k \coloneqq \{1,\ldots,k\}\subset \N$ for any $k\in\N$. Given the underlying symmetries, this set is subject to periodic boundary conditions on the $\varphi$-coordinate as well as $n$-fold degeneracy of points lying along the polar axis ($\jmath_\theta = 1,n+1$) with respect to $\jmath_\varphi$:
\begin{align}
    \left(\jmath_r\,\varepsilon_r, \jmath_\theta\,\varepsilon_\theta, \jmath_\varphi\,\varepsilon_\varphi\right) &= \big(\jmath_r\,\varepsilon_r, \jmath_\theta\,\varepsilon_\theta, \left(\jmath_\varphi+n\right)\varepsilon_\varphi\big), \label{v_phi_periodic} \\
    \left(\jmath_r\,\varepsilon_r, 0, \jmath_\varphi\,\varepsilon_\varphi\right) &= \left(\jmath_r\,\varepsilon_r, 0,\jmath_\varphi'\,\varepsilon_\varphi\right),\label{v_theta_NP_degeneracy} \\
    \left(\jmath_r\,\varepsilon_r, n\,\varepsilon_\theta, \jmath_\varphi\,\varepsilon_\varphi\right) &= \left(\jmath_r\,\varepsilon_r, n\,\varepsilon_\theta, \jmath_\varphi'\,\varepsilon_\varphi\right),\label{v_theta_SP_degeneracy}
\end{align}
for any $\jmath_r,\jmath_\varphi,\jmath_\varphi'\in\mathcal{J}_n$ and $\jmath_\theta\in\mathcal{J}_{n+1}$.

\begin{figure}[t]
    \centering
    
    \begin{subfigure}{0.49\textwidth}
        \includegraphics[width=\textwidth]{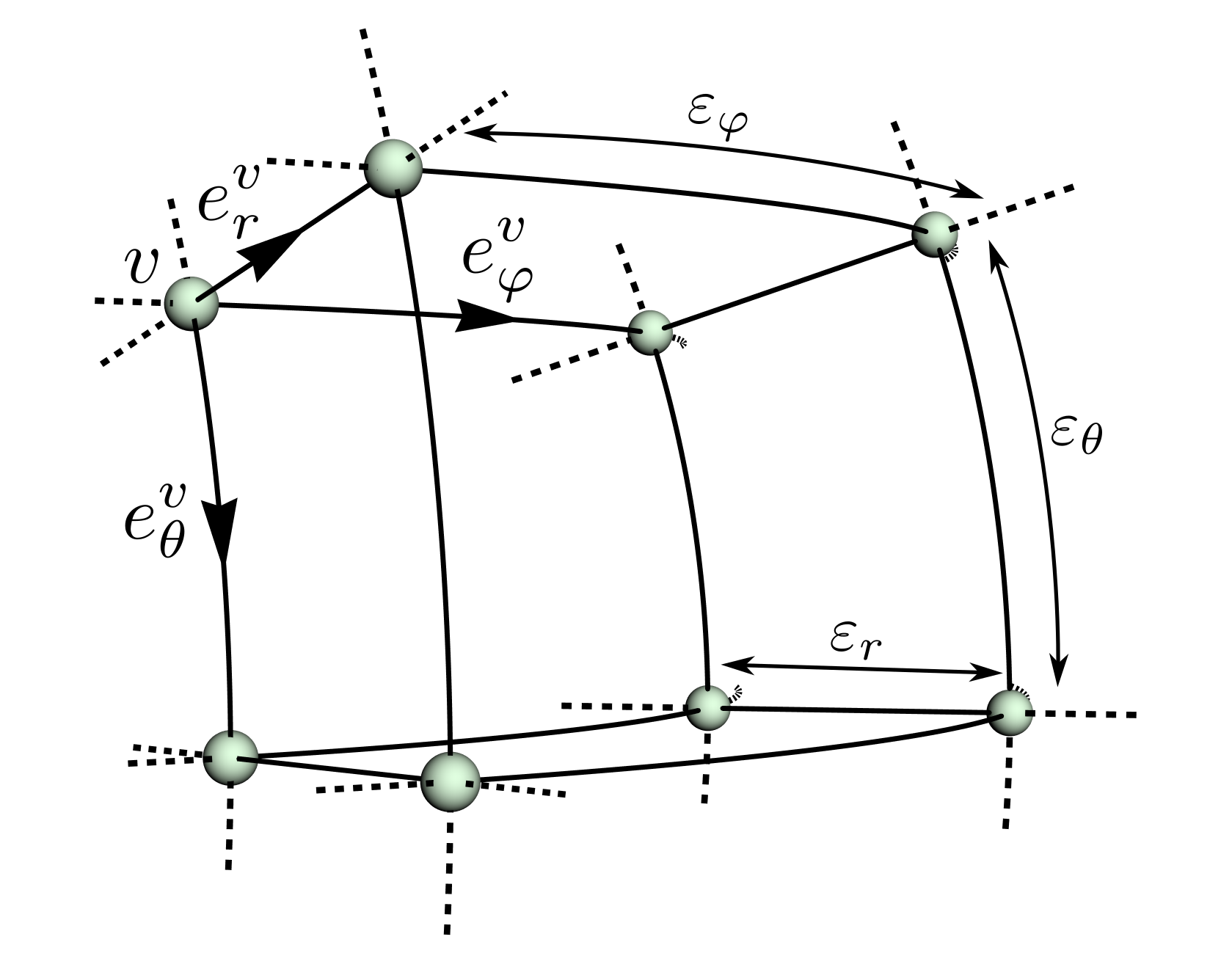}
    \end{subfigure}
    \hfill
    \begin{subfigure}{0.49\textwidth}
        \includegraphics[width=\textwidth]{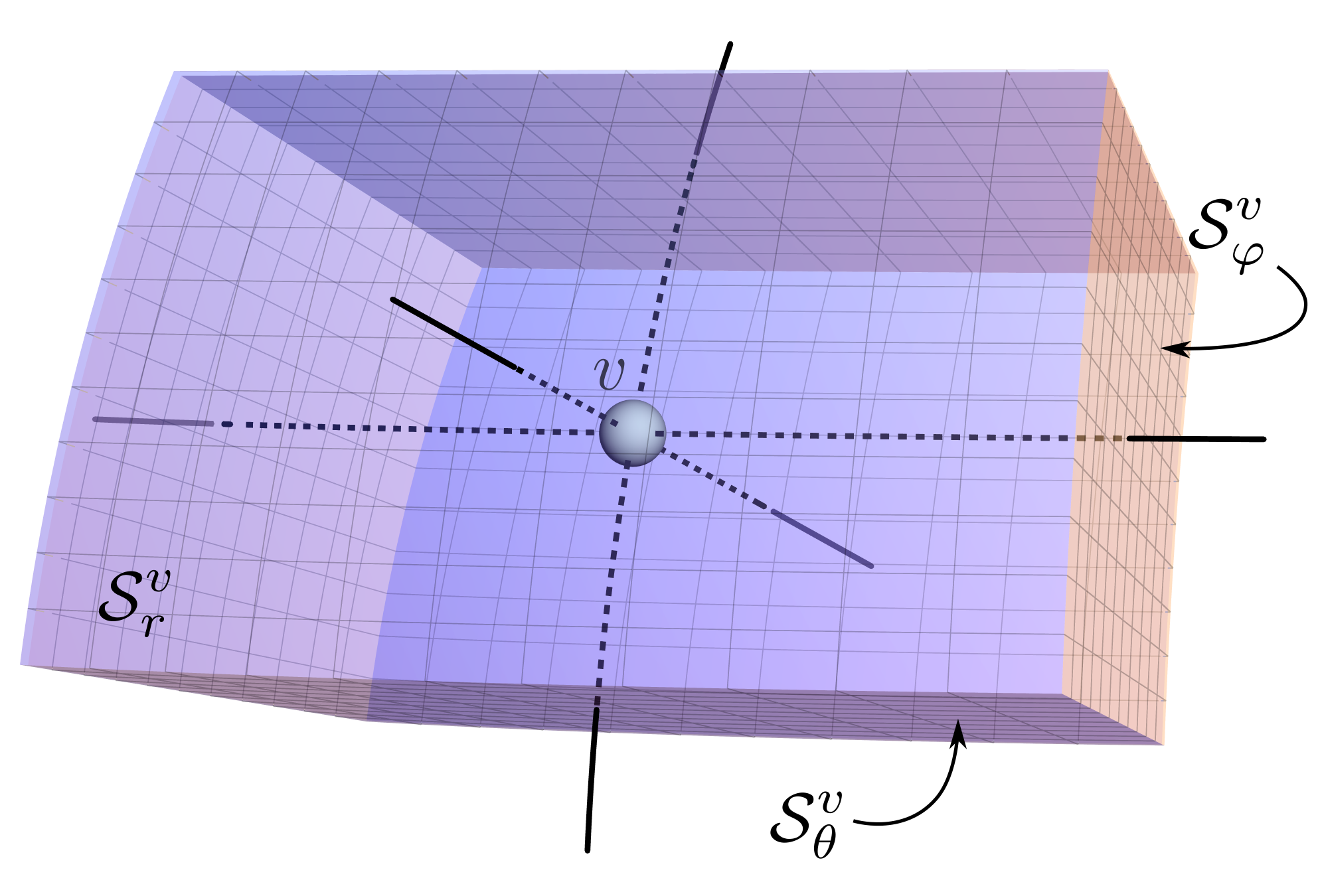}
    \end{subfigure}

    \caption{Graphic depiction of the edges (left) and surfaces (right) comprising the building blocks of a spherical graph $\gamma$ and its dual $\gamma^*$. On the left, arrows have been placed on the three positively-oriented edges defined at the vertex $v\in\gamma^v$, and the discretization parameter $n\in\N$ is encoded by the coordinate lengths $\varepsilon_i$ defined in \eqref{edge_length}. On the right, the three positively-oriented surfaces associated with the vertex $v$ (situated in the centre of the cubic region) are labelled, and we note that each $\mathcal{S}^v_i$ has a coordinate area $\varepsilon_j\varepsilon_k$ with $i\neq j \neq k$.}
    \label{fig:spherical_graph_cells}
\end{figure}

For each vertex $v\in\gamma^v$, we define a set of edges and surfaces according to \eqref{edge_param}--\eqref{surface_param}, with displacements given explicitly by 
\begin{align}
    \mu_i(s) = s \left(\delta_{i,r}\,\varepsilon_r,\delta_{i,\theta}\,\varepsilon_\theta, \delta_{i,\varphi}\,\varepsilon_\varphi\right),
\end{align}
for any $i\in\mathcal{I}^+ = \{r,\theta,\varphi\}$ and, as before, we shall typically write $\mu_i$ to denote the full edge-displacement $\mu_i(1)$. As a result, each edge $e^v_i\in\gamma$ has coordinate length $\varepsilon_i$, as given in \eqref{edge_length}, while the corresponding dual surface $\mathcal{S}^v_i\in\gamma^*$ has coordinate area $\varepsilon_j\varepsilon_k$ (where $\epsilon_{ijk}=1$). From hereon out, we shall refer to such a graph $\gamma$ as a \textit{spherical graph} (see Figs.~\ref{fig:spherical_graph_cells} \& \ref{fig:spherical_graph}).

\begin{figure}[t]
    \centering
    
    \begin{subfigure}{0.45\textwidth}
        \includegraphics[width=\textwidth]{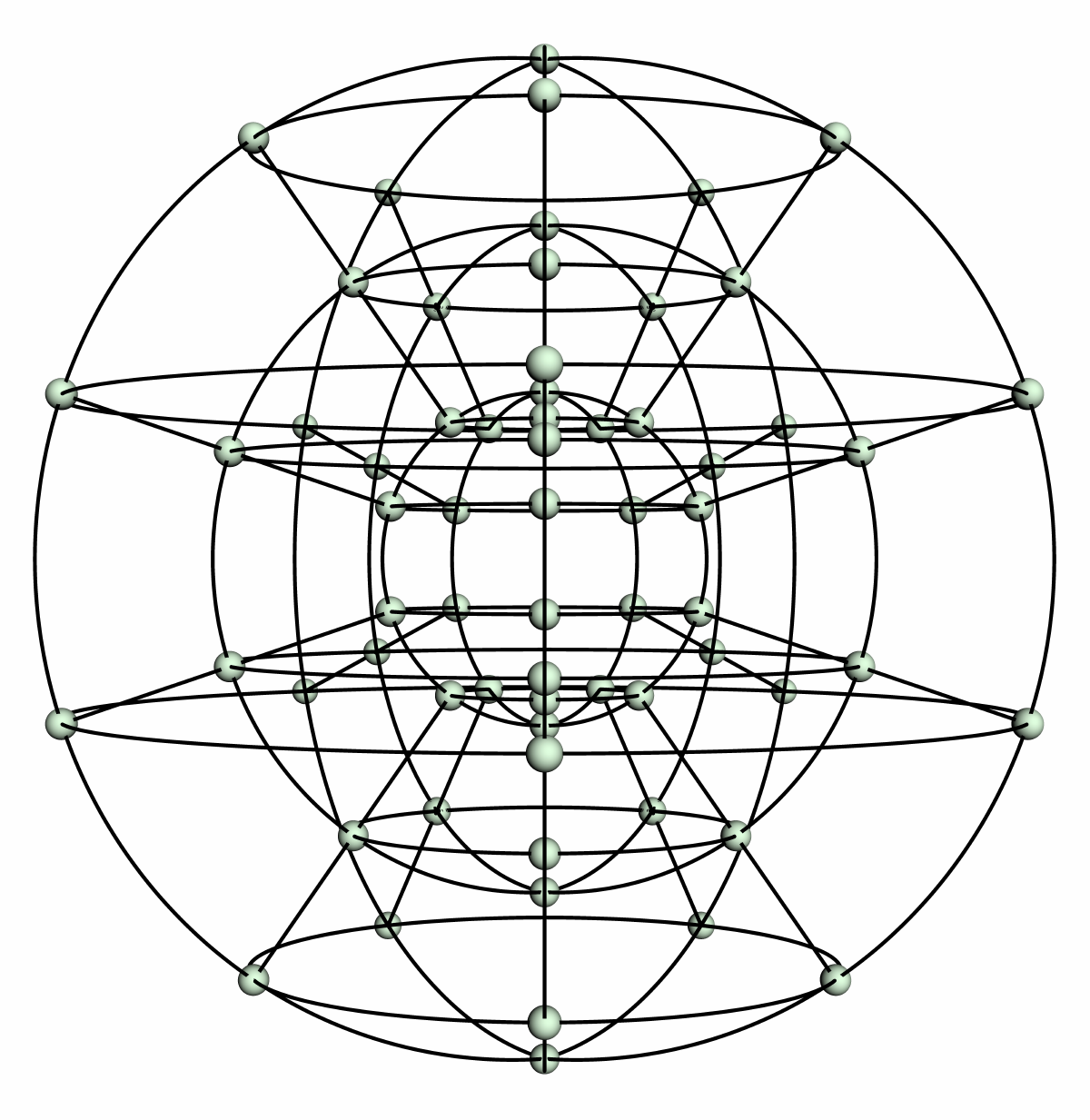}
    \end{subfigure}
    \hfill
    \begin{subfigure}{0.45\textwidth}
        \includegraphics[width=\textwidth]{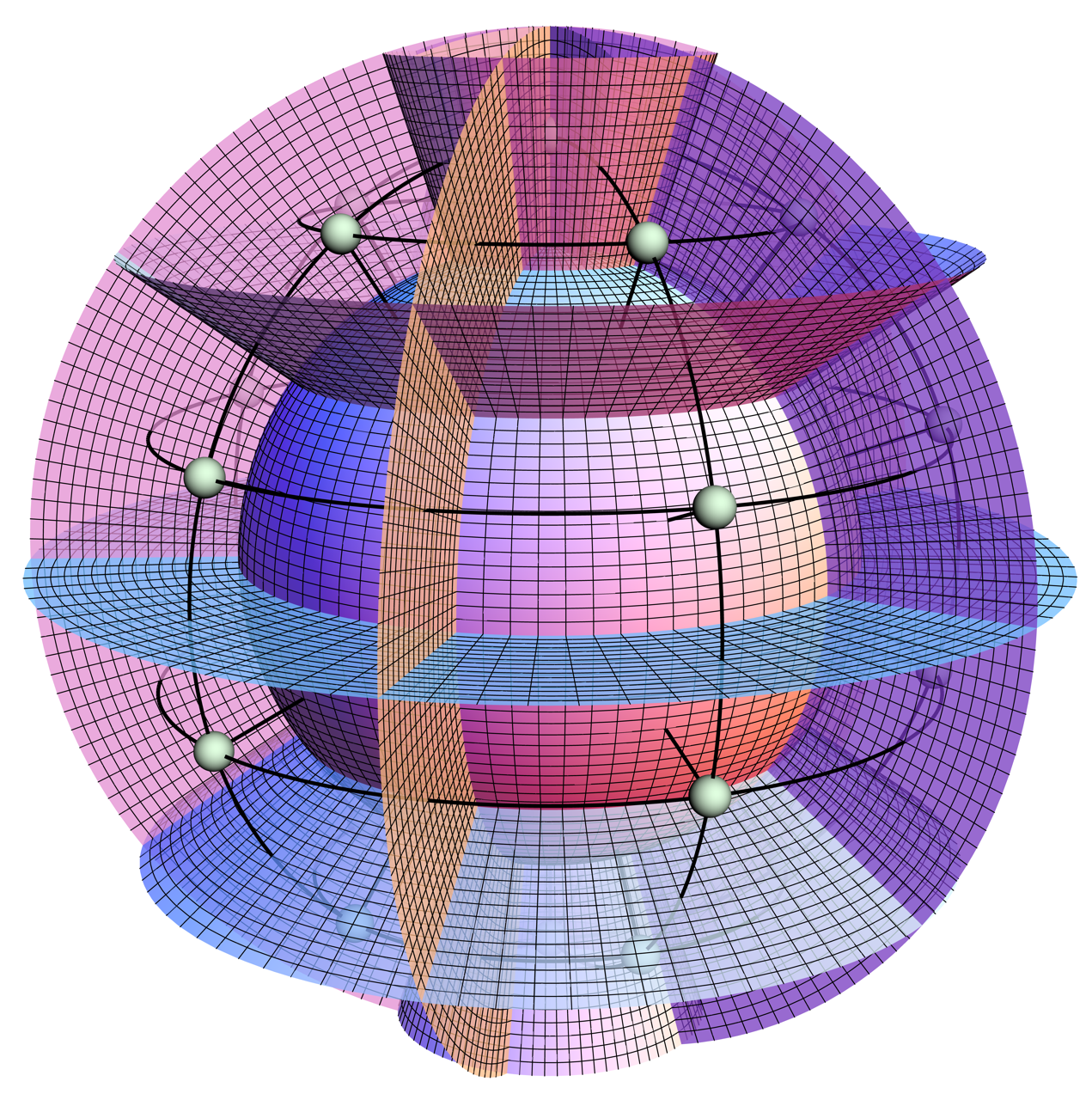}
    \end{subfigure}

    \caption{Visualization of a spherical graph $\gamma$ (left) and the associated dual cell complex $\gamma^*$ (right) embedded in $\R^3$. Each green point represents a vertex $v\in \gamma^v$, which are connected by edges adapted to the spherical coordinate directions. For each edge, there is exactly one corresponding surface in the dual graph.}
    \label{fig:spherical_graph}
\end{figure}

\subsubsection{Graph-Preserving Diffeomorphisms \label{s4.2.1:diffeos}}
Identifying the group of graph-preserving diffeomorphisms $\mathrm{Stab}_\Psi(\gamma,\gamma^*)$ associated with a spherical graph naturally leads to the consideration of finite subgroups of $O(3)$. The most general of these subgroups arise as particular realizations of finite \textit{Coxeter groups} --- discrete groups generated entirely by reflections that encode the symmetries of regular geometric objects such as polygons and polyhedra \cite{Coxeter_2013}. These groups are well understood and fully classified up to isomorphism, with each being decomposable into a direct product of up to ten irreducible types \cite{Humphreys_1990,Bourbaki_2008}. Among these, only five admit faithful realizations as finite subgroups of $O(3)$: those corresponding to reflections through a single plane, dihedral symmetries of a regular polygon, and the full symmetry groups of the tetrahedron, cube (or octahedron), and dodecahedron (or icosahedron) \cite{Johnson_2018}. To facilitate the construction of $\mathrm{Stab}_\Psi(\gamma,\gamma^*)$ in terms of these groups, we begin by reviewing a selection of foundational concepts from Coxeter theory relevant to the analysis that follows.

In general, a \textit{Coxeter group} of rank $k\in\N$ is a discrete group $W$ generated by a set of $k$ elements $\mathscr{R}\subset W$, referred to as \textit{reflections}. Each such reflection is subject to a set of product/composition relations of the form $(\rho\circ\rho')^{m(\rho,\rho')} = \mathrm{Id}$, where $m(\rho,\rho) = 1$ and $m(\rho,\rho')\geq 2$ for distinct generators $\rho,\rho'\in\mathscr{R}$ \cite{Coxeter_2013,Humphreys_1990,Bourbaki_2008}. For our purposes, these relations can be interpreted as encoding the geometric angles between reflection planes in real Euclidean space --- for instance, if $\rho,\rho'$ are reflections across planes intersecting at an angle $\pi/\alpha$ for some integer $\alpha\geq 2$, then $m(\rho,\rho')=\alpha$. In cases where no such relation exists between $\rho$ and $\rho'$, one conventionally takes $m(\rho,\rho') = \infty$\footnote{For example, let $\rho,\rho'$ be reflections across two distinct hyperplanes. If the angle between these hyperplanes is irrational with respect to $\pi$, then $m(\rho,\rho') = \infty$.}.The full structure of $W$ is conveniently captured by means of a group \textit{presentation}, which we shall write as
\begin{align}
    W = \bigl\langle\mathscr{R}\,\big|\,(\rho\circ\rho')^{m(\rho,\rho')} = \mathrm{Id}\bigr\rangle.
\end{align}
In keeping with standard conventions, only those relations for which $m(\rho,\rho')\geq 3$ are listed explicitly, as the cases $m(\rho,\rho)=1$ (involution property of all generators) and $m(\rho,\rho') = 2$ (commuting generators) are implicitly understood.

Coxeter theory provides a rigorous framework for classifying the symmetry groups of discrete geometric structures, particularly those admitting faithful representations in $\R^3$. In the present context, it allows for a systematic identification of the graph-preserving diffeomorphism group $\mathrm{Stab}_\Psi(\gamma,\gamma^*)$ by reducing the problem to the construction of a small set of reflections compatible with the underlying structure of the graph. Indeed, any Coxeter group acting faithfully on a three-dimensional Euclidean space is of rank at most three, and therefore admits at most three generating reflections \cite{Humphreys_1990}. As a consequence, identifying $\mathrm{Stab}_\Psi(\gamma,\gamma^*)$ reduces to selecting no more than three distinct reflections that preserve the structure of the graph and its dual. On the other hand, by choosing a finite Coxeter group of interest \textit{a priori}, one may tailor the construction of the graph to align with the geometry associated with that group, thereby ensuring that its symmetries are inherited by the resulting spatial discretization. Such an approach not only trivializes the identification of $\mathrm{Stab}_\Psi(\gamma,\gamma^*)$, but also provides a systematic pathway to alternative realizations of spherical graphs beyond those considered in the present work.

Returning now to our spherical graph construction, a generating set for the group $\mathrm{Stab}_\Psi(\gamma,\gamma^*)$ can be formed by combining the non-trivial reflections that preserve two particular subgraphs (see Fig.~\ref{fig:spherical_subgraphs}):
\begin{align}\label{subgraph_def}
    \gamma^\varphi \coloneqq \left\{e^v_\varphi\in\gamma\,\big|\,\jmath_r,\jmath_\theta = \text{const}\right\}, && \gamma^\pm \coloneqq \left\{e^v_r\in\gamma\,\big|\,\jmath_\theta = 1,n-1\right\}.
\end{align}
For any fixed $\jmath_\theta\notin\{1,n-1\}$, the subgraph $\gamma^\varphi$ forms a closed cycle of $n$ vertices connected by $n$ edges of equal length, and is therefore topologically equivalent to a regular $n$-gon. As such, is admits $2n$ discrete symmetry transformations: $n$ reflections and $n$ rotations \cite{Humphreys_1990}. All of these symmetries can be generated by two reflections which, on the full edge-set, are defined as maps $\rho_\varphi,\sigma_\varphi\in\mathrm{Aut}(\gamma^v\times\mathcal{I})$ acting via
\begin{align}\label{refl_phi}
    \rho_\varphi(v,i) = \big(v + (n - 2\jmath_\varphi)\mu_\varphi, s_\varphi(i)\big), && \sigma_\varphi(v,i) = \big(v + (1-2\jmath_\varphi)\mu_\varphi, s_\varphi(i)\big),
\end{align}
where $v = (\jmath_r\varepsilon_r,\jmath_\theta\varepsilon_\theta,\jmath_\varphi\varepsilon_\varphi)$ and $s_\varphi(i) \coloneqq i\,(1 - 2\delta_{|i|,\varphi})$ reverses the orientation of the edges $e^v_{\pm\varphi}$ while leaving the orientations of all other edges unchanged.

\begin{figure}
    \centering
    
    \begin{subfigure}{0.5\textwidth}
        \includegraphics[width=\textwidth]{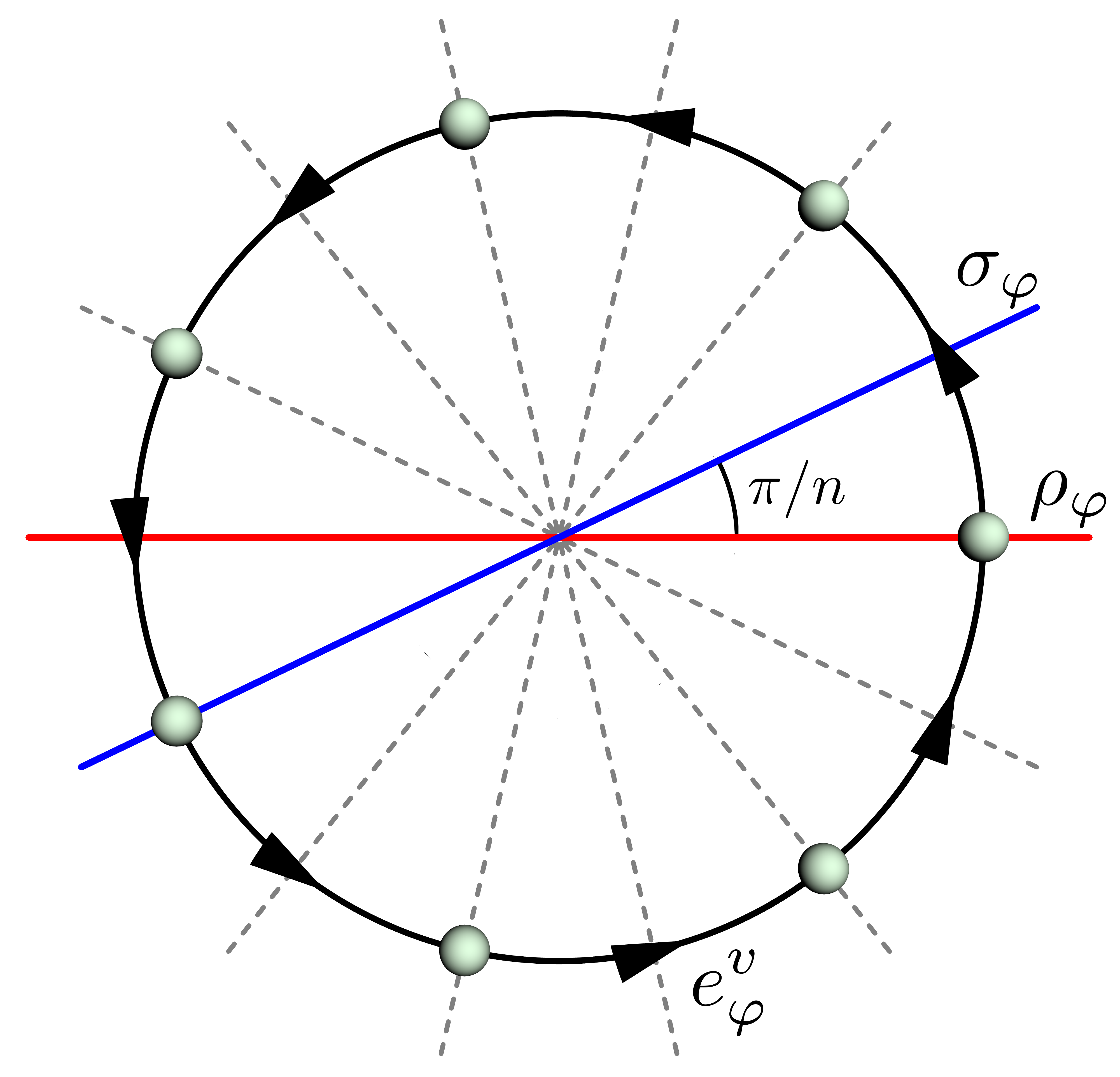}
        \caption{$\gamma^\varphi$}
        \label{fig:gamma_phi}
    \end{subfigure}
    \hfill
    \begin{subfigure}{0.4\textwidth}
        \includegraphics[width=\textwidth]{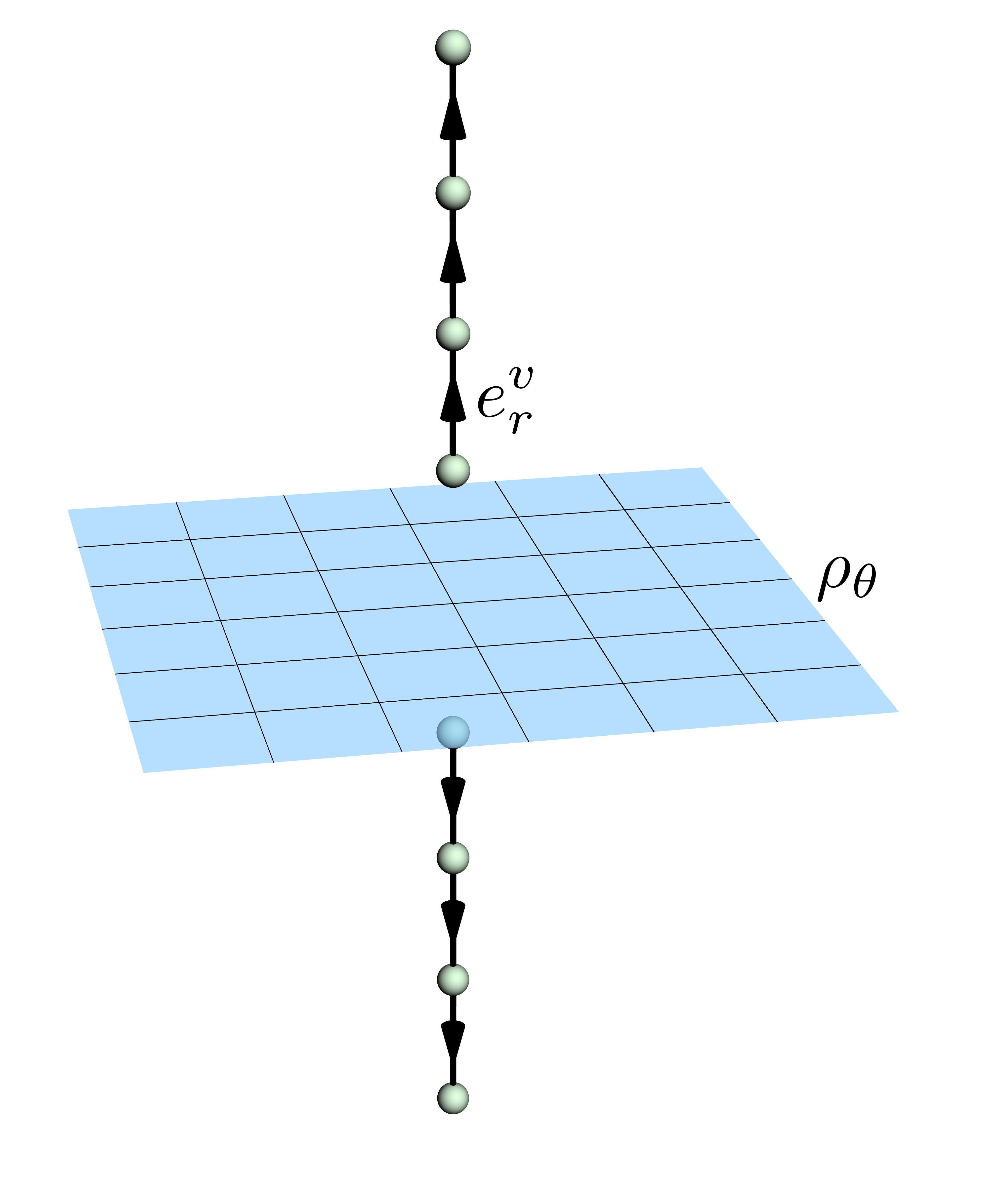}
        \caption{$\gamma^\pm$}
        \label{fig:gamma_pm}
    \end{subfigure}

    \caption{Visual representations of the azimuthal (left) and polar (right) subgraphs defined in Eq. \eqref{subgraph_def}. In Fig.~\ref{fig:gamma_phi}, the red and blue lines correspond to the ``mirrors'' (i.e., axes of reflection) associated with the symmetry transformations $\rho_\varphi$, $\sigma_\varphi$, respectively (see \eqref{refl_phi}). The dashed lines indicate mirrors associated with all other symmetries of the planar graph, each of which can be obtained from $\rho_\varphi$ or $\sigma_\varphi$ by means of the rotation \eqref{R_phi}. Fig.~\ref{fig:gamma_pm} contains a 3-dimensional mirror corresponding to the reflection $\rho_\theta$ (see \eqref{refl_theta}), which exchanges the two disconnected components of the subgraph.}
    \label{fig:spherical_subgraphs}
\end{figure}

Because of the degeneracy of vertices lying on the polar axes, the reflections $\rho_\varphi$ and $\sigma_\varphi$ act trivially on the subgraph $\gamma^\pm$. However, an additional symmetry does arises in this case, corresponding to reflection through the equatorial plane --- the direct analog of the map $\rho_\theta$ that played a key role in the continuum treatment of Sec.~\ref{s2.4:spherical_symmetry}. Maintaining the notation employed in the continuum setting, this reflection is defined in the discrete context by
\begin{align}\label{refl_theta}
\rho_\theta(v,i) = \big(v + (n - 2\jmath_\theta)\,\mu_\theta,s_\theta(i)\big),
\end{align}
where $s_\theta(i) \coloneqq i\,(1- 2\delta_{|i|,\theta})$ reverses the orientation of the edges $e^v_{\pm\theta}$ and keeps all others fixed. Having identified three independent reflection symmetries, we can now characterize the full group of spherical graph-preserving diffeomorphisms via the presentation
\begin{align}\label{Stab_Psi_spherical}
    \mathrm{Stab}_\Psi(\gamma,\gamma^*) = \left\langle\rho_\theta,\rho_\varphi,\sigma_\varphi\,\big|\,\left(\rho_\varphi\circ\sigma_\varphi\right)^n = \mathrm{Id}\right\rangle \cong \Z_2\times D_n,
\end{align}
where the $\Z_2$-factor corresponds to the single reflection symmetry of the polar subgraph $\gamma^\pm$, while $D_n\cong \Z_n\rtimes\Z_2$ (the \textit{dihedral group} of order $2n$) characterizes the full set of symmetries of $\gamma^\varphi$.

Since our primary interest lies in the group of \textit{proper} graph-preserving diffeomorphisms \eqref{Psi_gamma}, it remains to isolate the subgroup $\Psi_\gamma < \mathrm{Stab}_\Psi(\gamma,\gamma^*)$ consisting of finite rotations alone. As in the general context of Coxeter theory, all such rotations arise as the composition of an even number of reflections \cite{Johnson_2018}. In the case of \eqref{Stab_Psi_spherical}, the rotation subgroup can be defined as
\begin{align}\label{Psi_gamma_spherical}
    \Psi_\gamma = \left\langle R_\pi, R_\varphi\,\big|\,R_\varphi^n = \mathrm{Id}\right\rangle \cong \Z_2\ltimes\Z_n,
\end{align}
where $R_\pi \coloneqq \rho_\theta\circ\rho_\varphi$ and $R_\varphi \coloneqq \sigma_\varphi\circ\rho_\varphi$, with respective inverses given by $R_\pi^{-1} = \rho_\varphi\circ\rho_\theta = R_\pi$ and $R_\varphi^{-1} = \rho_\varphi\circ\sigma_\varphi$. More concretely, for any $(v,i)\in\gamma^v\times\mathcal{I}$, the rotations generating $\Psi_\gamma$ act as
\begin{align}\label{R_phi}
    R_\varphi(v,i) = \left(v + \mu_\varphi,i\right)
\end{align}
and
\begin{align}\label{R_pi}
    R_\pi(v,i) = \big(v + (n-2\jmath_\theta)\,\mu_\theta + (n-2\jmath_\varphi)\,\mu_\varphi, (s_\theta\circ s_\varphi)(i)\big),
\end{align}
where $(s_\theta\circ s_\varphi)(i) = i\,(1 - 2\delta_{|i|,\theta} - 2\delta_{|i|,\varphi})$ simultaneously reverses the orientation of $e^v_{\pm\theta}$ and $e^v_{\pm\varphi}$, thereby preserving the overall orientation of the graph. Geometrically, $R_\varphi$ corresponds to a rotation about the polar axis by an angle $\pi/n$, whereas $R_\pi$ represents a rotation by $\pi$ about an axis lying in the equatorial plane and passing through $\varphi=0,\pi$.

To summarize, we have shown that the graph-preserving diffeomorphism group $\mathrm{Stab}_\Psi(\gamma,\gamma^*)$ associated with a spherical graph admits a complete characterization in terms of a minimal generating set of discrete reflections, giving rise to the Coxeter group structure $\mathrm{Stab}_\Psi(\gamma,\gamma^*)\cong \Z_2\times D_n$. The rotation subgroup $\Psi_\gamma \cong \Z_2 \ltimes \Z_n$ captures the proper automorphisms of the graph and serves as the diffeomorphism group relevant for symmetry restriction. In the following section, we complete the construction of the full symmetry group $\Phi_\gamma^\Psi < \Symp(\M_\gamma,\omega_\gamma)$ by determining the internal gauge transformations needed to compensate the action of $\Psi_\gamma$ on the discretized phase space.

\subsubsection{Gauge Transformations \label{s4.2.2:gauge}}

Continuing with the procedure laid out in Sec.~\ref{s3.3:symmetry_restriction}, we now identify the $SU(2)$-gauge transformations associated with each proper graph-preserving diffeomorphism in \eqref{Psi_gamma_spherical}. In particular, we construct the $SO(3)$ rotation matrices associated with the generators \eqref{R_phi}--\eqref{R_pi} of $\Psi_\gamma$, which can then be lifted to $SU(2)$ via Fact~\ref{fact:double_cover}. This yields the two generators of $\Phi_\gamma^\Psi < SU(2)\rtimes \Psi_\gamma$, from which all other elements follow by group multiplication.

We define the orthogonal matrices $O_\varphi,O_\pi\in SO(3)$ associated with the finite rotations $R_\varphi,R_\pi\in\Psi_\gamma$ in exactly the same manner as in the continuum theory. Applying \eqref{O_psi} to the spherically symmetric AB variables \eqref{AE_bar}, we find that the internal rotations on $\M_{\text{AB}}$ which counteract the pullbacks of $\overline{E}$ by $R_\varphi$ and $R_\pi$ are given by
\begin{align}
    O_\varphi = \mathrm{Id}_{SO(3)}, && O_\pi = \mathrm{diag}(1,-1,-1),
\end{align}
respectively. Using \eqref{pi_U_components}, one can easily verify that the pre-images of these matrices under the double-cover homomorphism read
\begin{align}\label{O_spherical}
    \left(D^{(1)}\right)^{-1}(O_\varphi) = \big\{\pm\mathrm{Id}_{SU(2)}\big\}, && \left(D^{(1)}\right)^{-1}(O_\pi) = \big\{\exp(\pm\pi\tau_1)\big\}.
\end{align}
Therefore, the conditions $D^{(1)}(U_\varphi) = O_\varphi$ and $D^{(1)}(U_\pi) = O_\pi$ are satisfied by choosing
\begin{align}\label{U_spherical}
    U_\varphi = \mathrm{Id}_{SU(2)}, && U_\pi = \exp(\pi\tau_1) = 2\,\tau_1
\end{align}
as the corresponding lifts to $SU(2)$. At this point, we have obtained a complete set of generators $\{(U_\varphi,R_\varphi),(U_\pi,R_\pi)\}\in SU(2)\rtimes\Psi_\gamma$ for the spherical-graph symmetry group $\Phi_\gamma^\Psi$ defined in \eqref{Phi_gamma}, thereby marking the end of its construction.

\subsection{Spherical-Graph Phase Space \label{s4.3:phase_space}}

Our aim now is to implement Theorem~\ref{thm} on the discretized phase space $\M_\gamma$ with respect to the spherical-graph symmetry group $\Phi_\gamma^\Psi$ constructed above. To this end, we must first identify the fixed-point set of $\M_\gamma$ under the action of $\Phi_\gamma^\Psi$, i.e., the invariant submanifold
\begin{align}\label{M_bar_gamma}
    \overline{\M}_\gamma = \left\{(\overline{h}_e,\overline{P}_e)\in\M_\gamma\,\big|\,\mathcal{A}^R_{(U_\psi,\psi)}(\overline{h}_e,\overline{P}_e) = (\overline{h}_e,\overline{P}_e),\,\forall \,(U_\psi,\psi)\in\Phi_\gamma^\Psi\right\}.
\end{align}
As we shall see, however, a full characterization of $\overline{\M}_\gamma$ is not sufficient on its own to render the computations practically feasible. For this reason, our analysis proceeds in two stages. First, in Sec.~\ref{s4.3.1:invariant_submanifold}, we extract a set of necessary conditions that define $\Phi_\gamma^\Psi$-invariance within $\M_\gamma$, thereby capturing the essential features of configurations in $\overline{\M}_\gamma$. 
Then, in Sec.~\ref{s4.3.2:physical_subspace}, we further restrict to a physically-motivated subspace of $\overline{\M}_\gamma$ which retains the relevant structure while offering reasonable approximations that significantly reduce the computational complexity.

\subsubsection{Invariant Submanifold \label{s4.3.1:invariant_submanifold}}

Because each generator $(U_\varphi,R_\varphi),(U_\pi,R_\pi)$ carries a constant $SU(2)$-component, the general action \eqref{full_action_gamma} simplifies considerably when restricted to the spherical-graph symmetry group $\Phi_\gamma^\Psi$. In particular, for holonomies and fluxes $h_i(v),P_i(v)$ defined in the coordinate-adapted parameterization \eqref{edge_param}--\eqref{surface_param}, we have
\begin{align}
    \label{spherical_action}
    \mathcal{A}^R_{(U_\psi,\psi)}\big(h_i(v),P_i(v)\big) = \left(U_\psi^\dagger\,h_{\psi(i)}\big(\psi(v)\big)\,U_\psi, U_\psi^\dagger\,P_{\psi(i)}\big(\psi(v)\big)\,U_\psi\right),
\end{align}
where $\psi(v)\coloneqq \psi(v,\cdot)$ and $\psi(i)\coloneqq \psi(\cdot,i)$ denote the component-wise maps induced by $\psi\in\mathrm{Aut}(\gamma^v\times\mathcal{I})$. The action therefore factorizes into a global conjugation by $U_\psi\in SU(2)$, together with a pushforward of the edge and surface labels under $\psi\in\Psi_\gamma$.

Suppose now that $(\overline{h}_i(v),\overline{P}_i(v))\in SU(2)\times\mathfrak{su}(2)$ is $\Phi_\gamma^\Psi$-invariant, and let us write
\begin{align}
    \label{hP_bar_general}
    \overline{h}_i(v) = \exp\big[\varepsilon_i\,\overline{u}_i(v)\big], && \overline{P}_i(v) = \delta_i\,\overline{p}_i(v),
\end{align} 
for some $\mathfrak{su}(2)$-valued functions $\overline{u}_i(v),\overline{p}_i(v)$, and where $\delta_i\in\R$ is an $\varepsilon$-dependent constant. 
Because $SU(2)$ is compact and simply connected, the exponential map $\exp:\mathfrak{su}(2)\to SU(2)$ is surjective, and so this representation for the holonomies is \textit{always} possible (albeit non-unique) \cite{Brocker_Compact_Lie_Groups}. Holonomies and fluxes can therefore be treated on the same level with respect to the action \eqref{spherical_action}, in virtue of the following result: 
\begin{fact}\label{fact:adjoint}
    For any $U\in SU(2)$ and $X = X^I\,\tau_I\in\mathfrak{su}(2)$, we have
    \begin{align}
        U\exp(X)\,U^\dagger = \exp\left(\mathrm{Ad}_UX\right), && (\mathrm{Ad}_U X)^I = D^{(1)I}_{\phantom{(1)I}J}(U)\,X^J,
    \end{align}
    where $D^{(1)}$ is the double-cover homomorphism \eqref{pi_U}.
\end{fact}
\begin{proof}
    Both of these identities are well-known results, and can be established through explicit computation. First, we use the series expansion for $\exp(X)$ to obtain
    \begin{align}
        U\exp(X)\,U^\dagger = U\left(\sum_{n=0}^\infty\frac{1}{n!}X^n\right)U^\dagger = \sum_{n=0}^\infty\frac{1}{n!}\left(\mathrm{Ad}_U X\right)^n = \exp\left(\mathrm{Ad}_U X\right).
    \end{align}
    Next, we make can use of Fact~\ref{fact:double_cover} and write
    \begin{align}
        \mathrm{Ad}_U X = X^J\mathrm{Ad}_U(\tau_J) = D^{(1)I}_{\phantom{(1)I}J}(U)X^J\tau_I,
    \end{align}
    with the final equality following from \eqref{Ad_U_tau}. Therefore, we have
    \begin{align}
        (\mathrm{Ad}_U X)^I = -2\,\tr\left(\tau^I\mathrm{Ad}_U X\right) = D^{(1)I}_{\phantom{(1)I}J}(U) X^J,
    \end{align}
    as claimed.
\end{proof}

The identities in Fact~\ref{fact:adjoint} reduce the invariance condition for holonomies and fluxes of the form \eqref{hP_bar_general} to the requirement that the components of each 
$\overline{Z}_i(v)\in\{\overline{u}_i(v),\overline{p}_i(v)\}$ satisfy
\begin{align}\label{Z_bar_invariance}
    \left(O_\psi^T\right)^I_{\phantom{I}J}\overline{Z}^J_{\psi(i)}\big(\psi(v)\big) = \overline{Z}^I_i(v),
\end{align}
for all $\psi\in\Psi_\gamma$ and $(v,i)\in\gamma^v\times\mathcal{I}$, where we have used the fact that $D^{(1)}(U_\psi^\dagger) = O_\psi^T$. For convenience when working with the individual coordinates of a vertex $v\in\gamma^v$, we will adopt the notation $v = (v_r,v_\theta,v_\varphi)$ with $v_i = \jmath_i\varepsilon_i$ from hereon out.

Beginning with the generator $(U_\varphi,R_\varphi)\in\Phi_\gamma^\Psi$, the invariance condition \eqref{Z_bar_invariance} yields
\begin{align}\label{phi_invariance}
    \overline{Z}^I_i(v) = \overline{Z}^I_i(v + \mu_\varphi).
\end{align}
Thus, $\overline{Z}_i(v)$ must be invariant under discrete translations in the $\varphi$-direction, and therefore depends only on the radial and polar coordinates: $\overline{Z}_i(v) = \overline{Z}_i(v_r,v_\theta)$. Demanding invariance under the second generator $(U_\pi,R_\pi)\in\Phi_\gamma^\Psi$ then leads to the additional requirement
\begin{align}
    \overline{Z}^I_i(v) = (-1)^{1+\delta_{I,1}}\,\overline{Z}^I_{(s_\theta\circ s_\varphi)(i)}(v_r,\pi-v_\theta).
\end{align}
Accounting for specific orientations, we obtain the following edge-wise symmetry conditions:
\begin{equation}\label{theta_invariance}
    \begin{gathered}
        \overline{Z}^I_r(v_r,v_\theta) = (-1)^{1+\delta_{I,1}}\,\overline{Z}^I_r(v_r,\pi-v_\theta), \\
        \overline{Z}^I_\theta(v_r,v_\theta) = (-1)^{\delta_{I,1}}\,\overline{Z}^I_\theta\big(v_r,\pi - (v_\theta + \varepsilon_\theta)\big), \\
        \overline{Z}^I_\varphi(v_r,v_\theta) = (-1)^{\delta_{I,1}}\,\overline{Z}^I_\varphi(v_r,\pi-v_\theta),
    \end{gathered}
\end{equation}
with the latter two following from the orientation-reversal relation $\overline{Z}_{-i}(v) = \overline{Z}_i^\dagger(v-\mu_i)$.

In summary, equations \eqref{phi_invariance} and \eqref{theta_invariance} together constitute a set of necessary conditions that must be satisfied by any $\Phi_\gamma^\Psi$-invariant holonomies and fluxes. Of particular interest is the first condition in \eqref{theta_invariance}, which implies that at the equator ($v_\theta = \pi/2$), we have
\begin{equation}
    \overline{Z}_r(v_r,\pi/2) = \overline{Z}^1_r(v_r,\pi/2)\,\tau_1.
\end{equation}
As a direct consequence, both $\overline{h}_r(v)$ and $\overline{P}_r(v)$ are completely determined by a single scalar function of the radial coordinate --- for instance, $\overline{u}^1_r(v_r,\pi/2) = a_\varepsilon(v_r)$ and $\overline{p}^1_r(v_r,\pi/2) = p_a^\varepsilon(v_r)$. This directly mirrors the reduction to one-dimensional data observed in the radial components of spherically symmetric configurations in continuum GR (recall Sec.~\ref{s2.4:spherical_symmetry}). However, comparable simplifications are not immediately evident for the $\theta$- and $\varphi$-directions, nor at vertices beyond the equatorial plane. This indicates that the spherical-graph phase space $\overline{\M}_\gamma$ retains a significantly richer structure than the typical parameterization with two canonical pairs.
Because this excess structure arises from discretization-induced limitations in capturing continuous angular symmetries (e.g., $\theta$-rotations), the following section shall serve to physically motivate an additional restriction to configurations which align more closely with the continuum parameterization.

\subsubsection{Physically-Relevant Subspace \label{s4.3.2:physical_subspace}}

The invariant subspace $\overline{\M}_\gamma \subset \M_\gamma$ captures all configurations invariant under the action of the spherical-graph symmetry group $\Phi_\gamma^\Psi$. As illustrated by the preceding discussion, however, this phase space is much larger than its continuum analog $\overline{\M}_{\text{AB}}\subset\M_{\text{AB}}$, which is typically parameterized by two $\theta,\varphi$-independent canonical pairs, $(a,p_a)$ and $(b,p_b)$. 
This indicates that $\overline{\M}_\gamma$ encompasses configurations beyond those arising from direct discretizations of the spherically symmetric continuum theory.

Strictly speaking, it is only $\overline{\M}_\gamma$ in its entirety that one could establish as being preserved under time evolution in virtue of Theorem~\ref{thm}. The most rigorous treatment should therefore consider all of $\overline{\M}_\gamma$ in order to study the effective dynamics of the system. In lieu of additional mathematical justification, any further restriction of the phase space must be regarded as a heuristic simplification rather than a fundamental reduction.

Nevertheless, physical considerations strongly suggest that only genuine spherically symmetric degrees of freedom should be retained in the system throughout its evolution. In this sense, it is both natural and practical to introduce an additional truncation to a subspace $\overline{\M}_{\gamma,\text{sym}}\subset\overline{\M}_\gamma$ which better reflects the structure of the symmetry-restricted continuum theory.
While not guaranteed to be preserved under the Hamiltonian flow of the system, this reduced phase space offers a more manageable and intuitive framework for the extraction of dynamics in effective LQG.

To motivate how this truncation might be chosen, recall that $\overline{h}_r(v)|_{v_\theta = \pi/2}$ and $\overline{P}_r(v)|_{v_\theta = \pi/2}$ are fully determined by a single pair of functions, $(a_\varepsilon(v_r),p_a^\varepsilon(v_r))$. Although more general parameterizations should be anticipated for other radial edges, deviations from the equatorial expressions are expected to be sub-leading in the continuum limit. It is therefore reasonable to assume that the one-dimensional parameterization extends to all other radial edges, with insignificant corrections arising from terms of order $\varepsilon$ or higher. Moreover, this is precisely the form one obtains by applying the spherical-graph discretization map $\mathfrak{D}_{e^v_r}$ to the spherically symmetric AB variables specified in \eqref{AE_bar}, as we shall see below. We are thus led to conjecture that a physically-reasonable truncation of $\overline{\M}_\gamma$ is provided by 
\begin{align}\label{M_bar_gamma_sym}
    \overline{\M}_{\gamma,\text{sym}}\coloneqq \left\{\mathfrak{D}_{e^v_i}\left(\overline{A},\overline{E}\right)\,\big|\,\left(\overline{A},\overline{E}\right)\in\overline{\M}_{\text{AB}},\,(v,i)\in\gamma^v\times\mathcal{I}\right\} \subset \overline{\M}_\gamma,
\end{align}
i.e., the image of the restricted continuum phase space under the collection of discretization maps associated with our spherical graph: $\overline{\M}_{\gamma,\text{sym}} = \mathfrak{D}_\gamma(\overline{\M}_\text{AB})$ (see Figs.~\ref{fig:phase_spaces}--\ref{fig:symp_mflds}).

Applying the spherical-graph discretization map to the AB variables \eqref{AE_bar}, we obtain a set of spherically symmetric holonomies of the form
\begin{equation} \label{h_bar}
    \begin{gathered}
        \overline{h}_r(v) = \exp\big[\varepsilon_r\,a_\varepsilon(v_r)\,\tau_1\big], \\
        \overline{h}_\theta(v) = \exp\Big(\varepsilon_\theta\big[b_\varepsilon(v_r)\,\tau_2 + \Pi_\varepsilon(v_r)\,\tau_3\big]\Big), \\
        \overline{h}_\varphi(v) = \exp\bigg[\varepsilon_\varphi\Big(\cos v_\theta\,\tau_1 - \big[\Pi_\varepsilon(v_r)\,\tau_2 - b(v_r)\,\tau_3\big]\sin v_\theta\Big)\bigg].
    \end{gathered}
\end{equation}
Here, we have defined regularized analogs of the continuum functions $a,b,\Pi$ by
\begin{align}\label{a_b_Pi_epsilon}
    a_\varepsilon(v_r)\coloneqq \int_0^1 a(v_r + s\varepsilon_r)\,ds, && b_\varepsilon(v_r) \coloneqq b\big|_{\gamma^v}(v_r), && \Pi_\varepsilon(v_r) \coloneqq \frac{\Delta_r p_a^\varepsilon(v_r)}{2\,p_b^\varepsilon(v_r)},
\end{align}
where 
\begin{align}\label{finite_diff}
    \Delta_r f(v_r) \coloneqq \frac{1}{2\varepsilon_r}\left[f(v_r + \varepsilon_r) - f(v_r - \varepsilon_r)\right]
\end{align}
is a finite approximation
of $\partial_r f$. Furthermore,
\begin{align}\label{pa_pb_epsilon}
    p_a^\varepsilon(v_r) \coloneqq \mathrm{sgn}\left(p_b\left(v_r + \frac{\varepsilon_r}{2}\right)\right)\,p_a\left(v_r + \frac{\varepsilon_r}{2}\right), && p_b^\varepsilon(v_r) \coloneqq \int_{-1/2}^{1/2}\big|p_b(v_r + s\varepsilon_r)\big|\,ds
\end{align}
are discretizations of the quantities $\sgn{p_b}\,p_a$ and $|p_b|$ determining the momentum variables on $\overline{\M}_{\text{AB}}$.

In the case of the gauge-covariant fluxes, one can in theory construct a suitable system of paths in each $\mathcal{S}^v_i$ and use the double-cover \eqref{pi_U} to translate the conjugations appearing in \eqref{D_gamma_def} into rotations of the internal indices, thereby providing minor computational relief in the evaluation of $\overline{P}^I_i(v)$. The resulting expressions, however, are quite complicated and far too lengthy for the illustrative purposes of this section. We will therefore approximate the gauge-covariant fluxes using the \textit{standard fluxes} of the densitized triad over the surfaces in $\gamma^*$, i.e., we take $\overline{P}_i(v) \approx \int_{\mathcal{S}^v_i}\star\overline{E}$. Doing so, we obtain a simple set of diagonal fluxes parameterized by the functions $p_a^\varepsilon$ and $p_b^\varepsilon$ defined in \eqref{pa_pb_epsilon}:
\begin{equation}\label{P_bar}
    \begin{gathered}
        \overline{P}_r(v) = \varepsilon_\theta\varepsilon_\varphi\sinc\left(\frac{\varepsilon_\theta}{2}\right)\,p_a^\varepsilon(v_r)\sin v_\theta\,\tau_1, \\
        \overline{P}_\theta(v) = \varepsilon_\varphi\varepsilon_r\,p_b^\varepsilon(v_r)\sin\left(v_\theta+\frac{\varepsilon_\theta}{2}\right)\,\tau_2, \\
        \overline{P}_\varphi(v) = \varepsilon_r\varepsilon_\theta\,p_b^\varepsilon(v_r)\,\tau_3,
    \end{gathered}
\end{equation}
where $\sinc(x)\coloneqq \sin(x)/x$ denotes the sinc function.

Let us now observe that the presence of radial derivatives in the connection $\overline{A}$ introduces a subtle peculiarity in the construction of $\overline{\M}_{\gamma,\text{sym}}$. In particular, the continuum quantity $\Pi = \partial_r p_a/2p_b$ encodes information about the behaviour of $p_a$ under \textit{infinitesimal} variations in the radial coordinate --- information that is not accessible within the holonomies, which only capture behaviour associated with \textit{finite} edges in the graph. In principle, one could extend $\overline{\M}_{\gamma,\text{sym}}$ and view the function $\Pi_\varepsilon$ in \eqref{h_bar} as an additional independent parameter, but doing so would necessitate the introduction of a corresponding momentum variable in \eqref{P_bar} so as to maintain a well-defined symplectic structure.

In this work, however, we adopt a more pragmatic approach: we assume that $p_a$ varies sufficiently slowly across the graph, such that its derivative can be approximated by a \textit{finite difference operator} $\Delta_r$, constructed entirely from data defined on $(\gamma,\gamma^*)$.
For our purposes, we have taken $\Delta_r$ to be the \textit{central} finite difference operator defined in \eqref{finite_diff}, but this is by no means a unique or necessary choice, thus reflecting a novel discretization ambiguity at the \textit{kinematical} level of the theory. 

\begin{figure}[htp]
    \centering
    \includegraphics[width=\linewidth]{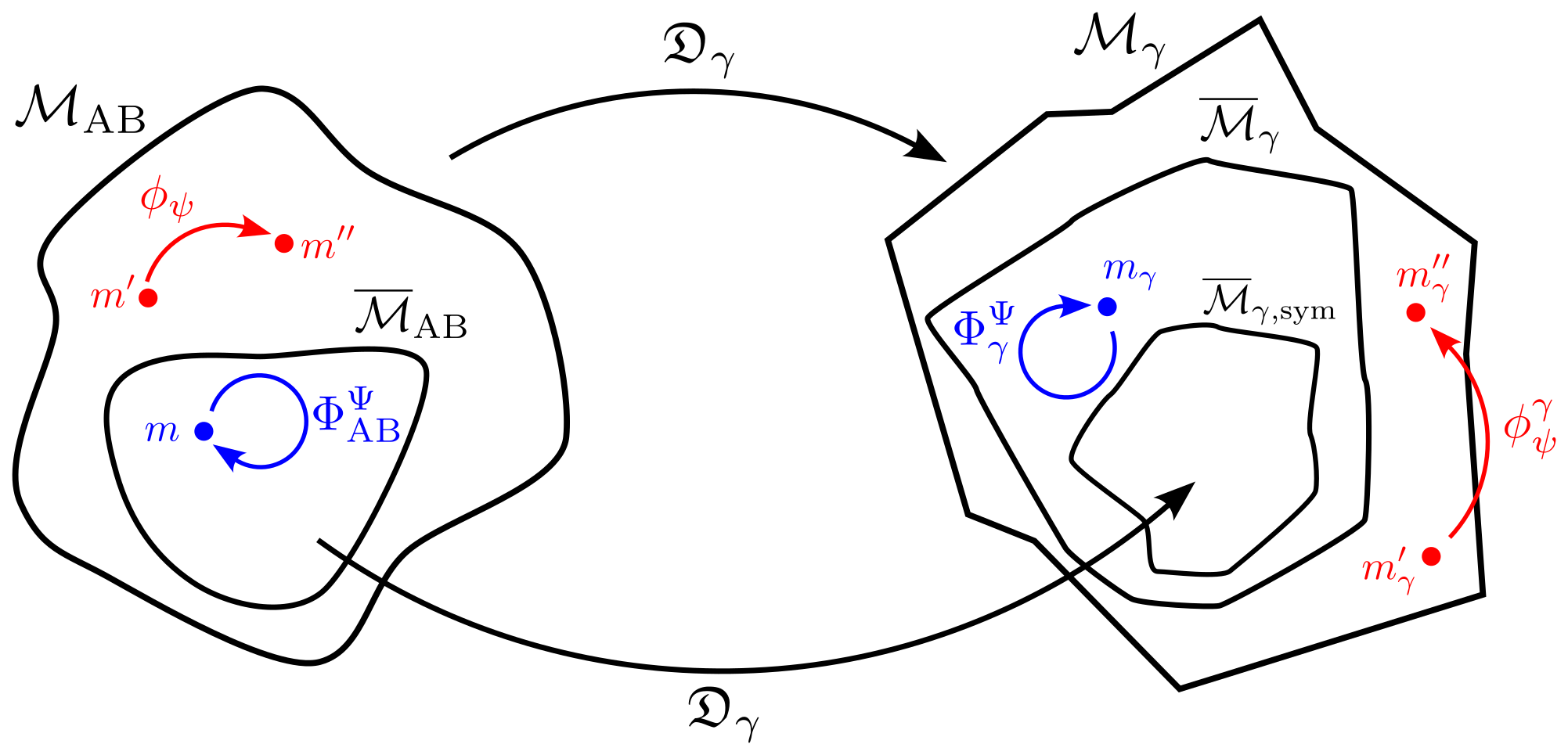}
    \caption{Depiction of the various phase spaces that we consider in both the continuum and truncated theories. On the left, the invariant subspace $\overline{\M}_{\text{AB}}\subset\M_{\text{AB}}$ consists of all fixed points under the action of $\Phi^\Psi_{\text{AB}}$ (e.g., the blue point $m$). For points such as $m'\in\M_{\text{AB}}\backslash\overline{\M}_{\text{AB}}$, there exists at least one element $\phi_\psi\in\Phi^\Psi_{\text{AB}}$ such that $\mathcal{A}^R_{\phi_\psi}(m') = m'' \neq m'$. On the right, we have the discretized phase space $\M_\gamma = \mathfrak{D}_\gamma(\M_{\text{AB}})$, where $\mathfrak{D}_\gamma = \{\mathfrak{D}_e\}_{e\in\gamma}$ denotes the collection of discretization maps associated with a graph $\gamma$. As in the continuum case, the invariant subspace $\overline{\M}_\gamma \subset \M_\gamma$ is defined with respect to a group $\Phi^\Psi_\gamma$ by $m_\gamma\in\overline{\M}_\gamma \iff \mathcal{A}^R_{\phi_\psi^\gamma}(m_\gamma) = m_\gamma$, $\forall \phi_\psi^\gamma\in\Phi^\Psi_\gamma$. Within $\overline{\M}_\gamma$, we also have the physically-relevant subspace $\overline{M}_{\gamma,\text{sym}} = \mathfrak{D}_\gamma(\overline{\M}_\gamma)$, which completes the tower of truncated phase spaces, $\overline{\M}_{\gamma,\text{sym}} \subset \overline{\M}_\gamma \subset \M_\gamma$.}
    \label{fig:phase_spaces}
\end{figure}

\begin{figure}
    \centering
    \includegraphics[width=\linewidth]{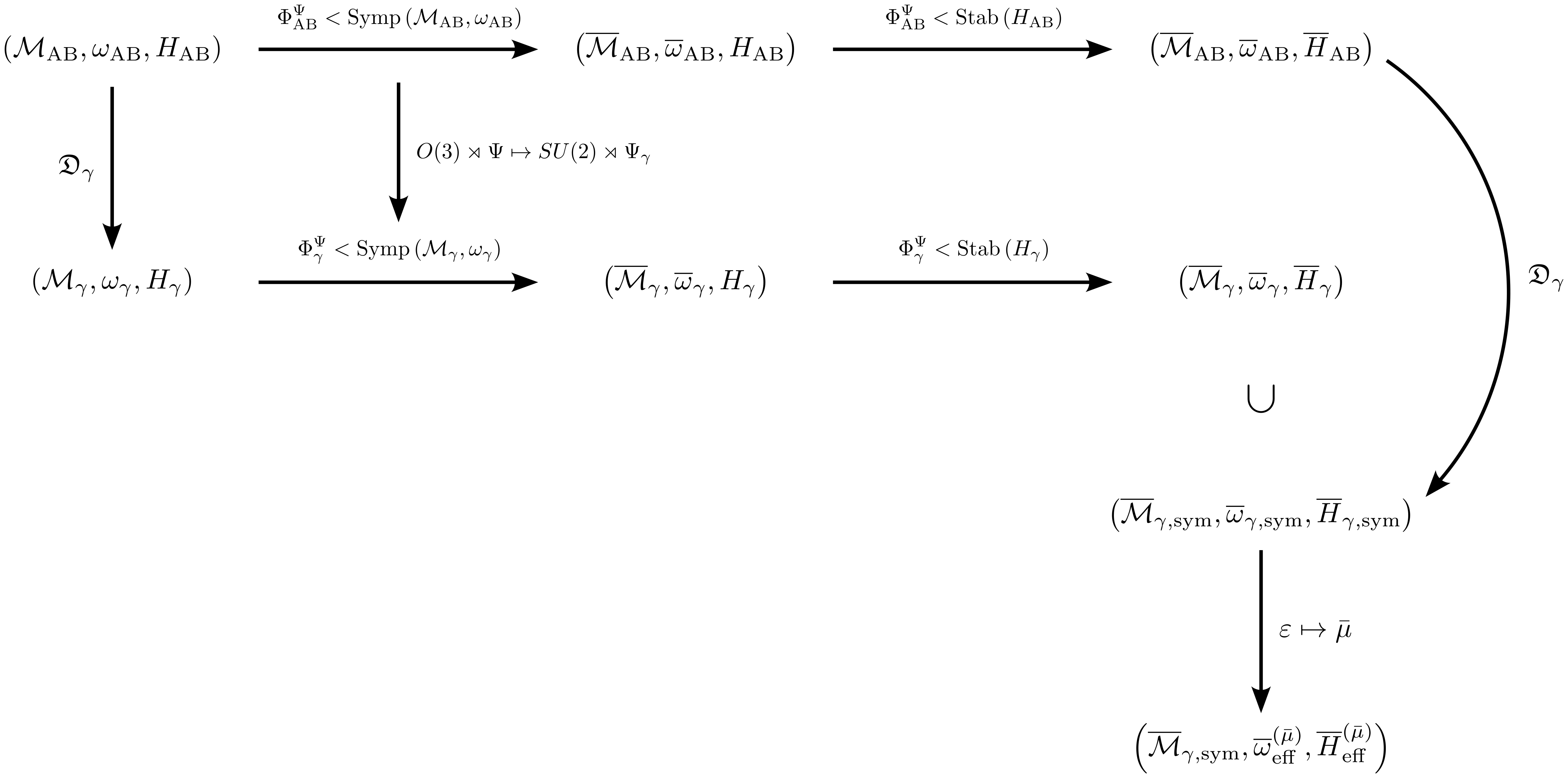}
    \caption{
        Adaptation of the effective dynamics diagram (Fig.~\ref{fig:algorithm}) to the specific approach taken in this work, which differs only by a particular choice of effective approximation. In our case, we transition from the full spherically symmetric sector of the discretized theory $(\overline{\M}_\gamma, \overline{\omega}_\gamma, \overline{H}_\gamma)$ to the physically-relevant subsystem $(\overline{\M}_{\gamma,\text{sym}}, \overline{\omega}_{\gamma,\text{sym}}, \overline{H}_{\gamma,\text{sym}})$ via direct discretization of the symmetric continuum theory plus an additional restriction of the symplectic form and Hamiltonian: $\overline{\M}_{\gamma,\text{sym}} = \mathfrak{D}_\gamma(\overline{\M}_\text{AB})$, $\overline{\omega}_{\gamma,\text{sym}} = \overline{\omega}_\gamma|_{\overline{\M}_{\gamma,\text{sym}}}$, $\overline{H}_{\gamma,\text{sym}} = \overline{H}_\gamma|_{\overline{\M}_{\gamma,\text{sym}}}$.
    }
    \label{fig:symp_mflds}
\end{figure}

\subsection{Symplectic Structure}\label{s4.4:symplectic_structure}

Having defined the phase space of interest, $\overline{\M}_{\gamma,\text{sym}}$, we proceed to compute the induced symplectic form $\overline{\omega}_{\gamma,\text{sym}} \coloneqq \overline{\omega}_\gamma|_{\overline{\M}_{\gamma,\text{sym}}}$ along with the resulting Poisson algebra.  
In much of the LQG literature, the symplectic structure on the reduced phase space is simply postulated in the form of Poisson brackets under the assumption that the relations at the quantum level should mirror those of the continuum theory \cite{Symmetry_Restriction}. We demonstrate, however, that this assumption \textit{fails} in general: the symplectic structure inherited from the full phase space in the case of a spherical graph exhibits non-trivial corrections for finite discretization parameters, which could in turn impact the macroscopic dynamics.

The starting point of our derivation is the symplectic potential on $\M_\gamma$,
\begin{align}\label{xi_gamma}
    \xi_\gamma = -\frac{4}{\kappa\beta}\sum_{e\in\gamma}\tr\left(P_e\,\Theta_e\right),
\end{align}
which is related to the symplectic form \eqref{omega_gamma} via $\omega_\gamma = -\delta\xi_\gamma$. Evaluating this expression using holonomies and fluxes of the form \eqref{hP_bar_general}, we obtain the symplectic potential on the maximal $\Phi_\gamma^\Psi$-invariant phase space, $\overline{\M}_\gamma$:
\begin{align}\label{xi_bar_gamma}
    \xi_\gamma\big|_{\overline{\M}_\gamma} &= \frac{4\pi}{\kappa\beta}\sum_{v_r,v_\theta}\sum_{i\in\mathcal{I}^+}\varepsilon_i\delta_i\left[\delta_{KL}\left(\sinc(\varepsilon_i\bar{\lambda}_i)\,\overline{p}^L_i + \frac{1}{\bar{\lambda}_i^2}\left[1-\sin(\varepsilon_i\bar{\lambda}_i)\right]\delta_{IJ}\,\overline{p}^I_i\,\overline{u}^J_i\,\overline{u}^L_i\right)\right. \nonumber \\
    &\qquad\qquad\qquad -\frac{\varepsilon_i}{2}\sinc^2\left(\frac{\varepsilon_i}{2}\bar{\lambda}_i\right)\,\epsilon_{IJK}\,\overline{p}^I_i\,\overline{u}^J_i\bigg]\delta\overline{u}^K_i,
\end{align}
in which $\bar{\lambda}_i \coloneqq \Vert \overline{u}_i\Vert$, and we have dropped the explicit vertex labels for notational simplicity.

Restricting further to $\overline{\M}_{\gamma,\text{sym}}\subset \overline{\M}_\gamma$, in which the holonomies and fluxes are given explicitly by \eqref{h_bar}--\eqref{P_bar}, yields several key simplifications. In particular, the parameterization becomes independent of $v_\theta$ and the fluxes assume a diagonal form, causing the symplectic potential to reduce to
\begin{align}\label{xi_bar_sym}
    \xi_\gamma\big|_{\overline{\M}_{\gamma,\text{sym}}} &= \frac{4\pi}{\kappa\beta}\sum_{v_r}\varepsilon_r\bigg(2\cos\left(\frac{\varepsilon_\theta}{2}\right)\,p_a^\varepsilon\,\delta a_\varepsilon + \mathcal{F}_\Omega(b_\varepsilon,\Pi_\varepsilon)\,p_b^\varepsilon\,\delta b_\varepsilon \nonumber \\
    &\qquad\qquad\qquad + \frac{1}{2}\,\mathcal{G}_\Omega(b_\varepsilon,\Pi_\varepsilon)\,\Pi_\varepsilon\,b_\varepsilon\Big[\Delta_r\big(\delta p_a^\varepsilon) - 2\Pi_\varepsilon\,\delta p_b^\varepsilon\Big]\bigg),
\end{align}
where we have defined
\begin{align}\label{F_Omega}
    \mathcal{F}_\Omega(b_\varepsilon,\Pi_\varepsilon) &\coloneqq \varepsilon_\theta\bigg(\frac{1}{\bar{\lambda}_\theta^2}\Big[b_\varepsilon^2 + \Pi_\varepsilon^2\sinc(\varepsilon_\theta\bar{\lambda}_\theta)\Big]\cos\left(\frac{\varepsilon_\theta}{2}\right)\cot\left(\frac{\varepsilon_\theta}{2}\right) \nonumber \\
    &\qquad + \sum_{v_\theta}\frac{1}{\bar{\lambda}_\varphi^2}\Big[b_\varepsilon^2\sin^2v_\theta + \big(\cos^2v_\theta + \Pi_\varepsilon^2\sin^2v_\theta\big)\sinc(\varepsilon_\varphi\bar{\lambda}_\varphi)\Big]\sin v_\theta\bigg),
\end{align}
and
\begin{align}\label{G_Omega}
    \mathcal{G}_\Omega(b_\varepsilon,\Pi_\varepsilon) &\coloneqq \varepsilon_\theta\bigg(\frac{1}{\bar{\lambda}_\theta^2}\Big[1-\sinc(\varepsilon_\theta\bar{\lambda}_\theta)\Big]\cos\left(\frac{\varepsilon_\theta}{2}\right)\cot\left(\frac{\varepsilon_\theta}{2}\right) \nonumber \\
    &\qquad\qquad + \sum_{v_\theta}\frac{1}{\bar{\lambda}_\varphi^2}\Big[1-\sinc(\varepsilon_\varphi\bar{\lambda}_\varphi)\Big]\sin^3v_\theta\bigg),
\end{align}
with $\bar{\lambda}_\theta^2 = b_\varepsilon^2 + \Pi_\varepsilon^2$ and $\bar{\lambda}_\varphi^2 = \cos^2v_\theta + \bar{\lambda}_\theta^2\,\sin^2v_\theta$. The symplectic form on $\overline{\M}_{\gamma,\text{sym}}$ can then be obtained by exterior differentiation of \eqref{xi_bar_sym}, leading to
\begin{align}\label{omega_bar_sym}
    \overline{\omega}_{\gamma,\text{sym}} &= \frac{4\pi}{\kappa\beta}\sum_{v_r}\varepsilon_r\bigg[2\cos(\varepsilon_\theta/2)\,\delta a_\varepsilon\wedge \delta p_a^\varepsilon + \Big(\mathcal{F}_\Omega + \mathcal{H}_\Omega\,\Pi_\varepsilon^2\Big) \delta b_\varepsilon\wedge \delta p_b^\varepsilon \nonumber \\
    &\qquad\qquad + \frac{1}{2}\,\mathcal{H}_\Omega\,\Pi_\varepsilon\,\Delta_r\big(\delta p_a^\varepsilon\big)\wedge \delta b_\varepsilon + \frac{1}{2}\,\mathcal{G}_\Omega\,\Pi_\varepsilon\,\frac{b_\varepsilon}{p_b^\varepsilon}\,\Delta_r\big(\delta p_a^\varepsilon\big)\wedge \delta p_b^\varepsilon\bigg],
\end{align}
in which we have introduced the additional structure function
\begin{align}\label{H_Omega}
    \mathcal{H}_\Omega(b_\varepsilon,\Pi_\varepsilon) &\coloneqq \mathcal{G}_\Omega + b_\varepsilon\,\partial_{b_\varepsilon}\mathcal{G}_\Omega - \frac{1}{\Pi_\varepsilon}\partial_{\Pi_\varepsilon}\mathcal{F}_\Omega \nonumber \\
    &= \frac{\varepsilon_\theta}{2}\bigg[\varepsilon_\theta^2\sinc^2(\varepsilon_\theta\bar{\lambda}_\theta/2)\cos\left(\frac{\varepsilon_\theta}{2}\right)\cot\left(\frac{\varepsilon_\theta}{2}\right) + \varepsilon_\varphi\sum_{v_\theta}\sinc^2\left(\frac{\varepsilon_\varphi}{2}\bar{\lambda}_\varphi\right)\sin^3v_\theta\bigg].
\end{align}

Using the symplectic form \eqref{omega_bar_sym}, one can construct the unique Hamiltonian vector fields associated with the functions $a_\varepsilon$,$b_\varepsilon$,$p_a^\varepsilon$, and $p_b^\varepsilon$ coordinatizing the reduced phase space (see Appendix~\ref{app:restricted_algebra}). Doing so, we find that the formal Poisson algebra on $\overline{\M}_{\gamma,\text{sym}}$ reads
\begin{align}
    &\big\{a_\varepsilon(v_r),p_a^\varepsilon(v_r')\big\} = \frac{\kappa\beta}{8\pi\varepsilon_r}\sec\left(\frac{\varepsilon_\theta}{2}\right)\,\delta_{v_r,v_r'}, \label{a_pa_bracket} \\
    &\big\{b_\varepsilon(v_r),p_b^\varepsilon(v_r')\big\} = \frac{\kappa\beta}{4\pi\varepsilon_r}\Bigg(\frac{1}{\mathcal{F}_\Omega(v_r) + \mathcal{H}_\Omega(v_r)\Pi_\varepsilon(v_r)^2}\Bigg)\delta_{v_r,v_r'}, \label{b_pb_bracket} \\
    &\big\{a_\varepsilon(v_r),p_b^\varepsilon(v_r')\big\} = \frac{\kappa\beta}{32\pi\varepsilon_r^2}\sec\left(\frac{\varepsilon_\theta}{2}\right)\Bigg(\frac{\mathcal{H}_\Omega(v_r')\Pi_\varepsilon(v_r')}{\mathcal{F}_\Omega(v_r') + \mathcal{H}_\Omega(v_r')\Pi_\varepsilon(v_r')^2}\Bigg) \nonumber \\ 
    &\qquad\qquad\qquad\qquad\qquad\qquad\times\big(\delta_{v_r,v_r'+\varepsilon_r} - \delta_{v_r,v_r'-\varepsilon_r}\big), \label{a_pb_bracket}\\
    &\big\{a_\varepsilon(v_r),b_\varepsilon(v_r')\big\} = \frac{\kappa\beta}{32\pi\varepsilon_r^2}\sec\left(\frac{\varepsilon_\theta}{2}\right)\Bigg(\frac{\mathcal{G}_\Omega(v_r')\Pi_\varepsilon(v_r')}{\mathcal{F}_\Omega(v_r') + \mathcal{H}_\Omega(v_r')\Pi_\varepsilon(v_r')^2}\Bigg)\frac{b_\varepsilon(v_r')}{p_b^\varepsilon(v_r')} \nonumber \\
    &\qquad\qquad\qquad\qquad\qquad\qquad\times\big(\delta_{v_r,v_r'-\varepsilon_r} - \delta_{v_r,v_r'+\varepsilon_r}\big). \label{a_b_bracket}
\end{align}
One immediately striking feature of $\overline{\M}_{\gamma,\text{sym}}$ illustrated by this algebra is the existence of \textit{non-commuting configuration variables}. This is in stark contrast to even the full discretized phase space $\M_\gamma$, in which we had non-commuting momenta, but vanishing Poisson brackets among the configuration variables (recall \eqref{hf_alg}). Further investigation into this peculiar facet of the reduced phase space ---including physical interpretations--- shall be left as a topic for future research.

Although the preceding algebra may appear unorthodox \textit{prima facie}, we stress that it does in fact exhibit the correct behaviour in the continuum limit. In particular, we have
\begin{align}
    &\big\{a_\varepsilon(v_r),p_a^\varepsilon(v_r')\big\} \approx \frac{\kappa\beta}{8\pi}\left(1 + \frac{\varepsilon_\theta^2}{8}\right)\delta(v_r-v_r'), \\
    &\big\{b_\varepsilon(v_r),p_b^\varepsilon(v_r')\big\} \approx \frac{\kappa\beta}{16\pi}\Bigg\{1 + \frac{37}{6}\left(\frac{\varepsilon_\theta^2}{4!}\right)\Bigg[1 - \frac{22}{37}\left(1 + \varepsilon_r\,\frac{\partial_r^2p_a(v_r)}{\partial_rp_a(v_r)}\right)\nonumber \\
    &\qquad\qquad\qquad\qquad\qquad\qquad\qquad\qquad\qquad\times\left(\frac{\partial_rp_a(v_r)}{p_b(v_r)}\right)^2\Bigg]\Bigg\}\delta(v_r-v_r'), \\
    &\big\{a_\varepsilon(v_r),p_b^\varepsilon(v_r')\big\} \approx \frac{33\kappa\beta}{32\pi}\left(\frac{\varepsilon_\theta}{3!}\right)^2\bigg[\partial_r p_a(v_r') + \frac{\varepsilon_r}{2}\,\partial_r^2p_a(v_r')\bigg]\frac{1}{p_b(v_r')}\,\partial_{v_r'}\delta(v_r-v_r'), \\
    &\big\{a_\varepsilon(v_r),b_\varepsilon(v_r')\big\} \approx -\frac{11\kappa\beta}{32\pi}\left(\frac{\varepsilon_\theta}{3!}\right)^2\bigg[\partial_rp_a(v_r') + \frac{\varepsilon_r}{2}\,\partial_r^2p_a(v_r')\bigg]\frac{b(v_r')}{p_b(v_r')^2}\,\partial_{v_r'}\delta(v_r-v_r'),
\end{align}
where we have used the fact that $(1/\varepsilon_r)\delta_{v_r,v_r'} = \delta(v_r-v_r') + \mathcal{O}(\varepsilon_r)$. Recalling from \eqref{pa_pb_epsilon} that $p_a^\varepsilon$ and $p_b^\varepsilon$ both contain factors of $\sgn{p_b}$, we see that in the limit $\varepsilon \to 0$, we recover precisely the spherically symmetric continuum algebra \eqref{M_bar_AB_alg}.

\subsection{Truncation of the Scalar Constraint}\label{s4.5:scalar_constraint}

In order for the effective dynamics program outlined in Sec.~\ref{s4.1:algorithm} to yield consistent results, it is essential that the symmetry-restricted phase space be preserved under dynamical evolution. It therefore remains to verify that the Thiemann constraint \eqref{C_Thiemann} is invariant with respect to the spherical-graph symmetry group $\Phi_\gamma^\Psi$:
\begin{align}\label{C_gamma_invariance}
    \phi_\psi^*C_\gamma[N] = C_\gamma[N], \qquad \forall\phi_\psi = (U_\psi,\psi)\in\Phi_\gamma^\Psi,
\end{align}
where the pullback $\phi_\psi^*$ denotes the induced $\Phi_\gamma^\Psi$-action on functions $f\in C^\infty(\M_\gamma)$, defined by
\begin{align}\label{pullback_gamma}
    \left(\phi_\psi^*f\right)\big|_{(h_e,P_e)} \coloneqq f\big|_{\mathcal{A}^R_{\phi_\psi}(h_e,P_e)}.
\end{align}
To this end, it suffices to show that the individual constraint densities \eqref{CE_Gamma}--\eqref{CL_Gamma} transform covariantly under the action of $\Phi_\gamma^\Psi$, in the sense that \cite{Symmetry_Restriction}
\begin{align}\label{constraint_covariance}
    \phi_\psi^*C_\gamma^E = C_\gamma^E\circ\psi, && \phi_\psi^*C_\gamma^L = C_\gamma^L\circ\psi.
\end{align}
When evaluating these actions explicitly, we shall make extensive use of the fact that each $\phi_\psi\in\Phi_\gamma^\Psi$ is a symplectomorphism of $(\M_\gamma,\omega_\gamma)$, hence the associated pullback necessarily preserves Poisson brackets \cite{Thiemann_MCQGR}:
\begin{align}
    \phi_\psi^*\{f,g\} = \{\phi_\psi^* f, \phi_\psi^* g\}, && \forall f,g\in C^\infty(\M_\gamma).
\end{align}

Now, rather than confining ourselves to the spherical-graph symmetry group $\Phi_\gamma^\Psi$, we shall consider a much broader class of symmetries described by a more general group $\Phi_\gamma < SU(2)\rtimes\Diff(\Sigma)$. 
In particular, we only assume that each group element is of the form $\phi_\psi = (U_\psi,\psi)$, where $\psi\in\Diff(\Sigma)$ is a proper graph-preserving diffeomorphism (with respect to an \textit{arbitrary} coordinate-adapted graph $\gamma$) and $U_\psi\in SU(2)$ is phase-space-independent. 
In this case, the holonomies and gauge-covariant fluxes transform under $\phi_\psi\in\Phi_\gamma$ according to 
\begin{align}
    h_i(v) &\mapsto U_\psi^\dagger\big(\psi(v)\big)\,h_{\psi(i)}\big(\psi(v)\big)\,U_\psi\big(\psi(v+\mu_i)\big), \label{h_transformation} \\
    P^I_i(v) &\mapsto \left(O_\psi^T\right)^I_{\phantom{I}J}\big(\psi(v)\big)\,P^J_{\psi(i)}\big(\psi(v)\big), \label{P_transformation}
\end{align}
where $O_\psi = D^{(1)}(U_\psi)\in SO(3)$,
and we have used the fact that $b_{\psi(e^v_i)} = (\psi\circ e^v_i)(0) = \psi(v)$ and $f_{\psi(e^v_i)} = (\psi\circ e^v_i)(1) = \psi(v+\mu_i)$.

Beginning with the Euclidean density $C_\gamma^E$, we first observe that the transformation \eqref{h_transformation} extends to the plaquette holonomies \eqref{h_box} as
\begin{align}
    h\left(\square^v_{ij}\right) \mapsto U_\psi^\dagger\big(\psi(v)\big)\,h\left(\square^{\psi(v)}_{\psi(i)\psi(j)}\right)\,U_\psi\big(\psi(v)\big),
\end{align}
which follows from the fact that each plaquette is constructed so as to begin and end at the same point. Because the gauge transformations $U_\psi\in SU(2)$ associated with each $\phi_\psi\in\Phi_\gamma$ are phase-space-independent by assumption, one can easily show that 
\begin{align}\label{CE_gamma_transformation}
    \left(\phi_\psi^*C_\gamma^E\right)(v) &= -\frac{1}{2\kappa^2\beta}\sum_{\psi^{-1}(\{i,j,k\})}\epsilon(i,j,k)\,\tr\bigg(\left[h\left(\square^{\psi(v)}_{ij}\right) - h^\dagger\left(\square^{\psi(v)}_{ij}\right)\right] \nonumber \\
    &\qquad\qquad\qquad\times \Big\{h_k\big(\psi(v)\big), \phi_\psi^*V[\gamma]\Big\}\,h_k^\dagger\big(\psi(v)\big)\bigg),
\end{align}
in which we have re-indexed the sum via $\{i,j,k\}\to \psi^{-1}(\{i,j,k\})$ and used the cyclic property of the trace. 
Consequently, the covariance requirement \eqref{constraint_covariance} for the Euclidean part of the Thiemann constraint follows automatically if the volume $V[\gamma]$ is shown to be $\Phi_\gamma$-invariant.

Using the $SO(3)$-invariance of the Levi-Civita symbol, $\epsilon_{LMN} = \epsilon_{IJK}O^I_{\phantom{I}L}O^J_{\phantom{J}M}O^K_{\phantom{K}N}$, it follows from \eqref{V_Gamma} that
\begin{align}
    \left(\phi_\psi^*V_\gamma\right)(v) &= \sqrt{\frac{1}{2^3\cdot 3!}\Bigg|\sum_{\psi^{-1}(\{i,j,k\})}\epsilon(i,j,k)\epsilon_{IJK}P^I_i\big(\psi(v)\big)P^J_j\big(\psi(v)\big)P^K_k\big(\psi(v)\big)\Bigg|},
\end{align}
where we have re-indexed the sum as above. Because $\psi$ induces a bijection on the index set $\mathcal{I}$, we see that this is in fact equivalent to the covariance condition $\phi_\psi^*V_\gamma = V_\gamma\circ\psi$. In particular, we have established 
\begin{align}
    \label{V_gamma_invariance}
    \phi_\psi^*V[\gamma] = \sum_{v\in\gamma^v}V_\gamma\big(\psi(v)\big) = V[\gamma],
\end{align}
hence the total discretized volume is manifestly $\Phi_\gamma$-invariant, and \eqref{CE_gamma_transformation} collapses to the desired result: $\phi_\psi^*C_\gamma^E = C_\gamma^E\circ\psi$.

Progressing now to the Lorentzian density $C_\gamma^L$, we proceed in a similar manner as above to obtain
\begin{align}\label{CL_gamma_transformation}
    \left(\phi_\psi^*C_\gamma^L\right)(v) &= \frac{2^3(1+\beta^2)}{\kappa^4\beta^7}\sum_{\psi^{-1}(\{i,j,k\})}\epsilon(i,j,k)\,\tr\bigg[\Big\{h_i\big(\psi(v)\big),\phi_\psi^*K[\gamma]\Big\}h_i^\dagger\big(\psi(v)\big) \nonumber \\
    &\quad\times \Big\{h_j\big(\psi(v)\big),\phi_\psi^*K[\gamma]\Big\}\,h_j^\dagger\big(\psi(v)\big)\Big\{h_k\big(\psi(v)\big),V[\gamma]\Big\}\,h_k^\dagger\big(\psi(v)\big)\bigg],
\end{align}
where we have invoked \eqref{V_gamma_invariance} in the final Poisson bracket. In this case, we need only show that the regularized extrinsic curvature is $\Phi_\gamma$-invariant, i.e., $\phi_\psi^* K[\gamma] = K[\gamma]$. Recalling the definition \eqref{K_Gamma}, however, we see that this is in fact a trivial task: $K[\gamma]$ is constructed entirely from objects whose invariance we have already established --- namely, $V[\gamma]$ and $C_\gamma^E[1] = s_e\sum_{v\in\gamma^v}C_\gamma^E(v)$. The covariance property of the Lorentzian part, $\phi_\psi^* C_\gamma^L = C_\gamma^L\circ\psi$, thus follows immediately from \eqref{CL_gamma_transformation}.

Having verified that both $C_\gamma^E$ and $C_\gamma^L$ satisfy the covariance requirement \eqref{constraint_covariance}, we can conclude that full Thiemann constraint is indeed $\Phi_\gamma$-invariant. By Theorem~\ref{thm}, it follows that the Hamiltonian flow generated by $C_\gamma[N]$ preserves the invariant subspace determined by $\Phi_\gamma$. In particular, initially-symmetric configurations are \textit{guaranteed} to remain symmetric throughout the entirety of their evolution. Although this invariance only formally guarantees preservation of the full $\Phi_\gamma^\Psi$-invariant phase space $\overline{\M}_\gamma\subset\M_\gamma$, we assume that any deviations from the physically-relevant subspace $\overline{\M}_{\gamma,\text{sym}}\subset\overline{\M}_\gamma$ can be considered negligible for the purposes of effective dynamics.

Additionally, because all Poisson brackets appearing in $C_\gamma[N]$ (evaluated on $\M_\gamma$ \textit{a priori}) involve at least one manifestly $\Phi_\gamma$-invariant function, we are justified in evaluating them using either the symplectic structure of the full phase space or that of the invariant subspace. In order to accommodate the use of gauge-covariant fluxes and to avoid the tedious task of propagating Poisson brackets through complicated phase space functions, the former approach turns out to be the more practical option. To this end, we shall adopt the procedure established in \cite{ADL20}, whereby all brackets in the constraint are first resolved by means of the full Poisson algebra, after which the complete expression is restricted to the symmetric phase space of interest (see Appendix~\ref{app:poly_scalar_constr}).

Summarizing, we have executed the effective dynamics algorithm of Sec.~\ref{s4.1:algorithm} in the case of physically-relevant, spherically symmetric phase space $\overline{\M}_{\gamma,\text{sym}}$. We are thus led to the effective scalar constraint governing spherically symmetric gravity on a graph:
\begin{align}
    C_\gamma[N]\mapsto \overline{C}_{\gamma,\text{sym}}[N] \coloneqq C_\gamma[N]\big|_{\overline{\M}_{\gamma,\text{sym}}},
\end{align}
expressed entirely in terms of the functions $a_\varepsilon$, $b_\varepsilon$, $p_a^\varepsilon$, and $p_b^\varepsilon$ introduced in Sec.~\ref{s4.3.2:physical_subspace}. As mentioned earlier, however, the full expression in all generality is prohibitively lengthy for the purposes of this paper, and further details regarding its construction can be found in Appendix~\ref{app:poly_scalar_constr}. Moreover, forthcoming instalments in this series will present explicit methods by which tractable approximations can be constructed in specific physical applications of interest.

\section{Discussion \& Outlook}\label{s5:discussion}

This manuscript is the first in a series concerning spherical symmetric solutions for gravity on a graph. Here, we have recapped the general procedure on how an effective description in the LQG philosophy can be obtained, starting from choosing an initial discretization, verifying the restriction to a symmetrically-invariant phase space and the replacement of the regulator by a $\bar\mu$-like quantity, which depends on the momentum of the gravitational field. Consequently, we applied the first steps of this procedure to spherical symmetry. We identified the symmetry group that maximizes the invariant parameters within the ``usual'' graph envisioned for spherically symmetric solutions\footnote{Note, that the choice of the graph is by no means unique and a clever choice of graph could indeed provide further advantages 
.}. The phase space of holonomies and fluxes over said graph was restricted by the means of symmetry restriction, yielding the invariant sub space, which guarantees that the dynamics of the system will respect said symmetry. However, the later set was deemed unfeasible for actual computations; assuming relations to physical reasonable configurations related to the continuum, the phase space was further restricted to a manageable subset $\overline{\M}_{\gamma,\text{sym}}$. We stress, however, that invariance of the dynamics needs to be assumed at this stage. The scalar constraint was evaluated on the $\overline{\M}_{\gamma,\text{sym}}$ as well as the symplectic structure. While the former is known to obtain phase-space-dependent corrections proportional to the lattice spacing $\varepsilon$, the same had never previously occurred for the symplectic structure. Here, however, these corrections appear and are identified as a motivation for potential interesting quantum effects at the kinematical level. Ultimately, the relevance of these effects can only be established {\it a posteriori} when quantum gravity measurements from isotropic cosmology or black confirm them. At present, they may be eventually overshadowed by the necessity to introduce the renormalization group \cite{Saeed_RG_1,Saeed_RG_2,LLT18} to both the kinematics and the dynamics of the underlying quantum gravity formalism. We reserve these questions for future investigations.


\acknowledgments{
The authors thank Tomasz Paw{\l}owksi and Wojciech Kami{\'n}ski for clarifying discussion on various topics within this paper. S.R. and J.R. acknowledge the support of the Natural Sciences and Engineering
Research Council of Canada (NSERC). K.L. would like to acknowledge support by the Munich Quantum Valley, which is supported by the Bavarian
state government with funds from the Hightech Agenda Bayern Plus.
}







\appendix

\section{Regularized Scalar Constraint: Computational Details}\label{app:poly_scalar_constr}

In this appendix, we provide further computational details regarding the Thiemann constraint \eqref{C_Thiemann} introduced in Sec.~\ref{s3.2:discretized_dynamics}. In particular, we  follow \cite{ADL20} and demonstrate how the Poisson brackets can be resolved at the level of the full phase space, resulting in a general expression that is far more tractable from a computational perspective. We begin by introducing a modified version of the holonomy-flux algebra \eqref{hf_alg}, which will then be applied to Euclidean and Lorentzian parts of the constraint separately. Finally, we put everything together to obtain the desired form of the full discretized scalar constraint, and proceed to discuss its restriction to a symmetric phase space of interest, $\overline{\M}_\gamma \subset \M_\gamma$.

To evaluate the Poisson brackets in \eqref{CE_Gamma}--\eqref{CL_Gamma}, we employ the regularized algebra on the phase space consisting of holonomies $h_e$ and \textit{standard} (or \textit{commuting}) \textit{fluxes} $E(\mathcal{S}_e) = \int_{\mathcal{S}_e}\star E \approx P_e$ over a graph $\gamma$ and dual cell complex $\gamma^*$. For arbitrarily-parameterized edges and surfaces $e\in\gamma, \mathcal{S}_{e'}\in\gamma^*$, this is given by \cite{Klaus_Thesis,Ashtekar_2004,Thiemann_MCQGR}
\begin{equation}
    \begin{gathered}
        \left\{h_e, h_{e'}\right\} = 0, \quad \left\{E(\mathcal{S}_e),E(\mathcal{S}_{e'})\right\} = 0, \\
        \left\{h_e, E_J(\mathcal{S}_{e'})\right\} = \frac{\kappa\beta}{2}\left(\delta_{e,e'} - \delta_{e^{-1},e'}\right) h_{e_1}\tau_J h_{e_2},
    \end{gathered}
\end{equation}
where $e_1, e_2$ refer to the portions of the edge $e = e_1\circ e_2$ lying before and after the point of intersection $e\cap\mathcal{S}_{e'}$, respectively. For the specific parameterizations \eqref{edge_param} and \eqref{surface_param}, we denote the corresponding holonomies and fluxes by $h_i(v) = h_{e^v_i}$ and $P_i(v) = E(\mathcal{S}^v_i)$, and the nontrivial Poisson bracket above reads
\begin{align} \label{commuting_flux_algebra}
    \left\{h_i(v), P^J_j(v')\right\} = \frac{\kappa\beta}{2}\left(\delta_{ij}\delta_{v,v'} - \delta_{-i,j}\delta_{v+\mu_i,v'}\right) h_i^1(v)\tau^J h_i^2(v),
\end{align}
with $h_i^1(v) \coloneqq h_{1/2}(e^v_i)$ and $h_i^2(v) \coloneqq h_{1/2}^\dagger(e^{v+\mu_i}_{-i})$.

\subsection{Euclidean Part}

Beginning with the Euclidean density \eqref{CE_Gamma}, we first observe that the plaquette holonomies \eqref{h_box} satisfy $h^\dagger(\square^v_{ij}) = h(\square^v_{ji})$, hence
\begin{align}
    C_\gamma^E(v) = -\frac{1}{\kappa^2\beta}\sum_{i,j,k}\epsilon(i,j,k)\,\tr\Big[h\big(\square^v_{ij}\big)\left\{h_k(v),V[\gamma]\right\} h_k^\dagger(v)\Big].
\end{align}
For any $\ell\in\mathcal{I}$, a straightforward application of the algebra \eqref{commuting_flux_algebra} leads to
\begin{align} \label{h_V_bracket}
    \left\{h_\ell(v),V[\gamma]\right\} h_\ell^\dagger(v) &= \frac{\kappa\beta}{8}\mathrm{Ad}_{h_\ell^1(v)}\big[\mathcal{E}_\ell(v) + \mathcal{E}_\ell(v+\mu_\ell)\big],
\end{align}
in which we have defined 
\begin{align}
    \mathcal{E}_\ell(v) \coloneqq \frac{s_e}{8}\sum_{j,k}\epsilon(j,k,\ell)\,\epsilon_{JKL}\left(\frac{P_j^J(v) P_k^K(v)}{V_\gamma(v)}\right)\tau^L.
\end{align}
This is an $\mathfrak{su}(2)$-valued object satisfying $\lim_{n\to\infty}\frac{1}{\varepsilon_\ell}\mathcal{E}_\ell(v) = \sgn{k}e_{|k|}(v)$, thereby warranting its interpretation as the \textit{regularized co-triad} \cite{Klaus_Thesis}.

On the other hand, the curvature term can be rewritten as $h(\square^v_{ij}) = \mathrm{Id}_{SU(2)}\mathcal{F}^0(\square^v_{ij})$ $ + \mathcal{F}^K(\square^v_{ij})\,\tau_K$, with $\mathcal{F}^0(\square^v_{ij}) \coloneqq \frac{1}{2}\tr\big[h(\square^v_{ij})\big]$ and $\mathcal{F}^K(\square^v_{ij}) \coloneqq -2\tr\big[\tau^K h(\square^v_{ij})\big]$. Note that the traceless part comprises the \textit{regularized curvature} of the Ashtekar connection: $\lim_{n\to\infty}\frac{1}{\varepsilon_i\varepsilon_j}\mathcal{F}(\square^v_{ij}) = \sgn{ij} F_{|i||j|}(v)$. Following a brief computation, we arrive at our final expression for the Euclidean part of the Thiemann constraint density,
\begin{align}
    \label{CE_Gamma_poly}
    C_{\gamma}^E(v) = \frac{1}{16\kappa}\sum_{i,j,k}\sum_\alpha\epsilon(i,j,k)\, \mathcal{T}^k_{JK}(v)\,\mathcal{F}^J\big(\square^v_{ij}\big)\mathcal{E}_k^K(v + \alpha\mu_k),
\end{align}
where $\alpha = 0,1$, and we have introduced the notation $\mathcal{T}^k_{JK}(v) \coloneqq -2\,\tr\big[\tau_J\mathrm{Ad}_{h_k^1(v)}(\tau_K)\big]$.

\subsection{Lorentzian Part}

Now we turn our attention to the Lorentzian part \eqref{CL_Gamma}, which is significantly more involved than the Euclidean term. To begin with, we can use \eqref{h_V_bracket} to rewrite the Lorentzian density as 
\begin{align}
    C_{\gamma}^L(v) &= \frac{1+\beta^2}{\kappa^3\beta^6}\sum_{i,j,k}\sum_{\alpha}\epsilon(i,j,k)\,\tr\Big[\left\{h_i(v),K[\gamma]\right\}h_i^\dagger(v) \left\{h_j(v),K[\gamma]\right\}h_j^\dagger(v)\,\tau^K\Big] \nonumber \\
    &\qquad\qquad\qquad\qquad \times \mathcal{T}^k_{KL}\,\mathcal{E}^L_k(v+\alpha\mu_k).
\end{align}
In order to resolve the remaining Poisson brackets, we first compute
\begin{align}
    K[\gamma] = \left\{C_{\gamma}^E[1],V[\gamma]\right\} = s_e\sum_{v\in\gamma^v}\left\{C_{\gamma}^E(v),V[\gamma]\right\}.
\end{align}
To this end, we have 
\begin{align}
    K_\gamma(v) &\coloneqq \left\{C_{\gamma}^E(v),V[\gamma]\right\} \nonumber \\
    &= \frac{1}{16\kappa}\sum_{i,j,k}\sum_\alpha\epsilon(i,j,k)\mathcal{T}^k_{JK}(v)\left\{\mathcal{F}^J\big(\square^v_{ij}\big),V[\gamma]\right\}\mathcal{E}_k^K(v+\alpha\mu_k) \nonumber
\end{align}
where we have used the fact that 
\begin{align}
    \left\{\mathcal{T}^\ell_{IJ}(v),V[\gamma]\right\} = \left\{\mathcal{E}_\ell^L(v),V[\gamma]\right\} = 0,
\end{align}
for any $\ell\in\mathcal{I}$. Through repeated application of the algebra \eqref{commuting_flux_algebra}, we find
\begin{align}
    &\left\{\mathcal{F}^J\big(\square^v_{ij}\big),V[\gamma]\right\} = -\frac{\kappa\beta}{4}\sum_\alpha \Big(\tr\left[\tau^J\tau^K h\big(\square^v_{ij}\big)\right]\mathcal{T}^i_{KL}(v)\mathcal{E}_i^L(v+\alpha\mu_i) - \tr\left[\tau^K\tau^J h\big(\square^v_{ij}\big)\right]\nonumber \\
    &\qquad \times\mathcal{T}^j_{KL}(v)\mathcal{E}_j^L(v+\alpha\mu_j) + \tr\left[\tau^J\mathrm{Ad}_{h_i(v)}(\tau^K)h\big(\square^v_{ij}\big)\right]\mathcal{T}^j_{KL}(v+\mu_i)\mathcal{E}_j^L(v+\mu_i+\alpha\mu_j) \nonumber \\
    &\qquad\qquad\qquad  - \tr\left[\mathrm{Ad}_{h_j(v)}(\tau^K)\tau^J h\big(\square^v_{ij}\big)\right]\mathcal{T}^i_{KL}(v+\mu_j)\mathcal{E}_i^I(v+\mu_j+\alpha\mu_i)\Big).
\end{align}
It follows that
\begin{align}
    K_\gamma(v) &= \frac{\beta}{2^7}\sum_{i,j,k}\sum_{\alpha,\alpha'}\epsilon(i,j,k)\Big[R_I^{\phantom{I}K}\big(\square^v_{ij}\big)\mathcal{E}_i^I(v+\alpha'\mu_i) + S_J^{\phantom{J}K}\big(\square^v_{ij}\big)\mathcal{E}_j^J(v+\mu_i+\alpha'\mu_j)\Big] \nonumber \\
    &\qquad\qquad\qquad\qquad \times\mathcal{T}^k_{KL}(v)\mathcal{E}_k^L(v+\alpha\mu_k),
\end{align}
where, for notational convenience, we have defined
\begin{equation}
    \begin{gathered}
        R_I^{\phantom{I}J}\left(\square^v_{ij}\right) \coloneqq -4\tr\left[\tau^K h\left(\square^v_{ij}\right) \tau^J\right]\mathcal{T}^i_{KI}(v), \\
        S_I^{\phantom{I}J}\left(\square^v_{ij}\right) \coloneqq -4\tr\left[\mathrm{Ad}_{h_i(v)}\left(\tau^K\right)h\left(\square^v_{ij}\right)\tau^J\right]\mathcal{T}^j_{KI}(v).
    \end{gathered}
\end{equation}

At this point, we need to evaluate Poisson brackets of the form
\begin{align}
    \left\{h_\ell(v),K[\gamma]\right\} = s_e\sum_{v'\in\gamma^v}\left\{h_\ell(v),K_\gamma(v')\right\}.
\end{align}
For any $v,v'\in\gamma^v$ and $m\in\mathcal{I}$, we have 
\begin{align}
    \left\{h_\ell(v),\mathcal{E}_m^M(v')\right\} &= \frac{\kappa\beta}{8}\left(\delta_{v,v'} + \delta_{v+\mu_\ell,v'}\right)\tau^K h_\ell(v)\,\mathcal{T}^\ell_{KL}(v)\mathcal{Z}_{\ell m}^{LM}(v'),
\end{align}
with
\begin{align}
    \mathcal{Z}_{\ell m}^{LM}(v')\coloneqq \frac{1}{V_\gamma(v')}\left[s_e\sum_k\epsilon(k,\ell,m)\epsilon_K^{\phantom{K}LM}P_k^K(v') - \mathcal{E}_\ell^L(v')\mathcal{E}_m^M(v')\right].
\end{align}
It follows from a long computation that 
\begin{align}
    \left\{h_\ell(v),K[\gamma]\right\} = s_e\frac{\kappa\beta^2}{2^{10}}\,\tau^K h_\ell(v)\,\mathcal{T}^\ell_{KL}(v)\sum_{\alpha''} k_\ell^L(v+\alpha''\mu_\ell),
\end{align}
where 
\begin{align}
    k_\ell^L(v) &\coloneqq \sum_{i,j,k}\sum_{\alpha,\alpha'}\epsilon(i,j,k)\bigg(\mathcal{Z}_{\ell i}^{LI}(v)\Big[R_I^{\phantom{I}J}\big(\square^{v-\alpha'\mu_i}_{ij}\big)\mathcal{T}^k_{JK}(v-\alpha'\mu_i)\mathcal{E}_k^K(v+\alpha\mu_k - \alpha'\mu_i) \nonumber \\
    &\qquad - S_I^{\phantom{I}J}\big(\square^{v-\mu_j - \alpha'\mu_i}_{ji}\big)\mathcal{T}^k_{JK}(v-\mu_j-\alpha'\mu_i)\mathcal{E}_k^K(v-\mu_j + \alpha\mu_k - \alpha'\mu_i)\Big] \nonumber \\
    &\qquad\qquad + \mathcal{Z}_{\ell k}^{LK}(v)\mathcal{T}^k_{JK}(v)\Big[\mathcal{E}_i^I(v-\alpha\mu_k + \alpha'\mu_i) R_I^{\phantom{I}J}\big(\square^{v-\alpha\mu_k}_{ij}\big) \nonumber \\
    &\qquad\qquad\qquad - \mathcal{E}_i^I(v+\mu_j-\alpha\mu_k +\alpha'\mu_i) S_I^{\phantom{I}J}\big(\square^{v - \alpha\mu_k}_{ji}\big)\Big]\bigg).
\end{align}
Finally, we define $\mathfrak{su}(2)$-components of the \textit{regularized extrinsic curvature} by 
\begin{align}
    \mathcal{K}_\ell^L(v) &\coloneqq -s_e\frac{4}{\kappa\beta^3}\,\tr\left[\tau^L\left\{h_\ell(v),K[\gamma]\right\} h_\ell^\dagger(v)\right] \nonumber \\
    &= \frac{1}{2^9\beta}\,\delta^{IL}\,\mathcal{T}^\ell_{IJ}(v)\sum_\alpha k_\ell^J(v+\alpha\mu_\ell),
\end{align}
which can be shown to satisfy $\lim_{n\to\infty}\frac{1}{\varepsilon_\ell}\mathcal{K}_\ell^L(v) = \sgn{\ell} K_{|\ell|}^L(v)$, as required.

\subsection{Total Scalar Constraint}

Putting everything together, we see that \eqref{CL_Gamma} can be written in the simplified form
\begin{align}
    \label{CL_Gamma_poly}
    C_{\gamma}^L(v) = -\left(\frac{1+\beta^2}{16\kappa}\right)\sum_{i,j,k}\sum_\alpha \epsilon(i,j,k)\,\epsilon_{IJ}^{\phantom{IJ}K}\,\mathcal{K}_i^I(v)\,\mathcal{K}_j^J(v)\,\mathcal{T}^k_{KL}(v)\,\mathcal{E}_k^L(v+\alpha\mu_k).
\end{align}
Combining this with \eqref{CE_Gamma_poly}, we arrive at our final result for the full discretized scalar constraint:
\begin{align}
    \label{C_Gamma_poly}
    C_\gamma[N] &= \frac{1}{16\kappa}\sum_{v\in\gamma^v}N(v)\sum_{i,j,k}\sum_\alpha\epsilon(i,j,k)\Big[\mathcal{F}^K\big(\square^v_{ij}\big) - (1+\beta^2)\,\epsilon_{IJ}^{\phantom{IJ}K}\,\mathcal{K}_i^I(v)\,\mathcal{K}_j^J(v)\Big] \nonumber \\
    &\qquad\qquad\qquad \times\mathcal{T}^k_{KL}(v)\,\mathcal{E}_k^L(v+\alpha\mu_k),
\end{align}
for some lapse function $N$.

\subsubsection{Restriction to Symmetric Subspaces}

As discussed throughout this paper, one is ultimately interested in deriving an effective version of the Thiemann constraint applicable in the presence of certain symmetry requirements. Let us assume that the solutions of interest belong to the invariant subspace $\overline{\M}_\gamma \subset \M_\gamma$ determined by some symmetry group $\Phi_\gamma^\Psi < \Symp(\M_\gamma,\omega_\gamma)$. Then the effective scalar constraint of interest can be obtained by direct restriction of $C_\gamma[N]$ to this subspace ---that is, we define 
\begin{align}
    \overline{C}_\gamma[N] \coloneqq C_\gamma[N]\big|_{\overline{\M}_\gamma}, 
\end{align}
where the lapse function $N$ is assumed to obey the necessary requirements in order for $C_\gamma[N]$ to be $\Phi_\gamma^\Psi$-invariant.

The restricted scalar constraint $\overline{C}_\gamma[N]$ will depend entirely on the chosen parameterization for $\overline{\M}_\gamma$ (e.g., for the phase space $\overline{\M}_{\gamma,\text{sym}}$ of Sec.~\ref{s4:sph_symm_GOG}, this would be a function of $a_\varepsilon,b_\varepsilon,p_a^\varepsilon,p_b^\varepsilon$). If the corresponding holonomies and fluxes are denoted by $\overline{h}_i(v)$ and $\overline{P}_i(v)$, then restricting \eqref{C_Gamma_poly} to $\overline{\M}_\gamma$ results in an expression of the form 
\begin{align}\label{C_Gamma_Bar_poly}
    \overline{C}_\gamma[N] &= \frac{1}{16\kappa}\sum_{v\in\gamma^v}N(v)\sum_{i,j,k}\sum_\alpha\epsilon(i,j,k)\Big[\overline{\mathcal{F}}^K\big(\square^v_{ij}\big) - (1+\beta^2)\,\epsilon_{IJ}^{\phantom{IJ}K}\,\overline{\mathcal{K}}_i^I(v)\,\overline{\mathcal{K}}_j^J(v)\Big] \nonumber \\
    &\qquad\qquad\qquad \times\overline{\mathcal{T}}^k_{KL}(v)\,\overline{\mathcal{E}}_k^L(v+\alpha\mu_k),
\end{align}
in which all quantities are exactly as defined above, but evaluated using $\overline{h}_i(v)$ and $\overline{P}_i(v)$. The advantage of employing the Thiemann constraint in the form \eqref{C_Gamma_poly} rather than \eqref{C_Thiemann} is thus two-fold: (i) it does not require Poisson brackets to be propagated through the holonomies and fluxes in order to employ the algebra on $\overline{\M}_\gamma$, and (ii) gauge-covariant fluxes can be incorporated \textit{a posteriori} by simply changing the explicit expressions used for $\overline{P}_i(v)$ when evaluating \eqref{C_Gamma_Bar_poly}. 

\section{Symplectic Structure for Spherically Symmetric Gravity on a Graph}\label{app:restricted_algebra}

In this appendix, we shall explicitly derive the full Poisson algebra on $\overline{\M}_{\gamma,\text{sym}}$ given in \eqref{a_pa_bracket}--\eqref{a_b_bracket}. In particular, we will first use the symplectic form \eqref{omega_bar_sym} to define Hamiltonian vector fields associated with each function $a_\varepsilon,b_\varepsilon,p_a^\varepsilon,p_b^\varepsilon$ parameterizing the phase space. Once these vector fields have been identified, Poisson brackets can be easily extracted via 
\begin{align}\label{PB_def}
    \left\{f,g\right\} = \iota_{X_g}\iota_{X_f}\overline{\omega}_{\gamma,\text{sym}} = \mathcal{L}_{X_g}f,
\end{align}
or, equivalently, $\{f,g\} = -\mathcal{L}_{X_f}g$.

We begin by rewriting \eqref{omega_bar_sym} in the more compact form
\begin{align} \label{omega_app}
    \overline{\omega}_{\gamma,\text{sym}} = \sum_{v_r}\bigg[F^a_{p_a}\,\delta a_\varepsilon\wedge \delta p_a^\varepsilon + F^b_{p_b}\,\delta b_\varepsilon\wedge \delta p_b^\varepsilon + \Delta_r \big(\delta p_a^\varepsilon\big)\wedge\left(F^{p_a^\varepsilon}_{b_\varepsilon}\,\delta b_\varepsilon + F^{p_a}_{p_b}\,\delta p_b^\varepsilon\right)\bigg],
\end{align}
where we have have defined
\begin{equation}\label{F_defs}
    \begin{gathered}
        F^a_{p_a} \coloneqq \frac{8\pi}{\kappa\beta}\,\varepsilon_r\cos(\varepsilon_\theta/2)=\text{const}, \qquad F^b_{p_b} \coloneqq \frac{4\pi}{\kappa\beta}\,\varepsilon_r\left(\mathcal{F}_\Omega + \mathcal{H}_\Omega\Pi_\varepsilon^2\right), \\
        F^{p_a}_b \coloneqq \frac{2\pi}{\kappa\beta}\,\varepsilon_r\,\mathcal{H}_\Omega\Pi_\varepsilon, \qquad\qquad\qquad F^{p_a}_{p_b} \coloneqq \frac{2\pi}{\kappa\beta}\,\varepsilon_r\,\mathcal{G}_\Omega\Pi_\varepsilon\frac{b_\varepsilon}{p_b^\varepsilon},
    \end{gathered}
\end{equation}
with $\Pi_\varepsilon$, $\mathcal{F}_\Omega$, $\mathcal{G}_\Omega$, and $\mathcal{H}_\Omega$ originally introduced in \eqref{a_b_Pi_epsilon}, \eqref{F_Omega}, \eqref{G_Omega}, and \eqref{H_Omega}, respectively.
For any function $f\in C^\infty(\overline{\M}_{\gamma,\text{sym}})$, we define an associated vector field $X_{f(v_r')}\in\mathfrak{X}(\overline{\M}_{\gamma,\text{sym}})$ which admits an expansion of the form
\begin{align}
    \label{X_f_general}
    X_{f(v_r')} = \sum_{v_r}\left[X_{f(v_r')}^{a_\varepsilon(v_r)}\,\delta_{a_\varepsilon(v_r)} + X_{f(v_r')}^{b_\varepsilon(v_r)}\,\delta_{b_\varepsilon(v_r)} + X_{f(v_r')}^{p_a^\varepsilon(v_r)}\,\delta_{p_a^\varepsilon(v_r)} + X_{f(v_r')}^{p_b^\varepsilon(v_r)}\,\delta_{p_b^\varepsilon(v_r)}\right].
\end{align}
In order for $X_{f(v_r')}$ to be \textit{the} Hamiltonian vector field associated with $f(v_r')$, it is required to satisfy
\begin{align}\label{Hamiltonian_VF}
    \iota_{X_{f(v_r')}}\overline{\omega}_{\gamma,\text{sym}} = \delta f(v_r'),
\end{align}
where 
\begin{align}
    \label{func_differential}
    \delta f(v_r') = \sum_{v_r}\Bigg[\frac{\delta f(v_r')}{\delta a_\varepsilon(v_r)}\,\delta a_\varepsilon(v_r) + \frac{\delta f(v_r')}{\delta b_\varepsilon(v_r)}\,\delta b_\varepsilon(v_r) + \frac{\delta f(v_r')}{\delta p_a^\varepsilon(v_r)}\,\delta p_a^\varepsilon(v_r) + \frac{\delta f(v_r')}{\delta p_b^\varepsilon(v_r)}\,\delta p_b^\varepsilon(v_r)\Bigg]
\end{align}
is the functional differential of $f(v_r)$ on $\overline{\M}_{\gamma,\text{sym}}$. As we shall see, the condition \eqref{Hamiltonian_VF} is sufficient to fully determine the components of $X_{f(v_r')}$.

Computing the contraction of the symplectic form \eqref{omega_app} with the vector field \eqref{X_f_general}, we obtain
\begin{align}\label{X_f_contraction}
    \iota_{X_{f(v_r')}}\overline{\omega}_{\gamma,\text{sym}} &= \sum_{v_r}\Bigg\{\bigg(F^a_{p_a} X_{f(v_r')}^{a_\varepsilon(v_r)} + \Delta_r\left[F^{p_a}_b(v_r) X_{f(v_r')}^{b_\varepsilon(v_r)} + F^{p_a}_{p_b}(v_r) X_{f(v_r')}^{p_b^\varepsilon(v_r)}\right]\bigg)\delta p_a^\varepsilon(v_r) \nonumber \\
    &\qquad + \left[F^b_{p_b}(v_r) X_{f(v_r')}^{b_\varepsilon(v_r)} + F^{p_a}_{p_b}(v_r)\Delta_r X_{f(v_r')}^{p_a^\varepsilon(v_r)}\right]\delta p_b^\varepsilon(v_r) - F^a_{p_a} X_{f(v_r')}^{p_a^\varepsilon(v_r)}\,\delta a_\varepsilon(v_r) \nonumber \\
    &\qquad\qquad - \left[F^b_{p_b}(v_r) X_{f(v_r')}^{p_b^\varepsilon(v_r)} - F^{p_a}_b(v_r)\Delta_r X_{f(v_r')}^{p_a^\varepsilon(v_r)}\right]\delta b_\varepsilon(v_r)\Bigg\},
\end{align}
where we have used the fact that the finite difference operator satisfies\footnote{Note that this identity holds if contributions from boundary terms are treated as negligible.}
\begin{align}
    \sum_{v_r}F(v_r)\,\Delta_r\delta p_a^\varepsilon(v_r) = -\sum_{v_r}\Delta_r F(v_r)\,\delta p_a^\varepsilon(v_r),
\end{align}
for any function $F(v_r)$. By comparison with \eqref{func_differential}, it can then be shown that the condition \eqref{Hamiltonian_VF} holds \textit{iff} the components of $X_{f(v_r')}$ are taken to be
\begin{align}
    &X_{f(v_r')}^{a_\varepsilon(v_r)} = \frac{1}{F^a_{p_a}}\Bigg(\frac{\delta f(v_r')}{\delta p_a^\varepsilon(v_r)} + \Delta_r\Bigg[\frac{F^{p_a}_{p_b}(v_r)}{F^b_{p_b}(v_r)}\frac{\delta f(v_r')}{\delta b_\varepsilon(v_r)} - \frac{F^{p_a}_b(v_r)}{F^b_{p_b}(v_r)}\frac{\delta f(v_r')}{\delta p_b^\varepsilon(v_r)}\Bigg]\Bigg),\label{X_f_a} \\
    &X_{f(v_r')}^{b_\varepsilon(v_r)} = \frac{1}{F^b_{p_b}(v_r)}\Bigg(\frac{\delta f(v_r')}{\delta p_b^\varepsilon(v_r)} + \frac{1}{F^a_{p_a}}F^{p_a}_{p_b}(v_r)\,\Delta_r\frac{\delta f(v_r')}{\delta a_\varepsilon(v_r)}\Bigg), \label{X_f_b}\\
    &X_{f(v_r')}^{p_a^\varepsilon(v_r)} = -\frac{1}{F^a_{p_a}}\frac{\delta f(v_r')}{\delta a_\varepsilon(v_r)}, \label{X_f_pa}\\
    &X_{f(v_r')}^{p_b^\varepsilon(v_r)} = -\frac{1}{F^b_{p_b}(v_r)}\Bigg(\frac{\delta f(v_r')}{\delta b_\varepsilon(v_r)} + \frac{1}{F^a_{p_a}}F^{p_a}_b(v_r)\,\Delta_r\frac{\delta f(v_r')}{\delta a_\varepsilon(v_r)}\Bigg). \label{X_f_pb}
\end{align}
Let us also notice that \eqref{X_f_contraction} can be used to establish the non-degeneracy of $\overline{\omega}_{\gamma,\text{sym}}$: if we assume $F^b_{p_b}(v_r) \neq 0$, then for any $X\in\mathfrak{X}(\overline{\M}_{\gamma,\text{sym}})$, we have 
\begin{equation}
    \iota_X\overline{\omega}_{\gamma,\text{sym}} = 0 \iff X = 0.
\end{equation}
This allows us to conclude that $(\overline{\M}_{\gamma,\text{sym}},\overline{\omega}_{\gamma,\text{sym}})$ is indeed a genuine symplectic manifold, as required for Theorem~\ref{thm} to apply.

Applying \eqref{X_f_a}--\eqref{X_f_pb} to the functions parameterizing $\overline{\M}_{\gamma,\text{sym}}$, we find that the associated Hamiltonian vector fields are given by
\begin{align}
    &X_{a_\varepsilon(v_r')} = -\frac{1}{F^a_{p_a}}\Bigg(\delta_{p_a^\varepsilon(v_r')} + \frac{1}{2\,\varepsilon_r}\left[\frac{F^{p_a}_{p_b}(v_r'+\varepsilon_r)}{F^b_{p_b}(v_r'+\varepsilon_r)}\,\delta_{b_\varepsilon(v_r'+\varepsilon_r)} - \frac{F^{p_a}_{p_b}(v_r'-\varepsilon_r)}{F^b_{p_b}(v_r'-\varepsilon_r)}\,\delta_{b_\varepsilon(v_r'-\varepsilon_r)}\right], \nonumber \\
    &\qquad\qquad\qquad - \frac{1}{2\,\varepsilon_r}\left[\frac{F^{p_a}_b(v_r'+\varepsilon_r)}{F^b_{p_b}(v_r'+\varepsilon_r)}\,\delta_{p_b^\varepsilon(v_r'+\varepsilon_r)} - \frac{F^{p_a}_b(v_r'-\varepsilon_r)}{F^b_{p_b}(v_r'-\varepsilon_r)}\,\delta_{p_b^\varepsilon(v_r'-\varepsilon_r)}\right]\Bigg), \\
    &X_{b_\varepsilon(v_r')} = -\frac{1}{F^b_{p_b}(v_r')}\Bigg(\delta_{p_b^\varepsilon(v_r')} + \frac{F^{p_a}_{p_b}(v_r')}{2\,\varepsilon_r F^a_{p_a}}\Big[\delta_{a_\varepsilon(v_r'+\varepsilon_r)} - \delta_{a_\varepsilon(v_r'-\varepsilon_r)}\Big]\Bigg),\\
    &X_{p_a^\varepsilon(v_r')} = \frac{1}{F^a_{p_a}}\delta_{a_\varepsilon(v_r')}, \\
    &X_{p_b^\varepsilon(v_r')} = \frac{1}{F^b_{p_b}(v_r')}\Bigg(\delta_{b_\varepsilon(v_r')} + \frac{F^{p_a}_b(v_r')}{2\,\varepsilon_r F^a_{p_a}}\Big[\delta_{a_\varepsilon(v_r' + \varepsilon_r)} - \delta_{a_\varepsilon(v_r'-\varepsilon_r)}\Big]\Bigg).
\end{align}
As mentioned earlier, the formal Poisson algebra on $\overline{\M}_{\gamma,\text{sym}}$ can now be computed by means of differentiation along the flow generated by these fields:
\begin{align}
    \big\{a_\varepsilon(v_r),p_a^\varepsilon(v_r')\big\} &= \mathcal{L}_{X_{p_a^\varepsilon(v_r')}}a_\varepsilon(v_r) = \frac{1}{F^a_{p_a}}\,\delta_{v_r,v_r'}, \\
    \big\{b_\varepsilon(v_r),p_b^\varepsilon(v_r')\big\} &= \mathcal{L}_{X_{p_b^\varepsilon(v_r')}}b_\varepsilon(v_r) = \frac{1}{F^b_{p_b}(v_r')}\,\delta_{v_r,v_r'}, \\
    \big\{a_\varepsilon(v_r),p_b^\varepsilon(v_r')\big\} &= \mathcal{L}_{X_{p_b^\varepsilon(v_r')}}a_\varepsilon(v_r) \nonumber \\
    &= \frac{1}{2\,\varepsilon_r F^a_{p_a}}\frac{F^{p_a}_b(v_r')}{F^b_{p_b}(v_r')}\big(\delta_{v_r,v_r'+\varepsilon_r} - \delta_{v_r,v_r'-\varepsilon_r}\big), \\
    \big\{a_\varepsilon(v_r),b_\varepsilon(v_r')\big\} &= \mathcal{L}_{X_{b_\varepsilon(v_r')}}a_\varepsilon(v_r) \nonumber \\
    &= \frac{1}{2\,\varepsilon_r F^a_{p_a}}\frac{F^{p_a}_{p_b}(v_r')}{F^b_{p_b}(v_r')}\big(\delta_{v_r,v_r'-\varepsilon_r} - \delta_{v_r,v_r'+\varepsilon_r}\big),
\end{align}
with all other combinations yielding vanishing Poisson brackets. Finally, one simply restores the definitions \eqref{F_defs} in order to obtain the algebra presented in \eqref{a_pa_bracket}--\eqref{a_b_bracket}.


\bibliographystyle{jhep}





\bibliography{mainbib}{}

\end{document}